\documentclass[a4paper,11pt]{report}

\usepackage{epsfig}

\usepackage[english]{babel}

\usepackage{amsmath}

\usepackage{graphics,amsmath,amsthm,amsfonts,amscd,latexsym,enumerate,multirow,epsfig,subfigure,yfonts,eufrak}
\usepackage{USCthesis2004}

          \renewcommand{\bibname}{BIBLIOGRAPHY}
        \renewcommand\contentsname{\centerline{\bf Table of Contents}}

\newtheorem{Def}{Definition}
\newtheorem{theorem}{Theorem}

\newtheorem{Cor}{Corollary}

\newtheorem{Prop}{Proposition}

\newtheorem{lemma}{Lemma}

\renewcommand{\thechapter}{\arabic{chapter}}

\title{ENTANGLEMENT-ASSISTED CODING THEORY.}

\author{Min-Hsiu Hsieh}

\date{August 2008}

\renewcommand{\copyrightyear}{2008}

\renewcommand{\copyrightname}{Min-Hsiu Hsieh}

\def\cC{{\cal C}}
\def\cD{{\cal D}}
\def\cE{{\cal E}}

\def\cG{{\cal G}}
\def\cH{{\cal H}}

\def\cL{{\cal L}}

\def\cN{{\cal N}}

\def\cS{{\cal S}}

\def\cU{{\cal U}}
\def\cV{{\cal V}}

\def\cZ{{\cal Z}}

\newcommand{\bA}{{\mathbf A}}

\newcommand{\bE}{{\mathbf E}}

\newcommand{\bH}{{\mathbf H}}

\newcommand{\bS}{{\mathbf S}}
\newcommand{\bT}{{\mathbf T}}

\newcommand{\ba}{{\mathbf a}}
\newcommand{\bb}{{\mathbf b}}
\newcommand{\bc}{{\mathbf c}}
\newcommand{\bd}{{\mathbf d}}
\newcommand{\be}{{\mathbf e}}

\newcommand{\bg}{{\mathbf g}}
\newcommand{\bh}{{\mathbf h}}

\newcommand{\bn}{{\mathbf n}}

\newcommand{\br}{{\mathbf r}}
\newcommand{\bs}{{\mathbf s}}

\newcommand{\bu}{{\mathbf u}}
\newcommand{\bv}{{\mathbf v}}
\newcommand{\bw}{{\mathbf w}}
\newcommand{\bx}{{\mathbf x}}
\newcommand{\by}{{\mathbf y}}
\newcommand{\bz}{{\mathbf z}}

\def\bbC{\mathbb{C}}

\def\bbF{\mathbb{F}}

\def\bbZ{\mathbb{Z}}

\def\fM{\mathfrak{M}}
\def\fN{\mathfrak{N}}

\def\yM{\textswab{M}}
\def\yN{\textswab{N}}

\DeclareMathOperator{\id}{id}

 \DeclareMathOperator{\tr}{Tr}
 \DeclareMathOperator{\wt}{wt}

\newcommand{\bra}[1]{\langle #1 |}
\newcommand{\ket}[1]{| #1 \rangle}
\newcommand{\proj}[1]{| #1 \rangle\!\langle #1 |}

\def\T{{{T}}}
\def\b0{\mathbf{0}}
\newcommand{\inner}[2]{\langle{#1},{#2}\rangle}
\def\Csp{C_{\text{sp}}}
\def\Hsp{H_{\text{sp}}}
\def\Cea{\cC^{\text{EA}}}
\def\Cop{\cC^{\text{OP}}}
\def\Ceao{\cC^{\text{EAO}}}

\newcommand{\iso}[1] {{\rm iso} ( #1)}
\newcommand{\symp}[1] {{\rm symp} ( #1)}
\newcommand{\spann}[1] {{\rm span} \{ #1 \}}

\makeatletter
\let\l@figureOLD \l@figure
\renewcommand{\l@figure}{\vspace{\baselineskip}\l@figureOLD}
\makeatother

\makeatletter
\let\l@tableOLD \l@table
\renewcommand{\l@table}{\vspace{\baselineskip}\l@tableOLD}
\makeatother

\begin{document}

\renewcommand\contentsname{Table of Contents}

\maketitle

\begin{preface}

\begin{singlespace}

\newcommand*\oldhss{}
\let\oldhss\hss
\renewcommand*\hss{\oldhss\normalfont}
\tableofcontents
\let\hss\oldhss

\newpage

\addcontentsline{toc}{chapter}{List of Figures}

\addtocontents{lof}{\vspace*{-\baselineskip}}

\listoffigures
\newpage

\addcontentsline{toc}{chapter}{List of Tables}

\addtocontents{lot}{\vspace*{-\baselineskip}}

\listoftables
\newpage

\end{singlespace}

\end{preface}

\chapter*{Chapter 1:  \hspace{1pt} Overview}%
\addcontentsline{toc}{chapter}{Chapter 1:\hspace{0.15cm} Overview}%

The theory of \emph{quantum mechanics}, founded in the early 1920s,
ended the turmoil caused by the \emph{classical physics} that
predicted various absurd results such as electrons spiraling
inexorably into the atom nucleus. Though the mathematical framework
of quantum mechanics is simple, even geniuses like Albert Einstein
found it counter-intuitive. Generations of physicists since put a
lot of effort to sharpen our intuition about quantum mechanics, and
make it more transparent to normal human minds. Several fundamental
results discovered later on, such as the famous \emph{no-cloning
theorem} \cite{WZ82} that denies the possibility of using quantum
effects to signal faster than light, help us better understand
quantum mechanics.

Research on quantum mechanics evolved into a interdisciplinary
science due to several successful applications of quantum effects on
classical computation and communication problems in 1990s. Among
them, Shor proposed a quantum algorithm for the enormously important
problem \cite{Shor97} --- the problem of finding the prime factors
of an integer --- showing exponential speed-up over the best known
classical algorithm. This result not only attracted broad interest
because this problem is believed to have no efficient solution on
classical computers, but also provided strong evidence that quantum
computers are more powerful than classical computers.

However, the power of quantum computation and communication over
classical computation and communication comes from implementing
\emph{entangled} quantum states that are easily spoilt by their
vulnerability to errors. Namely, the destructive interference of the
omnipresent environment leads to an exponential loss of the
probability that the computation runs in the desired way. Up to that
point, there was a widespread belief that decoherence
--- environmental noise --- would doom any chance of building large scale
quantum computers or quantum communication protocols. The equally
widespread belief that any analogue of classical error correction
was impossible in quantum mechanics due to the famous no cloning
theorem produced an even stronger pessimistic atmosphere in
developing quantum computers.

Luckily, the pessimistic atmosphere did not last long. One of the
most important discoveries in quantum information science, the
existence of quantum error-correcting codes (QECCs), defied those
expectations. The first quantum error-correcting code, considered as
a quantum analogue of the classical repetition code, was proposed by
Shor in 1995 \cite{Shor95}. The theory of quantum error correction
quickly became a popular research topic. The quantum
error-correcting conditions were proved independently by Bennett,
DiVincenzo, Smolin and Wootters \cite{BDSW96}, and by Knill and
Laflamme \cite{KL97}. The best quantum code that encodes one-qubit
information, the five-qubit code, was discovered by Laflamme,
Miquel, Paz, and Zurek \cite{LMPZ96}, and independently by
\cite{BDSW96}.

The development of quantum error-correcting theory then became
systematic. A construction of Calderbank, Shor, and Steane
\cite{CS96,Ste96} showed that it was possible to construct quantum
codes from classical linear codes --- the CSS codes --- thereby
drawing on the well-studied theory of classical error correction.
Furthermore, Gottesman invented the stabilizer formalism
\cite{Got96}, and used it to define stabilizer codes. In this view,
quantum error-correcting codes are simultaneous eigenspaces of a
group of commuting operators, the stabilizer. Independently,
Calderbank, Rain, Shor, and Sloane \cite{CRSS97} proposed a similar
idea to define quantum codes based on orthogonal geometry in
classical coding theory. This result connected quantum codes to
classical quaternary codes \cite{CRSS98}. The theory of quantum
error correction developed so far is called \emph{standard} quantum
error correction.

Important as these results were, they fell short of doing everything
that one might wish.  The connection between classical codes and
quantum codes was not universal.  Rather, only classical codes which
satisfied a \emph{self-orthogonality constraint} could be used to
construct quantum codes.  While this constraint was not too
difficult to satisfy for relatively small codes, it is a substantial
barrier to the use of highly efficient modern codes, such as Turbo
and Low-Density Parity Check (LDPC) codes, in quantum information
theory.  These codes are capable of achieving the classical
capacity; but the difficulty of constructing self-orthogonal
versions of them has made progress toward finding quantum versions
very slow.

These problems can be overcome with pre-existing entanglement.
Entanglement plays a central role in almost every quantum
computation and communication task. It enables the teleportation of
quantum states without physically sending quantum
systems\cite{BBCJPW93}; it doubles the capacity of quantum channels
for sending classical information\cite{BW92}; it is known to be
necessary for the power of quantum computation\cite{BCD02,JL03}.
Furthermore, descriptions in quantum information theory are often
simplified by the assumption that pre-existing entanglement is
available.

In the thesis, we show how shared entanglement provides a simpler
and more fundamental theory of quantum error correction, and at the
same time greatly generalize the existing theory of quantum error
correction. If the CSS construction for quantum codes is applied to
a classical code which is not self-orthogonal, the resulting
``stabilizer" group is not commuting, and thus has no code space. We
are able to avoid this problem by making use of pre-existing
entanglement. This noncommuting stabilizer group can be embedded in
a larger space, which makes the group commute, and allows a code
space to be defined.  Moreover, this construction can be applied to
\emph{any} classical quaternary code, not just self-orthogonal ones.
The existing theory of quantum error-correcting codes thus becomes a
special case of our theory: self-orthogonal classical codes give
rise to standard quantum codes, while non-self-orthogonal classical
codes give rise to entanglement-assisted codes.

Besides the entanglement-assisted formalism \cite{BDH06,BDH06IEEE}
we proposed in this thesis, there has been one other major
breakthrough in quantum error correction theory: the discovery of
operator quantum error-correcting codes (OQECCs)
\cite{AKS06OQECC,Bacon05a,bacon-2006,KS06a,KLP05,KS06,NP05,DP05}, or
subsystem codes. Instead of encoding quantum information into a
subspace, OQECCs encode quantum information into a subsystem of the
subspace. These provide a general theory which combines passive
error-avoiding schemes, such as decoherence-free subspaces
\cite{ZR97,LCW98} and noiseless subsystems
\cite{KLV00,KBLW01,Zan00,ZL03}, with conventional (active) quantum
error correction. OQECCs do not lead to new codes, but instead
provide a new kind of decoding procedure: it is not necessary to
actively correct all errors, but rather only to perform correction
modulo the subsystem structure. One potential benefit of the new
decoding procedure is to improve the threshold of fault-tolerant
quantum computation \cite{Bacon05a}.

The other major contribution of this thesis is the development of
the unifying formalism that unifies these two extensions of standard
QECCs: the operator quantum error-correcting codes (OQECCs), and the
entanglement-assisted quantum error-correcting codes (EAQECCs).
Furthermore, our formalism retains the advantages of both
entanglement-assisted and operator quantum error correction. On one
hand, OQECCs provide a general theory which combines passive
error-avoiding schemes with standard quantum error correction. On
the other hand, EAQECCs provide a general theory which links any
classical quaternary code, not just self-orthogonal ones, to a
quantum code. In addition to presenting our formal theory, we have
given several examples of code construction. These examples
demonstrate that our formalism can be used to construct quantum
codes tailored to the needs of particular applications.

Because classical LDPC codes have such high performance ---
approaching the channel capacity in the limit of large block size
--- there has been considerable interest in finding quantum versions
of these codes. However, quantum low-density parity-check codes
\cite{HI07QLDPC,MMM04QLDPC,COT05QLDPC,PC08QLDPC} are far less
studied than their classical counterparts. The main obstacle comes
from the \emph{dual-containing} constraint of the classical codes
that are used to construct the corresponding quantum codes. The
second obstacle comes from the bad performance of the iterative
decoding algorithm such as the famous sum-product algorithm (SPA).
Though the SPA decoding can be directly used to decode the quantum
errors, its performance is severely limited by the many 4-cycles,
which are usually the by-product of the dual-containing property, in
the standard quantum LDPC codes \cite{MMM04QLDPC}.

In the last part of the thesis, we will show that, with the
entanglement-assisted formalism \cite{BDH06,BDH06IEEE}, these two
obstacles of standard quantum LDPC codes can be overcome. By
allowing the use of pre-shared entanglement between senders and
receivers, the dual-containing constraint can be removed.
Constructing quantum LDPC codes from classical LDPC codes becomes
transparent. That is, arbitrary classical quaternary codes can be
used to construct quantum codes via the \emph{generalized CSS
construction} \cite{BDH06}. Furthermore, we can easily construct
quantum LDPC codes from classical LDPC codes with girth at least 6.
We make use of classical quasi-cyclic LDPC codes in our
construction, and show that given similar \emph{net yield} these
quantum LDPC codes perform better than the standard quantum LDPC
codes by numerical simulation.

This thesis is organized as follows. We give various background
materials in chapter 2. In chapter 3, we introduce standard QECCs
using the canonical code method and stabilizer formalism. In chapter
4, we present our first result: the entanglement-assisted formalism.
In chapter 5, we introduce operator quantum error-correcting codes.
In chapter 6, we present our second result: the
entanglement-assisted operator formalism. Finally, we show how to
use the entanglement-assisted formalism to construct quantum LDPC
codes with better performance in chapter 7. Notice that we
explicitly assume a communication scenario throughout the thesis.
That is, noise is modeled as a quantum channel, and it only happens
in the channel. Two parties involved in the information processing
are called sender and receiver, respectively, and their operations
on the quantum states are assumed to be noise-free.

\chapter*{Chapter 2: \hspace{1pt} Background knowledge}%

\addcontentsline{toc}{chapter}{Chapter 2:\hspace{0.15cm} Background
knowledge}%

\section*{2.1 \hspace{2pt} Single qubit Pauli group}

\addcontentsline{toc}{section}{2.1 \hspace{0.15cm} Single qubit
Pauli group}

The set of \emph{Pauli matrices} over a two-dimensional Hilbert
space $\cH_2$ is defined as
$$
I = \left[\begin{array}{cc} 1 & 0 \\ 0 & 1 \end{array}\right], \quad
X = \left[\begin{array}{cc} 0 & 1 \\ 1 & 0 \end{array}\right], \quad
$$
$$
Y = \left[\begin{array}{cc} 0 & -i \\ i & 0
\end{array}\right],  \quad
Z = \left[\begin{array}{cc} 1 & 0 \\ 0 & -1 \end{array}\right].
$$
The Pauli matrices are Hermitian unitary matrices with eigenvalues
belonging to the set $\{1, -1 \}$. The multiplication table of these
matrices is given by:
$$
\begin{array}{|c|cccc|}  \hline
\times  & I & X & Y & Z \\  \hline
I & I & X & Y & Z \\
X & X & I & iZ & -iY \\
Y & Y & -iZ & I & iX \\
Z & Z & iY & -iX & I \\ \hline
\end{array}
$$
Observe that the Pauli matrices either commute or anticommute. Let
$[S] = \{\beta S \mid  \beta \in \bbC, |\beta| = 1\}$ be the
equivalence class of matrices equal to $S$ up to a phase
factor.\footnote{It makes good physical sense to neglect this
overall phase, which has no observable consequence.} Let $\cG$ be
the group generated by the set of Pauli matrices $\{I,X,Y,Z\}$ with
all possible phases, then the set $[\cG] = \{ [I], [X], [Y], [Z] \}$
is readily seen to form a commutative group under the multiplication
operation defined by $[S] [T] = [ST]$. We called $[\cG]$ the Pauli
group.

We are interested in relating the Pauli group to the additive group
$(\bbZ_2)^2 = \{00, 01, 10, 11\}$ of binary words of length $2$
described by the table:
$$
\begin{array}{|c|cccc|} \hline
+  & 00 & 01 & 11 & 10 \\  \hline
00 & 00 & 01 & 11 & 10 \\
01 & 01 & 00 & 10 & 11 \\
11 & 11 & 10 & 00 & 01 \\
10 & 10 & 11 & 01 & 00 \\ \hline
\end{array}
$$
This group is also a two-dimensional vector space over the field
$\bbZ_2$. A bilinear form can be defined over this vector space,
called the \emph{symplectic form} or \emph{symplectic
product}\footnote{Strictly speaking it is not an inner product.}
$\odot: (\bbZ_2)^2 \times (\bbZ_2)^2 \rightarrow \bbZ_2$, given by
the table
$$
\begin{array}{|c|cccc|} \hline
\odot  & 00 & 01 & 11 & 10 \\  \hline
00 & 0 & 0 & 0 & 0 \\
01 & 0 & 0 & 1 & 1 \\
11 & 0 & 1 & 0 & 1 \\
10 & 0 & 1 & 1 & 0 \\ \hline
\end{array}
$$
In what follows we will often write elements of $(\bbZ_2)^2$ as $u =
(z|x)$, with $z,x \in \bbZ_2$. For instance, $01$ becomes $(0|1)$.
For $u = (z|x), v = (z'|x') \in (\bbZ_2)^2$ the symplectic product
is equivalently defined by
$$
u \odot v = z x' + z' x.
$$
Define the map $N: (\bbZ_2)^2  \to  \cG$ by the following table:
$$
\begin{array}{|c|c|} \hline
(\bbZ_2)^2 & \cG\\  \hline
00 & I \\
01 & X \\
11 & Y \\
10 & Z \\ \hline
\end{array}
$$
This map is defined in such a way that $N_{(z|x)}$ and $Z^{z} X^{x}$
are equal up to a phase factor, i.e.
$$[N_{(z|x)}] = [ Z^{z} X^{x} ].$$
We make two key observations
\begin{enumerate}
\item  The map $[N] :  (\bbZ_2)^2  \to  [\cG]$ induced by $N$
is an isomorphism:
$$
[N_u] [N_v] = [N_{u + v}].
$$
\item The commutation relations of the Pauli matrices are captured by the
symplectic product
$$
N_u N_v = (-1)^{u \odot v} N_v N_u.
$$
\end{enumerate}
Both properties are readily verified from the tables.

\section*{2.2 \hspace{2pt} Multi-qubit Pauli group}
\addcontentsline{toc}{section}{2.2 \hspace{0.15cm} Multi-qubit Pauli
group}

Consider an $n$-qubit system corresponding to the tensor product
Hilbert space $\cH_2^{\otimes n}$. Define an $n$-qubit Pauli matrix
$\bS$ to be of the form $\bS = S_1\otimes S_2\otimes \cdots\otimes
S_n$, where $S_j \in \cG$. Let $\cG^n$ be the group of all $4^n$
$n$-qubit Pauli matrices with all possible phases. Define as before
the equivalence class $[\bS] = \{\beta \bS \mid \beta \in \bbC,
|\beta| = 1\}$. Then
$$
[\bS] [\bT] = [S_1 T_1]\otimes[S_2 T_2]\otimes \cdots \otimes[S_n
T_n] = [\bS \bT].
$$
Thus the set $[\cG^n] = \{ [\bS]: \bS \in \cG^n \}$ is a commutative
multiplicative group, and is called the $n$-fold Pauli Group.

Now consider the group/vector space $(\bbZ_2)^{2n}$ of binary
vectors of length $2n$. Its elements may be written as $\bu = (\bz|
\bx)$, $\bz = z_1 \dots z_n \in (\bbZ_2)^n$, $\bx = x_1 \dots x_n
\in (\bbZ_2)^n$. We shall think of $\bu$, $\bz$ and $\bx$ as row
vectors. The symplectic product of $\bu = (\bz|\bx)$ and $\bv =
(\bz'|\bx')$ is given by
$$
\bu \odot \bv^{\T} = \bz \, {\bx'}^{\T} + \bz' \, \bx^{\T}.
$$
The right hand side are binary inner products and the superscript
$\T$ denotes the transpose. This should be thought of as a kind of
matrix multiplication of a row vector and a column vector. We use
$\bu \odot \bv^{\T}$ rather than the more standard $\bu \, \bv^\T$
to emphasize that the symplectic form is used rather than the binary
inner product. Equivalently,
$$
\bu \odot \bv^\T = \sum_i u_i \odot v_i
$$
where $u_i = (z_i|x_i), v_i = (z'_i|x'_i)$ and this sum represents
Boolean addition. Observe that if $ {\bu} \odot {\bv}^\T = 0$, these
two vectors are ``orthogonal'' to each other with respect to the
symplectic inner product.

The map $N: (\bbZ_2)^{2n} \rightarrow \cG^n $ is now defined as
$$
N_\bu = N_{u_1} \otimes \dots \otimes  N_{u_n}.
$$
Writing
$$
X^{\bx} = X^{x_1} \otimes  \dots  \otimes  X^{x_n},
$$
$$
Z^{\bz} = Z^{z_1} \otimes  \dots  \otimes  Z^{z_n},
$$
as in the single qubit case, we have
$$[N_{(\bz|\bx)}] = [ Z^{\bz} X^{\bx} ].$$
The two observations made for the single qubit case also hold:
\begin{enumerate}
\item  The map $[N] :  (\bbZ_2)^{2n}  \to  [\cG^n]$ induced by $N$
is an isomorphism:
\begin{equation}%
\label{eq:primi}%
[N_\bu] [N_\bv] = [N_{\bu + \bv}].
\end{equation}%
Consequently, if $\{ \bu_1, \dots, \bu_m \}$ is a linearly
independent set then the elements of the Pauli group subset $\{
[N_{\bu_1}], \dots, [N_{\bu_m}] \}$ are independent in the sense
that no element can be written as a product of others.

\item The commutation relations of the $n$-qubit Pauli matrices are captured by the
symplectic product
\begin{equation}%
\label{eq:seconda}%
N_\bu N_\bv = (-1)^{\bu \odot \bv^\T} N_\bv N_\bu.
\end{equation}%
\end{enumerate}
\bigskip
We define the \emph{weight} of a Pauli operator $N_\bu$,
$\wt(N_\bu)$, to be the number of single-qubit Pauli matrices in
$N_\bu$ not equal to the identity $I$. Define the \emph{weight} of a
vector $\bu  = (\bz| \bx) \in (\bbZ_2)^{2n}$ by $\wt_{\rm sp}(\bu) =
\wt_2(\bz \vee \bx)$. Here $\vee$ denotes the bitwise logical
``or'', and $\wt_2(\by)$ is the number of non-zero bits in $\by \in
(\bbZ_2)^{n}$. It is easy to verify that
$$\wt(N_{\bu})=\wt_{\rm sp}(\bu).$$

\section*{2.3 \hspace{2pt} Properties of the symplectic form}
\addcontentsline{toc}{section}{2.3 \hspace{0.15cm} Properties of the
symplectic form} \label{sympo}


A subspace $V$ of $(\bbZ_2)^{2n}$ is called \emph{symplectic}
\cite{Silva01} if there is no $\bv \in V$ such that
\begin{equation}%
\bv \odot \bu^\T = 0,  \,\,\, \forall \bu \in V. \label{eq:degen}
\end{equation}%
$(\bbZ_2)^{2n}$ is itself a symplectic subspace. Consider the
standard basis for $(\bbZ_2)^{2n}$, consisting of $\bg_i = (\be_i
|\b0)$ and $\bh_i = (\b0|\be_i)$ for $i = 1, \dots ,n$, where $\be_i
= (0,\dots, 0,1,0, \dots, 0)$ [$1$ in the $i$th position] are the
standard basis vectors of $(\bbZ_2)^n$. Observe that
\begin{eqnarray}
\bg_i \odot \bg_j^\T = 0, &  {\rm for \,\,  all \,\, } i,j \\
\bh_i \odot \bh_j^\T = 0, &  {\rm for \,\, all \,\, } i,j \\
\bg_i \odot \bh_j^\T = 0, &  {\rm for \,\,  all \,\, } i\neq j \\
\bg_i \odot \bh_i^\T = 1, &  {\rm for \,\,  all \,\, } i.
\end{eqnarray}
Thus, the basis vectors come in $n$ \emph{hyperbolic pairs} $(\bg_i,
\bh_i )$ such that only the symplectic product between hyperbolic
partners is nonzero. The matrix $J = [\bg_i \odot \bh_j^\T]$
defining the symplectic product with respect to this basis is given
by
\begin{equation}%
\label{eq:syprma}%
J = \left(\begin{array}{cc}
0_{n \times n} & I_{n \times n} \\
 I_{n \times n} & 0_{n \times n}
\end{array} \right),
\end{equation}%
where $I_{n \times n}$ and $0_{n \times n}$ are the $n \times n$
identity and zero matrices, respectively. A basis for
$(\bbZ_2)^{2n}$ whose symplectic product matrix $J$ is given by
(\ref{eq:syprma}) is called a \emph{symplectic basis}. In the Pauli
picture, the hyperbolic pairs $(\bg_i, \bh_i)$ correspond to
$(Z^{\be_i},X^{\be_i})$, and are sometimes expressed as $(Z_i,X_i)$,
-- the anticommuting $Z$ and $X$ Pauli matrices acting on the $i$th
qubit.

In contrast, a subspace $V$ of $(\bbZ_2)^{2n}$ is called
\emph{isotropic} if
(\ref{eq:degen}) holds for \emph{all} $\bv \in V$.
The largest isotropic subspace of $(\bbZ_2)^{2n}$ is
$n$-dimensional. The span of  the $\bg_i$, $i = 1, \dots ,n$, is an
example of a subspace saturating this bound.

A general subspace of $(\bbZ_2)^{2n}$ is neither symplectic nor
isotropic. The following theorem, stated in \cite{Silva01} and
rediscovered in Pauli language in \cite{FCY04}, says that an
arbitrary subspace $V$ can be decomposed as a direct sum of a
symplectic part and an isotropic part. Here, we prove this theorem
constructively, using a version of the Gram-Schmidt procedure.

\begin{theorem}%
\label{GSsymp}%
Let $V$ be an $m$-dimensional subspace of $(\bbZ_2)^{2n}$. Then
there exists a symplectic basis of $(\bbZ_2)^{2n}$ consisting of
hyperbolic pairs $(\bu_i, \bv_i)$, $i = 1, \dots ,n$, such that
$\{\bu_1, \dots, \bu_{c+ \ell}, \bv_1, \dots, \bv_{c}\}$ is a basis
for $V$, for some $c, \ell \geq 0$ with $2c + \ell = m$.

Equivalently,
$$
V =  \symp V \oplus \iso V
$$
where $\symp V  = \spann { \bu_1, \dots,  \bu_c, \bv_1, \dots,
\bv_{c} }$ is symplectic and $\iso V  = \spann { \bu_{c + 1}, \dots,
\bu_{c + \ell} }$ is isotropic.
\end{theorem}%

\begin{proof}
Pick an arbitrary basis $\{\bw_1, \dots, \bw_m \}$ for $V$ and
extend it to a basis $\{\bw_1, \dots, \bw_{2n} \}$ for
$(\bbZ_2)^{2n}$. The procedure consists of $n$ rounds. In each round
a new hyperbolic pair $(\bu_i, \bv_i)$ is generated; the index $i$
is added to the set $\cU$ (respectively, $\cV$) if  $\bu_i \in V$
($\bv_i \in V$).

Initially set $i =1$, $m' = m$, and $\cU = \cV = \emptyset$.
 The $i$th round reads as follows.
\begin{enumerate}
\item We start with vectors $\bw_1, \dots, \bw_{2(n-i+1)}$,
and $\bu_1, \dots \bu_{i-1},\bv_1, \dots \bv_{i-1}$,
such that
\begin{enumerate}
\item $\bw_1, \dots, \bw_{2(n-i+1)},$
$\bu_1, \dots \bu_{i-1},$$\bv_1, \dots \bv_{i-1}$ is a basis for
$(\bbZ_2)^{2n}$,
\item each of $\bu_1, \dots \bu_{i-1},\bv_1, \dots \bv_{i-1}$
has vanishing symplectic product with each of $\bw_1, \dots,
\bw_{2(n-i+1)}$,
\item $V = {\rm span}\{ \bw_j: 1 \leq j \leq m'  \}
\oplus {\rm span}\{ \bu_j: j \in \cU \} \oplus {\rm span}\{ \bv_j: j
\in \cV \}$.
\end{enumerate}
These conditions are satisfied for $i=1$ where we begin with vectors
$\bw_1, \dots, \bw_{2n}$. In this case, we implicitly assume that
($\bu_0$, $\bv_0$) is the empty set.

\item Define $\bu_i = \bw_1$.  If $m' \geq 1$ then and add $i$ to $\cU$.
Let $j\geq 2$ be the smallest index for which $\bw_1 \odot \bw_j^\T
= 1$. Such a $j$ exists because of (a), (b) and the fact that there
exists a $\bw \in (\bbZ_2)^{2n}$ such that ${\bu_i} \odot {\bw}^\T =
1$.

Set $\bv_i = \bw_j$.
\item If $j \leq m'$:

This means that there is a  hyperbolic  partner of $\bu_i$ in $V$.
 Add $i$ to $\cV$;
swap $\bw_j$ with $\bw_2$;
for $k = 3, \dots, 2(n-i+1)$ 
perform
$$
\bw'_{k-2} := \bw_k - (\bv_i \odot \bw_k^\T) \bu_i -  (\bu_i \odot
\bw_k^\T) \bv_i,
$$
so that
\begin{equation}%
\bw'_{k-2} \odot  \bu_i^\T = \bw'_{k-2} \odot \bv_i^\T = 0 ;
\label{upravo1}%
\end{equation}%
set $m' := m' - 2$.

If $j > m'$:

This means that there is no hyperbolic partner of $\bu_i$ in $V$.
Swap $\bw_j$ with $\bw_{2(n-i+1)}$; for $k = 2, \dots,2(n-i)+1$
perform
$$
\bw'_{k-1} := \bw_k - (\bv_i \odot \bw_k^\T) \bu_i -  (\bu_i \odot
\bw_k^\T) \bv_i,
$$
so that
\begin{equation}%
\bw'_{k-1} \odot  \bu_i^\T = \bw'_{k-1} \odot \bv_i^\T = 0 ;
\label{upravo2}%
\end{equation}%
if $m' \geq 1$ then set $m' := m' - 1$.
\item Let $\bw_k := \bw'_k$ for $1 \leq k \leq 2(n - i)$.
We need to show that the conditions from item 1 are satisfied for
the next round ($i: = i + 1$). Condition (a) holds because
 $\{ \bu_i, \bv_i, \bw'_1, \dots \bw'_{2(n-i)}\}$
are related to the old $\{\bw_1, \dots \bw_{2(n-i+1)}\}$ by an
invertible linear transformation. Condition (b) follows from
(\ref{upravo1}) and (\ref{upravo2}). Regarding condition (c), if $m'
= 0$ then it holds because $\cU$ and $\cV$ did not change from the
previous round. Otherwise, consider the two cases in item 3. If $j
\leq m'$ then $\{ \bu_i, \bv_i, \bw'_1, \dots \bw'_{m'-2}\}$ are
related to the old $\{\bw_1, \dots \bw_{m'}\}$ by an invertible
linear transformation. If $j > m'$ then $\{ \bu_i, \bw'_1, \dots
\bw'_{m'-1}\}$ are related to the old $\{\bw_1, \dots \bw_{m'}\}$ by
an invertible linear transformation (the  $(\bu_i \odot \bw_k^\T)
\bv_i$ terms vanish for $1 \leq k \leq m'$ because there is no
hyperbolic partner of $\bu_i$ in $V$).

\end{enumerate}
At the end of the $i$th round, $0 \leq m' \leq 2(n-i)$. Thus $m' =
0$ after $n$ rounds and hence $V =  {\rm span}\{ \bu_j: j \in \cU \}
\oplus {\rm span}\{ \bv_j: j \in \cV \}$. The theorem follows by
suitably reordering the $(\bu_j,\bv_j)$.

\end{proof}

\noindent {\bf Remark}\,\, It is readily seen that the space
$\iso{V}$ is unique, given $V$. In contrast,  $\symp V$ is not. For
instance, replacing $\bv_1$ by $\bv_1' = \bv_1 +  \bu_{c + 1}$ in
the above definition of $\symp V$ does not change its symplectic
property.

A \emph{symplectomorphism} $\Upsilon:(\bbZ_2)^{2n} \rightarrow
(\bbZ_2)^{2n} $  is a linear isomorphism which preserves the
symplectic form, namely
\begin{equation}%
\Upsilon (\bu) \odot \Upsilon (\bv)^\T = \bu \odot \bv^\T.
\label{smor}%
\end{equation}%
The following theorem relates symplectomorphisms on $(\bbZ_2)^{2n}$
to unitary maps on $\cH_2^{\otimes n}$. It appears, for instance, in
\cite{BFG05}. For completeness, we give an independent proof here.

\begin{theorem} \label{thm2}
For any symplectomorphism  $\Upsilon$ on $(\bbZ_2)^{2n}$ there
exists a unitary map $U_\Upsilon$ on $\cH_2^{\otimes n}$ such that
for all $\bu \in (\bbZ_2)^{2n}$,
$$
[N_{\Upsilon (\bu)}] = [ U_\Upsilon N_{\bu} U_\Upsilon^{-1} ].
$$
\end{theorem}

\noindent {\bf Remark.}\, The unitary map $U_\Upsilon$ may be viewed
as a map on $[\cG_n]$ given by  $[\bS] \mapsto [U_\Upsilon \bS
U_\Upsilon^{-1}]$. The theorem says that the following diagram
commutes
$$
\begin{CD}
(\bbZ_2)^{2n}   @>{{\Upsilon}}>> (\bbZ_2)^{2n}   \\
@V{{[N]}}VV      @VV{{[N]}}V \\
[\cG_n]   @>U_\Upsilon>> [\cG_n]
\end{CD}
$$

\begin{proof}
Consider the standard basis $\bg_i = (\be_i |\b0)$, $\bh_i =
(\b0|\be_i)$. Define the unique (up to a phase factor) state
$\ket{\b0}$ on $\cH_2^{\otimes n}$ to be the simultaneous $+1$
eigenstate of the commuting operators $N_{\bg_j}$, $j = 1, \dots,
n$. Define an orthonormal basis $\{ \ket{\bb}: \bb = b_1 \dots b_n
\in (\bbZ_2)^{n} \}$ for $\cH_2^{\otimes n}$ by
$$
\ket{\bb} = N_{\sum_i b_i \bh_i} \ket{\b0}.
$$
The orthonormality follows from the observation that $\ket{\bb}$ is
a simultaneous eigenstate of $N_{\bg_j}$, $j = 1, \dots, n$ with
respective eigenvalues $(-1)^{b_j}$:
\begin{equation}%
\begin{split}
N_{\bg_j} \ket{\bb} & = N_{\bg_j}  N_{\sum_i b_i \bh_i} \ket{\b0}\\
& = (-1)^{b_j}  N_{\sum_i b_i \bh_i}  N_{\bg_j}  \ket{\b0} \\
& = (-1)^{b_j}  N_{\sum_i b_i \bh_i}  \ket{\b0} \\
& = (-1)^{b_j} \ket{\bb}.
\end{split}
\end{equation}%
The second line is an application of (\ref{eq:seconda}).

Define $\tilde{\bg}_i: =  \Upsilon (\bg_i)$. 
We repeat the above construction for this new basis. Define the
unique (up to a phase factor) state ${\ket{\tilde{\b0}}}$ to be the
simultaneous $+1$ eigenstate of the commuting operators
$N_{\tilde{\bg}_i}$, $i = 1, \dots, n$. Define an orthonormal basis
$\{ {\ket{\tilde{\bb}}} \}$ by
\begin{equation}%
\ket{\tilde{\bb}} = N_{ \sum_i b_i \tilde{\bh}_i} \ket{\tilde{\b0}}.
\label{ono}%
\end{equation}%
Defining $\bu =  \sum_i z_i \bg_i + x_i \bh_i$, $\tilde{\bu} =
\sum_i z_i \tilde{\bg}_i + x_i \tilde{\bh}_i $ and $\bx = x_1 \dots
x_n$, we have
\begin{equation}%
\label{eq:ludak1}%
\begin{split}
N_{\tilde{\bu}} \ket{\tilde{\bb}}
& =  N_{\tilde{\bu}} N_{ \sum_i  b_i \tilde{\bh}_i} {\ket{\tilde{\b0}}} \\
& = (-1)^{ \tilde{\bu} \odot  (\sum_i  b_i \tilde{\bh}_i)^\T }
 N_{ \sum_i  b_i \tilde{\bh}_i} N_{\tilde{\bu}} {\ket{\tilde{\b0}}} \\
& = (-1)^{ \tilde{\bu} \odot  (\sum_i  b_i \tilde{\bh}_i)^\T } e^{i
\theta(\tilde{\bu})}
 N_{ \sum_i  b_i \tilde{\bh}_i}
 N_{\sum_i  x_i \tilde{\bh}_i}  N_{\sum_i  z_i \tilde{\bg}_i}
 {\ket{\tilde{\b0}}} \\
&= (-1)^{ \tilde{\bu} \odot  (\sum_i  b_i \tilde{\bh}_i)^\T } e^{i
\theta(\tilde{\bu})}
 N_{ \sum_i ( b_i + x_i) \tilde{\bh}_i}
 {\ket{\tilde{\b0}}} \\
&= (-1)^{ \tilde{\bu} \odot  (\sum_i  b_i \tilde{\bh}_i)^\T } e^{i
\theta(\tilde{\bu})}
 {\ket{\widetilde{\bb + \bx }}} \\
&= (-1)^{ {\bu} \odot  (\sum_i  b_i {\bh}_i)^\T } e^{i
\theta(\tilde{\bu})}
 {\ket{\widetilde{\bb + \bx }}},
\end{split}
\end{equation}%
where $\theta(\tilde{\bu})$ is a phase factor which is independent
of $\bb$. The first equality follows from (\ref{ono}), the second
from (\ref{eq:seconda}), the third from (\ref{eq:primi}), the fourth
from the definition of $\ket{\tilde{\b0}}$ and the fact that
$X^{\bb} X^{\bx} = X^{\bb +\bx}$, the fifth from (\ref{ono}), and
the sixth from (\ref{smor}). Similarly
\begin{equation}%
\label{eq:ludak2}%
 N_\bu \ket{\bb} = (-1)^{{\bu} \odot  (\sum_i  b_i {\bh}_i)^\T }
 e^{i \varphi({\bu})}  {\ket{{\bb + \bx }}},
\end{equation}%
where $\varphi(\bu)$ is a is a phase factor which is independent of
$\bb$.

Define $U_\Upsilon$ by the change of basis
$$
U_\Upsilon = \sum_\bb \ket{\tilde{\bb}}\bra{{\bb}}.
$$
Combining (\ref{eq:ludak1}) and (\ref{eq:ludak2}) gives for all
$\ket{\bb}$
\begin{equation}%
\begin{split}
N_{ \Upsilon (\bu)} U_\Upsilon \ket{\bb} &=
 (-1)^{ {\bu} \odot  (\sum_i  b_i {\bh}_i)^\T }
e^{i \theta(\tilde{\bu})}   U_\Upsilon  \ket{{\bb + \bx }} \\
&= e^{i [\theta(\tilde{\bu}) -   \varphi({\bu})]}
 U_\Upsilon N_\bu \ket{\bb}.
\end{split}
\end{equation}%
Therefore $[N_{\Upsilon (\bu)}] = [ U_\Upsilon N_{\bu}
U_\Upsilon^{-1} ]$.

\end{proof}

\section*{2.4 \hspace{2pt} Symplectic codes}

\addcontentsline{toc}{section}{2.4 \hspace{0.15cm} Symplectic codes}

An $[n,k]$ symplectic code $\Csp$ defined by an $(n-k)\times 2n$
parity check matrix $\Hsp$ is given by
$$ \Csp = \text{rowspace}({\Hsp})^\perp $$
where
$$ V^\perp = \{\bw: \bw \odot \bu^\T = 0, \,\, \forall \bu \in V\}. $$
The subscript $\rm sp$ emphasizes that the code is defined with
respect to the symplectic product. Note that $(V^\perp)^\perp = V$.
We say that $\Csp$ is \emph{dual-containing} if
\begin{equation}%
(\Csp)^\perp=\text{rowspace}(\Hsp)\subset \Csp;
\end{equation}%
this is true if $\Hsp$ is \emph{self-orthogonal} under the
symplectic product. For simplicity, the term ``self-orthogonal
code'' is often referred to a code with a self-orthogonal
parity-check matrix.

The notion of \emph{distance} provides a convenient way to
characterize the error-correcting properties of a code. An $[n,k]$
symplectic code $\Csp$ with a parity check matrix $\Hsp$ is said to
have distance $d$ if for each nonzero $\bu$ of weight $< d$, $\bu
\not\in \Csp$, or equivalently, $\Hsp \odot \bu^\T \neq \b0^\T$.

\section*{2.5 \hspace{2pt} Classical quaternary codes}

\addcontentsline{toc}{section}{2.5 \hspace{0.15cm} Classical
quaternary codes}\label{sec:quat}%

Following the presentation of Forney \emph{et al}. \cite{FGG07QCC},
the addition table of the additive group of the quaternary field
$\bbF_4 = \{0, 1, \omega, \overline{\omega}\}$ is given by
$$
\begin{array}{|c|cccc|}  \hline
+  & 0 & \overline{\omega} & 1 & \omega \\  \hline
0 & 0 &  \overline{\omega} & 1 & {\omega} \\
\overline{\omega} & \overline{\omega} & 0 & \omega & 1\\
1 & 1 & {\omega} & 0 &  \overline{\omega} \\
{\omega}& \omega & 1 & \overline{\omega} & 0   \\ \hline
\end{array}
$$
Comparing the above to  the addition table of $(\bbZ_2)^2$
establishes the isomorphism $\gamma: \bbF_4 \rightarrow (\bbZ_2)^2$,
given by the table
$$
\begin{array}{|c|c|} \hline
\bbF_4 & (\bbZ_2)^2 \\  \hline
 0 & 00\\
 \overline{\omega}  & 01\\
 1 & 11\\
 \omega  & 10\\ \hline
\end{array}
$$
The multiplication table for $\bbF_4$ is defined as
$$
\begin{array}{|c|cccc|}  \hline
\times  & 0 & \overline{\omega} & 1 & \omega  \\  \hline
0 & 0 & 0 & 0 & 0 \\
\overline{\omega} & 0 & \omega & \overline{\omega} & 1 \\
1 & 0 & \overline{\omega} & 1 & \omega \\
\omega & 0  & 1 & \omega & \overline{\omega} \\ \hline
\end{array}
$$
Define the \emph{traces} ($\tr$) of the elements $\{0, 1, \omega,
\overline{\omega}\}$ of $\bbF_4$ as $\{0, 0, 1, 1\}$, and their
\emph{conjugates} (``$^\dagger$'') as $\{0, 1, \overline{\omega},
\omega\}$. Intuitively, $\tr a$ measures the ``$\omega$-ness'' of $a
\in \bbF_4$. Observe that $a  = 0$  if and only if both  $\tr \omega
a = 0$ and $\tr  \overline{\omega} a = 0$. The \emph{Hermitian inner
product} of two elements $a, b \in \bbF_4$ is defined as
$\inner{a}{b} = a^\dagger b \in \bbF_4$. The \emph{trace product} is
defined as $\tr \inner{a}{b} \in \bbF_2$.  The trace  product table
is readily found to be
$$
\begin{array}{|c|cccc|}  \hline
\tr\inner{\,}{}  & 0 & \overline{\omega}& 1 & \omega  \\  \hline
0 & 0 & 0 & 0 & 0 \\
\overline{\omega} & 0& 0 & 1 & 1  \\
1 & 0& 1 & 0 & 1  \\
\omega & 0& 1 & 1 & 0  \\
\hline
\end{array}
$$
Comparing the above to  the $\odot$ table of $(\bbZ_2)^2$
establishes the identity
$$
\tr\inner{a}{b} = {\gamma(a) \odot \gamma(b)}.
$$
These notions can be generalized to $n$-dimensional vector spaces
over $\bbF_4$. Thus, for $\ba, \bb \in (\bbF_4)^n$,
\begin{equation}%
\tr\inner{\ba}{\bb} = {\gamma(\ba) \odot \gamma(\bb)^\T},
\label{unx}
\end{equation}%
where the Hermitian inner product over $(\bbF_4)^n$ is defined by
the componentwise sum $\inner{\ba}{\bb}=\sum_{i}a^\dagger b.$ Let
$\wt_4(\ba)$ be the number of non-zero bits in $\ba \in
(\bbF_4)^{n}$, then we have another identity
\begin{equation}%
\wt_{\rm sp}(\gamma(\ba)) = \wt_4(\ba), \label{deux}
\end{equation}%
where $\gamma(\ba) \in (\bbZ_2)^{2n}$.

An $[n,k]$ code $C_4$ (the subscript $4$ emphasizes that the code is
over $\bbF_4$) with the parity check matrix $H_4$ is said to have
distance $d$ if for each vector $\ba\in(\bbF_4)^n$ with
$\wt_4(\ba)<d$, $\ba\not\in C_4$, or equivalently, $\inner{H_4}
{\ba} \neq \b0^\T$.

\begin{Prop}%
\label{prop_QtoSP}%
Given an $[n,k,d]$ code $C_4$ with parity check matrix $H_4$, there
exists a corresponding $[n,2k-n,d]$ symplectic code $\Csp$.
\begin{proof}
Consider a classical $[n,k,d]_4$ code  with an $(n - k) \times n$
quaternary parity check matrix $H_4$. By definition, for each
nonzero $\ba \in (\bbF_4)^n$ such that $\wt_4(\ba) < d$,
$$
\inner{H_4}{\ba} \neq \b0^\T.
$$
This  is equivalent to the logical statement
$$
\tr \inner{\omega H_4}{\ba} \neq \b0^\T 
\,\,\vee\,\, \tr \inner{\bar{\omega} H_4}{\ba} \neq \b0^\T.
$$
This is further equivalent to
$$
\tr \inner{\tilde{H}_4}{\ba} \neq \b0^\T,
$$
where
\begin{equation}%
\label{unos} \tilde{H} = \left(\begin{array}{c}
 \omega H_4 \\
\bar{\omega} H_4
\end{array} \right).
\end{equation}%
Define the $(2n - 2k) \times 2n$ symplectic matrix $\Hsp =
\gamma(\tilde{H}_4)$. By the correspondences (\ref{unx}) and
(\ref{deux}),
$$
\Hsp \odot \bu^\T \neq \b0^\T,
$$
holds for each nonzero $\bu \in (\bbZ_2)^{2n}$ with $\wt(\bu) < d$.
Thus $\Csp$ is an $[n,2k-n,d]$ symplectic code defined by $\Hsp$.
\end{proof}
\end{Prop}%

\section*{2.6 \hspace{2pt} Encoding classical information into quantum states}

\addcontentsline{toc}{section}{2.6 \hspace{0.15cm} Encoding
classical information into quantum states}
\label{eci}%
In this section we review two schemes for sending classical
information over quantum channels: elementary coding and superdense
coding. These will be used later in the context of quantum error
correction to convey information to the decoder about which error
happened.

\subsection*{2.6.1 \hspace{2pt} Elementary coding}
\addcontentsline{toc}{subsection}{2.6.1 \hspace{0.15cm} Elementary
coding}
\label{sec_eltcoding}%
In the first scheme, Alice and Bob are connected by a perfect qubit
channel. Alice can send an arbitrary bit $a \in \bbZ_2$ over the
qubit channel in the following way:
\begin{itemize}
\item Alice locally prepares a state $\ket{0}$ in $\cH_2$.
This state is the $+1$ eigenstate of the $Z$ operator. Based on her
message $a$, she performs the encoding operation $X^{a}$, producing
the state $\ket{a}=X^a\ket{0}$.
\item Alice sends the encoded state to Bob
through the  qubit channel.
\item Bob decodes by performing the von Neumann  measurement in the
$\{\ket{0}, \ket{1} \}$ basis. As this is the unique eigenbasis of
the $Z$ operator, this is equivalently called ``measuring the $Z$
observable''.
\end{itemize}
We call this protocol ``elementary coding'' and write it
symbolically as a \emph{resource inequality} \cite{DHW03,DHW05RI,
DW03a} \footnote{In \cite{DHW05RI} resource inequalities were used
in the asymptotic sense. Here they refer to finite protocols, and
are thus slightly abusing their original intent.}
$$
[q \rightarrow q] \geq [c \rightarrow c].
$$
Here $[q \rightarrow q]$ represents a perfect qubit channel and $[c
\rightarrow c]$ represents a perfect classical bit channel. The
inequality $\geq$ signifies that the resource on the left hand side
can be used in a protocol to simulate the resource on the right hand
side.

Elementary coding immediately extends to $m$ qubits. Alice prepares
the simultaneous $+1$ eigenstate of the $Z^{\be_1}, \dots,
Z^{\be_m}$ operators $\ket{\b0}$, and encodes the message $\ba \in
(\bbZ_2)^m$ by applying $X^{\ba}$, producing the encoded state
$\ket{\ba}=X^{\ba}\ket{\b0}$. Bob decodes by simultaneously
measuring the $Z^{\be_1}, \dots, Z^{\be_m}$ observables. We could
symbolically represent this protocol by
$$
m \,[q \rightarrow q] \geq m \,[c \rightarrow c].
$$


\subsection*{2.6.2 \hspace{2pt} Superdense coding}
\addcontentsline{toc}{subsection}{2.6.1 \hspace{0.15cm} Superdense
coding}
\label{sec_sdcoding}%

In the second scheme, Alice and Bob share the ebit state
\begin{equation}%
\label{eq:ebit} \ket{\Phi} = \frac{1}{\sqrt{2}} (\ket{0} \otimes
\ket{0} + \ket{1} \otimes \ket{1})
\end{equation}%
in addition to being connected by the qubit channel. In
(\ref{eq:ebit}) Alice's state is to the left and Bob's is to the
right of the $\otimes$ symbol.

The state $\ket{\Phi}$ is the simultaneous $(+1, +1)$ eigenstate of
the commuting operators  $Z \otimes Z$ and  $X \otimes X$. Again,
the operator to the left of the $\otimes$ symbol acts on Alice's
system and the operator to  the right of the $\otimes$ symbol acts
on Bob's system. Alice can send a two-bit message $(a_1,a_2) \in
(\bbZ_2)^2$ to Bob using ``superdense coding'' \cite{BW92}:
\begin{itemize}
\item Based on her message $(a_1, a_2)$, Alice performs the encoding
operation $Z^{a_1} X^{a_2}$ on her part of the state $\ket{\Phi}$,
producing the state $\ket{a_1,a_2}=(Z^{a_1} X^{a_2} \otimes I^B
)\ket{\Phi}$.
\item Alice sends her part of the encoded state to Bob
through the perfect qubit channel.
\item Bob decodes by performing the von Neumann measurement in the
$\{(Z^{a_1} X^{a_2} \otimes I )\ket{\Phi}: (a_1, a_2) \in (\bbZ_2)^2
\}$ basis, i.e., by simultaneously measuring the $Z \otimes Z$ and
$X \otimes X$ observables.
\end{itemize}
The protocol is represented by the resource inequality
\begin{equation}%
[q \rightarrow q] + [q \, q] \geq 2 \, [c \rightarrow c],
\label{eq:sd}
\end{equation}%
where $[q \,q]$ now represents the shared ebit. It can also  be
extended to $m$ copies. Alice and Bob share the state
$\ket{\Phi}^{\otimes m}$ which is the simultaneous $+1$ eigenstate
of the $Z^{\be_1} \otimes Z^{\be_1}, \dots, Z^{\be_m} \otimes
Z^{\be_m}$ and $X^{\be_1} \otimes X^{\be_1}, \dots, X^{\be_m}
\otimes X^{\be_m}$ operators. Alice  encodes the message $(\ba_1,
\ba_2) \in (\bbZ_2)^{2m}$ by applying $Z^{\ba_1} X^{\ba_2}$,
producing the encoded state $\ket{\ba_1,\ba_2}=(Z^{\ba_1}
X^{\ba_2}\otimes I)\ket{\Phi}$. Bob decodes by simultaneously
measuring the
 $Z^{\be_1} \otimes Z^{\be_1}, \dots, Z^{\be_m} \otimes Z^{\be_m}$
and $X^{\be_1} \otimes X^{\be_1}, \dots, X^{\be_m} \otimes
X^{\be_m}$ observables. The corresponding resource inequality is
$$
m \, [q \rightarrow q] + m \,[q \, q] \geq 2m \, [c \rightarrow c].
$$
Superdense coding provides the simplest illustration of how
entanglement can increase the power of information processing.

\section*{2.7 \hspace{2pt} Useful lemmas}

\addcontentsline{toc}{section}{2.7 \hspace{0.15cm} Useful lemmas}

\begin{lemma}
\label{basis}%
Let $\cV$ be an arbitrary subgroup of $\cG_n$ with size $2^m$. Then
there exists a set of generators
$\{\bar{Z}_1,\cdots,\bar{Z}_{p+q},\bar{X}_{p+1},\cdots,\bar{X}_{p+q}\}$
that generates $\cV$ such that the $\bar{Z}$'s and $\bar{X}$'s obey
the same commutation relations as in (\ref{comm}), for some $p,q\geq
0$ and $p+2q=m$.
\begin{equation}\label{comm}
\begin{split}
[\bar{Z}_i,\bar{Z}_j]&=0 \ \ \ \ \forall i,j  \\
[\bar{X}_i,\bar{X}_j]&=0 \ \ \ \ \forall i,j  \\
[\bar{X}_i,\bar{Z}_j]&=0 \ \ \ \ \forall i\neq j \\
\{\bar{X}_i,\bar{Z}_i\}&= 0 \ \ \ \ \forall i.
\end{split}
\end{equation}
\end{lemma}
\begin{proof}%
Though the proof can be found in \cite{FCY04}; however, a new proof
can be easily obtained by combining Theorem \ref{GSsymp} and the
isomorphic map $[N]:(\bbZ_2)^{2n}\to[\cG_n]$.
\end{proof}%


%
The following lemma is a simply result from group theory, and a new
proof can be obtained from Theorem \ref{thm2} and
$[N]:(\bbZ_2)^{2n}\to[\cG_n]$.
\begin{lemma}%
\label{sim}%
If there is a one-to-one map between $\cV$ and $\cS$ which preserves
their commutation relations, which we denote $\cV\sim\cS$, then
there exists a unitary $U$ such that for each $V_i\in\cV$, there is
a corresponding $S_i\in\cS$ such that $V_i=U S_i U^{-1}$, up to a
phase which can differ for each generator.
\end{lemma}

\begin{lemma} \label{simi3}
If $\cC_0$ is a simultaneous eigenspace of Pauli operators from the
set $\cS'_0$, then $\cC = U^{-1} (\cC_0)$ is a simultaneous
eigenspace of Pauli operators  from the set $\cS = \{ U^{-1} \bA U :
\bA \in \cS'_0 \}$.
\end{lemma}

\begin{proof}
Observe that if
$$
\bA \ket{\psi} = \alpha  \ket{\psi},
$$
then
$$
(U^{-1} \bA U) U^{-1} \ket{\psi} = \alpha  U^{-1} \ket{\psi}.
$$
\end{proof}

\begin{lemma} \label{simi2}%
Performing $U$ followed by  measuring the operator ${\bf A}$ is
equivalent to measuring the operator $ U^{-1}{\bf A} U$ followed by
performing $U$.
\end{lemma}

\begin{proof}
Let $\Pi_i$ be a projector onto the eigenspace corresponding to
eigenvalue $\lambda_i$ of ${\bf A}$. Performing $U$ followed by
measuring the operator ${\bf A}$ is equivalent to the instrument
(generalized measurement) given by the set of operators $\{ \Pi_i U
\}$. The operator $ U^{-1}{\bf A} U$ has the same eigenvalues as
$\bf A$, and the projector onto the eigenspace corresponding to
eigenvalue $\lambda_i$ is $U^{-1} \Pi_i U$. Measuring the operator $
U^{-1}{\bf A} U$ followed by performing $U$ is equivalent to the
instrument $\{U (U^{-1} \Pi_i U) \} = \{ \Pi_i U \}$.
\end{proof}

\chapter*{Chapter 3: \hspace{1pt} Standard quantum error-correcting codes}
\label{cp_III}

\addcontentsline{toc}{chapter}{Chapter 3:\hspace{0.15cm} Standard
quantum error-correcting codes}

\section*{3.1 \hspace{2pt} Discretization of errors}
\label{sec_QECC_channel}%
\addcontentsline{toc}{section}{3.1 \hspace{0.15cm} Discretization of
errors}

It is well known that for standard quantum error correction (i.e.,
that unassisted by entanglement) it suffices to consider errors from
the Pauli group (see e.g. \cite{NC00}.) We will review this result
here.

Denote by $\cL$ the space of linear operators defined on the qubit
Hilbert space $\cH_2$. In general, a noisy channel is defined by a
completely positive, trace preserving (CPTP) map $\cN: \cL^{\otimes
n} \rightarrow \cL^{\otimes n}$ taking $n$-qubit density operators
on Alice's system to density operators on Bob's system. We will
often encounter isometric operators $U: \cH_2^{\otimes n_1}
\rightarrow  \cH_2^{\otimes n_2}$. The corresponding
\emph{superoperator}, or CPTP map, is marked by a hat $\hat{U}:
\cL^{\otimes n_1} \rightarrow \cL^{\otimes n_2}$ and defined by
$$
\hat{U}(\rho) = U \rho U^\dagger.
$$
Observe that $\hat{U}$ is independent of any phases factors
multiplying $U$. Thus, for a  Pauli operator $N_\bu$, $\hat{N}_\bu$
only depends on the equivalence class $[N_\bu]$.

Our communication scenario involves two spatially separated parties,
Alice and Bob, connected by a noise channel $\cN$. Alice wishes to
send $k$ qubits \emph{perfectly} to Bob through $\cN$. An $[[n,k]]$
QECC consists of
\begin{itemize}
\item An encoding isometry $\cE = \hat{U}_{\rm enc}:
 \cL^{\otimes k}  \rightarrow \cL^{\otimes n}$
\item A decoding CPTP map
$\cD: \cL^{\otimes n} \rightarrow \cL^{\otimes k}$
\end{itemize}
such that
$$
\cD \circ \cN \circ \hat{U}_{\rm enc} = \id^{\otimes k},
$$
where $\id: \cL \rightarrow \cL$ is the identity map on a single
qubit.

To make contact with classical error correction it is necessary to
discretize the errors introduced by $\cN$. This is done in two
steps. First, the CPTP map $\cN$ may be (non-uniquely) written in
terms of its Kraus representation
$$
\cN(\rho) = \sum_i A_i \rho A_i^\dagger.
$$
Second, each $A_i$ may be expanded in the Pauli operators
$$
A_i = \sum_{\bu \in  (\bbZ_2)^{2n}} \alpha_{i, \bu} N_\bu.
$$
Define the support of $\cN$ by ${\rm supp}(\cN) = \{\bu \in
(\bbZ_2)^{2n}: \exists i, \alpha_{i, \bu} \neq 0 \}$. The following
theorem allows us to replace the continuous map $\cN$ by the error
set $S = {\rm supp}(\cN)$.

\begin{theorem}%
\label{QECC_discretization}%
If $\cD \circ \hat{N}_\bu \circ \hat{U}_{\rm enc} = \id^{\otimes k}$
for all $\bu \in {\rm supp}(\cN)$, then
 $\cD \circ \cN \circ \hat{U}_{\rm enc} = \id^{\otimes k}$.
\end{theorem}

\begin{proof}
We may extend the map $\cD$ to its Stinespring dilation \cite{Sti55}
-- an isometric map $\hat{U}_{\rm dec}$  with a larger target
Hilbert space $\cL^{\otimes n} \otimes \cL'$, such that
$$
\cD (\rho)= \tr_{\cL'} \hat{U}_{\rm dec} (\rho).
$$
If for all $\bu \in {\rm supp}(\cN)$ and all pure states
$\ket{\psi}$ in $\cH_2^{\otimes n}$, the following equation holds
$$
U_{\rm dec}  N_\bu {U}_{\rm enc} \ket{\psi} =  \ket{\psi} \otimes
\ket{\bu}
$$
for some pure state $\proj{\bu}$ on $\cL'$, then by linearity, we
have
$$
U_{\rm dec} A_i {U}_{\rm enc} \ket{\psi} = \ket{\psi} \otimes
\ket{i},
$$
with the unnormalized state $\ket{i} =  \sum_\bu \alpha_{i, \bu}
\ket{\bu}$. Furthermore,
\begin{equation}%
\begin{split}
(\hat{U}_{\rm dec} \circ \cN \circ
 \hat{U}_{\rm enc}) (\proj{\psi})
&= U_{\rm dec} \left( \sum_i A_i {U}_{\rm enc} \proj{\psi} {U}_{\rm
enc}^\dagger A_i^\dagger \right)
 U_{\rm dec}^\dagger\\
& =  \proj {\psi} \otimes \sum_i \proj{i},
\end{split}
\end{equation}%
where the second subsystem corresponds to $\cL'$. Tracing out the
latter gives
$$
(\cD \circ \cN \circ \hat{U}_{\rm enc}) (\proj{\psi}) = \proj
{\psi},
$$
concluding the proof.
\end{proof}

\section*{3.2 \hspace{2pt} Canonical codes}

\addcontentsline{toc}{section}{3.2 \hspace{0.15cm} Canonical codes}

We first introduce the simplest form of standard quantum
error-correcting codes (QECCs), the canonical codes. The canonical
code $\cC_0$ is defined by the following trivial encoding operation
$\cE_{0}=\hat{U}_0$, where
\begin{equation}%
\label{QECC_en}%
U_0: \ket{\varphi} \mapsto \ket{\b0} \ket{\varphi}.
\end{equation}%
In other words, the register containing $\ket{\b0}$ (of size $s =
n-k$ qubits) is appended to the registers containing $\ket{\varphi}$
(of size $k$ qubits). We call the encoded state in (\ref{QECC_en}) a
\emph{codeword} of $\cC_0$. What errors can this canonical code
$\cC_0$ correct with such a simple-minded encoding?

\begin{Prop}
\label{QECC_code1}%
The encoding given by $\cE_0$ and a suitably-defined decoding map
$\cD_0$ can correct the error set
\begin{equation}%
\bE_0 = \{X^\ba Z^\bb \otimes X^{\alpha(\ba)} Z^{\beta(\ba)} :
 \ba,\bb \in (\bbZ_2)^s \} ,
\end{equation}%
for any fixed functions $\alpha,\beta:(\bbZ_2)^s  \to (\bbZ_2)^k$.
\end{Prop}

\begin{proof}
\begin{figure}
\centering
\includegraphics[width=5.0in]{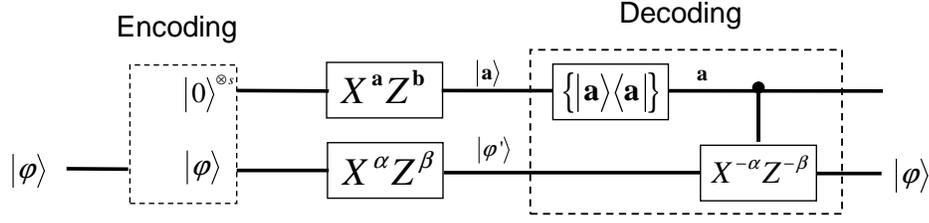}
\caption{A canonical quantum error-correcting code.} \label{QECC_CC}
\end{figure}
The protocol is shown in Figure \ref{QECC_CC}. After applying an
error $E \in \bE_0$, the channel output becomes (up to a phase
factor):
\begin{equation}%
E\left(\ket{\b0}\otimes\ket{\varphi}\right)=(X^\ba Z^\bb)\ket{\b0}
\otimes (X^{\alpha(\ba)} Z^{\beta(\ba)}) \ket{\varphi} = \ket{\ba}
\otimes \ket{\varphi'}
\end{equation}%
where $\ket{\ba}=X^{\ba}\ket{\b0}$, and
$\ket{\varphi'}=(X^{\alpha(\ba)} Z^{\beta(\ba)}) \ket{\varphi}$.

As the vector $(\ba,\bb)$ completely specifies the error operator
$E$, it is called the \emph{error syndrome}. However, in order to
correct this error, only the \emph{reduced syndrome}, $\ba$,
matters. In effect, $\ba$ has been encoded using elementary coding
(see section 2.6.1), and the receiver Bob can identify the reduced
syndrome by simultaneously measuring the
$Z^{\be_1},\cdots,Z^{\be_s}$ observables. He then performs
$X^{-\alpha(\ba)} Z^{-\beta(\ba)}$ on the remaining $k$-qubit system
$\ket{\varphi'}$, returning it to the original state
$\ket{\varphi}$.

Since the goal is the transmission of quantum information, no actual
measurement is necessary. Instead, Bob can perform the CPTP decoding
operation $\cD_0$ consisting of the controlled unitary
\begin{equation}%
\label{QECC_decode}%
U_{0,\rm dec} = \sum_{\ba}\proj{\ba}\otimes
X^{-\alpha(\ba)}Z^{-\beta(\ba)},
\end{equation}%
which is constructed based on the reduced syndrome, and is also
known as \emph{collective measurement}, followed by discarding the
unwanted systems.
\end{proof}

We can rephrase the above error-correcting procedure in terms of the
stabilizer formalism. Let $\cS_0=\langle Z_1 , \cdots , Z_s \rangle$
be an Abelian group of size $2^s$. Group $\cS_0$ is called the
stabilizer for $\cC_0$, since every element of $\cS_0$ fixes the
codewords of $\cC_0$. Notice that we have used $Z_i$ to represent
$Z^{\be_i}$ here for simplicity.

\begin{Prop}
\label{QECC_dumb}%
The QECC $\cC_0$ defined by $\cS_0$ can correct an error set $\bE_0$
if for all $E_1,E_2\in \bE_0$, $E_2^\dagger E_1 \in \cS_0 \bigcup
(\cG_n-\cZ(\cS_0))$, where $\cZ(\cS)$ is the normalizer of group
$\cS$.
\end{Prop}
\begin{proof}
Since the vector $(\ba,\bb)$ completely specifies the error operator
$E$, we consider the following two different cases:
\begin{itemize}
\item If two error operators $E_1$ and $E_2$ have the same reduced
syndrome $\ba$, then the error operator $E_2^\dagger E_1$ gives us
the all-zero reduced syndrome. Therefore, $E_2^\dagger E_1\in\cS_0$.
This error $E_2^\dagger E_1$ has no effect on the codeword.
\item If two error operators $E_1$ and $E_2$ have different reduced
syndromes, and let $\ba$ be the reduced syndrome of $E_2^\dagger
E_1$, then $E_2^\dagger E_1 \not\in \cZ(\cS_0)$. This error
$E_2^\dagger E_1$ can be corrected by the decoding operation given
in (\ref{QECC_decode}).
\end{itemize}
\end{proof}

\section*{3.3 \hspace{2pt} The general case}

\addcontentsline{toc}{section}{3.3 \hspace{0.15cm} The general case}

\begin{theorem}%
\label{QECC_general}%
Given an Abelian group $\cS_I$ of size $2^{n-k}$ that does not
contain $-I$, there exists an $[[n,k]]$ quantum error-correcting
code $\cC$ defined by the encoding and decoding pair $(\cE,\cD)$
with the following properties:
\begin{enumerate}%
\item The code $\cC$ can correct the error set $\bE$ if for all $E_1,
E_2 \in \bE$, $E_2^\dagger E_1 \in \cS_I \bigcup
(\cG_n-\cZ(\cS_I))$.
\item The codespace $\cC$ is a simultaneous eigenspace of the
$\cS_I$.
\item To decode, the reduced error syndrome is obtained by
simultaneously measuring the observables from $\cS_I$.
\end{enumerate}%
\end{theorem}%

\begin{proof}%
\begin{figure}
\centering
\includegraphics[width=5.8in]{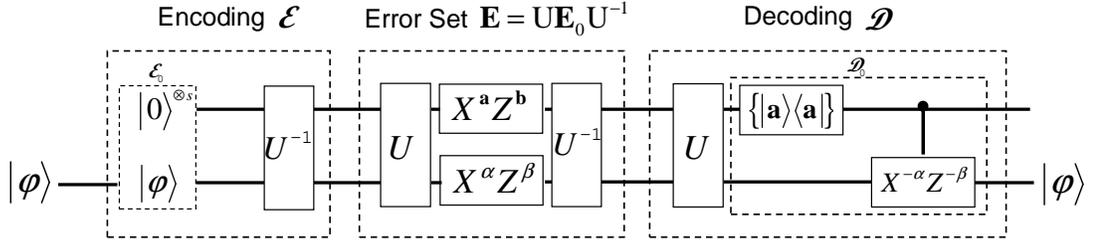}
\caption{A standard quantum error-correcting code.} \label{QECC}
\end{figure}
The protocol is shown in Figure \ref{QECC}. Since $\cS_I$ has the
same commutation relations with the stabilizer $\cS_0$ of the
canonical code $\cC_0$ given in the previous section, by Lemma
\ref{sim}, there exists an unitary matrix $U$ such that
$\cS_0=U\cS_IU^{-1}$. Define $\cE=U^{-1}\circ\cE_0$ and
$\cD=\cD_0\circ U$, where $\cE_0$ and $\cD_0$ are given in
(\ref{QECC_en}) and (\ref{QECC_decode}), respectively.
\begin{enumerate}
\item Let $\bE_0$ be the error set that can be corrected by $\cC_0$.
Then by Proposition \ref{QECC_code1},
\[
\cD_0\circ E_0\circ \cE_{0}=\text{id}^{\otimes k}
\]
for any $E_0\in\bE_0$. Let $\bE=\{U^{-1}E_0 U: \forall
E_0\in\bE_0\}$. It follows that, for any $E\in\bE$,
\[
\cD\circ E\circ\cE=\text{id}^{\otimes k}.
\]
Thus, the encoding and decoding pair $(\cE,\cD)$ corrects $\bE$.
Following Proposition \ref{QECC_dumb}, the correctable error set
$\bE$ contains all $E_1, E_2$ such that $E_2^\dagger E_1 \in \cS_I
\bigcup (\cG_n-\cZ(\cS_I)$.

\item Since $\cC_0$ is the simultaneous $+1$ eigenspace of $\cS_0$,
and $\cS_I=U^{-1}\cS_0 U$, Lemma \ref{simi3} guarantees that the
codespace $\cC$ after encoding $\cE$ is a simultaneous eigenspace of
$\cS_I$.

\item The decoding operation $\cD_0$ involves
\begin{enumerate}[i.]
\item measuring the set of generators of $\cS_0$, yielding the error
syndrome according to the error $E_0$.
\item  performing a recovering operation $E_0$ again to undo the
error.
\end{enumerate}
By Lemma \ref{simi2}, performing $\cD=\cD_0\circ U$ is equivalent to
measuring $\cS_I=U^{-1}\cS_0 U$, followed by performing the
recovering operation $U^{-1} E_0 U$ , followed by $U$ to undo the
encoding.
\end{enumerate}
\end{proof}%
\bigskip
We said an $[[n,k]]$ QECC defined by $\cS_I$ to have distance $d$,
if for all operators $E_1$ and $E_2$ with weigh $< d$ and $E_1 \neq
E_2$, either
\begin{enumerate}[i.]
\item $E_2^\dagger E_1 \notin \cG_n-\cZ(\cS_I)$, or
\item $E_2^\dagger E_1 \in \cS_I$.
\end{enumerate}
The code is called \emph{non-degenerate} if the second condition is
not invoked. A QECC with distance $d$ can correct up to $t$-qubit
errors, where $t=\lfloor (d-1)/2 \rfloor$. Such code is called an
$[[n,k,d]]$ QECC.

\section*{3.4 \hspace{2pt} Relation to symplectic codes}
\addcontentsline{toc}{section}{3.4 \hspace{0.15cm} Relation to
symplectic codes}

\begin{Prop}
\label{sptoqecc}%
Consider an $[n,k,d]$ symplectic code $\Csp$ defined by $\Hsp$. If
$\Csp$ is dual-containing, then $\Csp$ defines a non-degenerate
$[[n,k,d]]$ QECC.
\end{Prop}
\begin{proof}
Since $\Hsp$ is self-orthogonal, that means the group $\cS_I$
generated by the operator $g_i=N_{\br_i}$, where $\br_i$ is the
$i$-th row of $\Hsp$, is an Abelian group with size $2^{n-k}$. From
Theorem \ref{QECC_general}, $\cS_I$ defines an $[[n,k]]$ QECC $\cC$.

For all vectors $\bu_1,\bu_2$ with weight $<t$, where $t=\lfloor
(d-1)/2 \rfloor$, we have
$$\Hsp\odot (\bu_1-\bu_2) \neq \b0^\T,$$ or, equivalently,
$$N_{\bu_2}^{\dagger}N_{\bu_1} \not\in\cG_n-\cZ(\cS_I).$$
Therefore $\cC$ is a non-degenerate QECC with distance $d$.

\end{proof}

\subsection*{3.4.1 \hspace{2pt} The CSS construction}
\addcontentsline{toc}{subsection}{3.4.1 \hspace{0.15cm} The CSS
construction}

\begin{Prop}
Given a dual-containing classical binary codes $[n,k,d]$ $C$, there
exists an $[[n,2k-n,d]]$ QECC.
\end{Prop}

\begin{proof}
Let $H$ be the parity check matrix of $C$. Since
$$\text{rowspace}(H)= C^\perp \subset C =
\text{rowspace}(H)^\perp,$$ therefore
\begin{equation}%
\label{CSS}%
\Hsp = \left(\begin{array}{c|c}
 H & \b0\\
\b0 & H
\end{array} \right),
\end{equation}%
is dual-containing, and defines an $[n,2k-n]$ symplectic code
$\Csp$. By definition of classical linear codes, for each nonzero
$\ba \in (\bbZ_2)^n$ such that $\wt(\ba) < d$,
\begin{eqnarray*}%
\inner{H}{\ba} \neq \b0^\T,
\end{eqnarray*}%
Then
$$ \Hsp \odot \bu \neq \b0^\T,$$
holds for each nonzero $\bu \in (\bbZ_2)^{2n}$ with $\wt(\bu) < d$.
Thus $\Csp$ defines a non-degenerate $[[n,2k-n,d]]$ QECC by
Proposition \ref{sptoqecc}.
\end{proof}
Actually, instead of using the same code $C$, one can use two codes
$C_1$ and $C_2$, such that $C_1 \subset C_2$, in the CSS
construction \cite{NC00}. Furthermore, the CSS code have one
interesting property that its generators contain all $X$'s and
protect against phase flips and generators contain all $Z$'s and
protect against bit flips.

\section*{3.5 \hspace{2pt} Examples}
\addcontentsline{toc}{section}{3.5 \hspace{0.15cm} Examples}

\subsection*{3.5.1 \hspace{2pt} The $[[9,1,3]]$ Shor code}
\addcontentsline{toc}{subsection}{3.5.1 \hspace{0.15cm} The
$[[9,1,3]]$ Shor code}%

The first quantum error-correcting code constructed by Shor
\cite{Shor95} was a quantum analog of the classical repetition code,
which stores information redundantly by duplicating each bit several
times. We list the stabilizer generators for the $[[9,1,3]]$ Shor
code in Table \ref{QECC1}. It is easy to verify that it can correct
arbitrary single-qubit error.

\begin{table}[htdp]
\begin{center}
\begin{tabular}{c|ccccccccc}
\hline\hline
 $S_1$ & Z & Z & I & I & I & I & I & I & I\\
 $S_2$ & I & Z & Z & I & I & I & I & I & I\\
 $S_3$ & I & I & I & Z & Z & I & I & I & I\\
 $S_4$ & I & I & I & I & Z & Z & I & I & I\\
 $S_5$ & I & I & I & I & I & I & Z & Z & I\\
 $S_6$ & I & I & I & I & I & I & I & Z & Z\\
 $S_7$ & X & X & X & I & I & I & X & X & X\\
 $S_8$ & X & X & X & X & X & X & I & I & I\\ \hline
 $\bar{Z}$ & Z & Z & Z & Z & Z & Z & Z & Z & Z\\
 $\bar{X}$ & X & X & X & X & X & X & X & X & X\\
 \hline\hline
\end{tabular}
\end{center}
\caption{The [[9,1,3]] Shor code.} \label{QECC1}
\end{table}%

\subsection*{3.5.2 \hspace{2pt} The $[[7,1,3]]$ Steane code}
\addcontentsline{toc}{subsection}{3.5.2 \hspace{0.15cm} The
$[[7,1,3]]$ Steane code}

The second example, the $[[7,1,3]]$ Steane code, is constructed
using the CSS construction from dual-containing $[7,4,3]$ Hamming
code with the parity check matrix
\begin{equation}
H=\left(\begin{array}{ccccccc} 0 & 0 & 0 & 1 & 1 & 1 & 1 \\
0 & 1 & 1 & 0 & 0 & 1 & 1 \\
1 & 0 & 1 & 0 & 1 & 0 & 1
\end{array}\right).
\end{equation}
We list the stabilizer generators in Table \ref{QECC2}.
\begin{table}[h]
\begin{center}
\begin{tabular}{c|ccccccc}
\hline\hline
 $S_1$ & I & I & I & Z & Z & Z & Z \\
 $S_2$ & I & Z & Z & I & I & Z & Z \\
 $S_3$ & Z & I & Z & I & Z & I & Z \\
 $S_4$ & I & I & I & X & X & X & X \\
 $S_5$ & I & X & X & I & I & X & X \\
 $S_6$ & X & I & X & I & X & I & X \\ \hline
 $\bar{Z}$ & Z & Z & Z & Z & Z & Z & Z \\
 $\bar{X}$ & X & X & X & X & X & X & X \\
 \hline\hline
\end{tabular}
\end{center}
\caption{The [[7,1,3]] Steane code.} \label{QECC2}
\end{table}%

\section*{3.6 \hspace{2pt} Discussion}
\addcontentsline{toc}{section}{3.6 \hspace{0.15cm} Discussion}

We have developed a canonical code method together with the
stabilizer formalism \cite{CRSS97,DG97thesis,NC00} to introduce the
standard quantum error-correcting codes. The canonical code method
provides us essential insight into the error-correcting property.
First of all, the canonical code is obtained by attaching some
ancillas, initially in the $\ket{0}$ state, to the quantum
information we want to preserve. Therefore the \emph{codewords} of
the canonical code can be easily described by a set of commuting
Pauli $Z$ operators. The error syndrome of each correctable error
can be seen as classical information being encoded in the canonical
code by elementary coding. Therefore, reading out the error syndrome
is equivalent to recovering the classical message. Then we can
restore the codewords of the canonical code by performing a
correction operation based on the measurement outcome since the
outcome tells us which error happens. These two steps, reading out
the error syndrome and performing correction operation, are called
the decoding operation.

For a useful QECC, we expect it to be able to correct at least
arbitrary $t$-qubit errors, for some $t\geq1$. In this sense, the
canonical code is not a satisfactory QECC, but we can transform the
canonical code to a QECC with desirable distance property. The
mapping (encoding) is done with some unitary that takes the
codespace of the canonical code to the codespace specified by the
stabilizer of the QECC. Essentially, all QECCs developed to date are
stabilizer codes. The problem of finding QECCs was reduced to that
of constructing dual-containing symplectic codes, or equivalently,
classical dual-containing quaternary codes\cite{CRSS97}. When binary
codes are viewed as quaternary, this amounts to the well known
Calderbank-Shor-Steane (CSS) construction \cite{Ste96,CS96}. The
requirement that a code contains its dual is a consequence of the
need for a commuting stabilizer group. The virtue of this approach
is that we can directly construct quantum codes from classical codes
with a certain property, rather than having to develop a completely
new theory of quantum error correction from scratch. Unfortunately,
the need for a self-orthogonal parity check matrix presents a
substantial obstacle to importing the classical theory in its
entirety, especially in the context of modern codes such as
low-density parity check (LDPC) codes \cite{MMM04QLDPC}.

In the next chapter, we will show that actually every quaternary (or
binary) classical linear code, not just dual-containing codes, can
be transformed into a QECC, given that the encoder Alice and decoder
Bob have access to shared entanglement. If the classical codes are
not dual-containing, they correspond to a set of stabilizer
generators that do not commute; however, if shared entanglement is
an available resource, these generators may be embedded into larger,
commuting generators, giving a well-defined code space. We call this
the entanglement-assisted stabilizer formalism, and the codes
constructed from it are entanglement-assisted QECCs (EAQECCs).

\chapter*{Chapter 4: \hspace{1pt} Entanglement-assisted quantum error-correcting codes}%
\addcontentsline{toc}{chapter}{Chapter 4:\hspace{0.15cm}
Entanglement-assisted quantum error-correcting codes}

\label{cp_IV}%
We consider the following communication scenario depicted in Figure
\ref{figfather}. The protocol involves two spatially separated
parties, Alice and Bob, and the resources at their disposal are
\begin{itemize}
\item a noisy channel
defined by a CPTP map $\cN:  \cL^{\otimes n} \rightarrow
\cL^{\otimes n}$ taking density operators on Alice's system to
density operators on Bob's system;
\item the $c$-ebit state $\ket{\Phi}^{\otimes c}$ shared
between Alice and Bob.
\end{itemize}
Alice wishes to send $k$-qubit quantum information \emph{perfectly}
to Bob using the above resources. An $[[n,k;c]]$
entanglement-assisted quantum error correcting code (EAQECC)
consists of
\begin{itemize}
\item An encoding isometry $\cE = \hat{U}_{\rm enc}:
 \cL^{\otimes k} \otimes \cL^{\otimes c} \rightarrow \cL^{\otimes n}$
\item A decoding CPTP map
$\cD: \cL^{\otimes n} \otimes \cL^{\otimes c} \rightarrow
\cL^{\otimes k}$
\end{itemize}
such that
$$
 \cD \circ \cN \circ \hat{U}_{\rm enc} = \id^{\otimes k},
$$
where ${U}_{\rm enc}$ is the isometry which appends the state
$\ket{\Phi}^{\otimes c}$,
$$
U_{\rm enc} \ket{\varphi} =  \ket{\varphi} \ket{\Phi}^{\otimes c},
$$
and $\id: \cL \rightarrow \cL$ is the identity map on a single
qubit.
The protocol thus uses up $c$ ebits of entanglement and generates
$k$ perfect qubit channels. We represent it by the resource
inequality (with a slight abuse of notation \cite{DHW05RI})
$$
\langle \cN \rangle  + c \, [q \, q] \geq  k \, [q \rightarrow q].
$$
Even though a qubit channel is a strictly stronger resource than its
static analogue, an ebit of entanglement, the parameter $k-c$ is
still a good (albeit pessimistic) measure of the net noiseless
quantum resources gained. It should be borne in mind that a negative
value of $k$ still refers to a non-trivial protocol.

\begin{figure}
\centering
\includegraphics[width=5.0in]{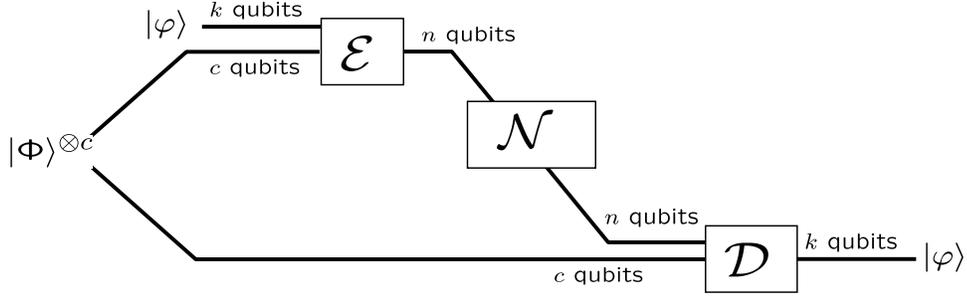}
\caption{A generic entanglement assisted quantum code.}
\label{figfather}
\end{figure}

\section*{4.1 \hspace{2pt} The channel model: discretization of errors}
\addcontentsline{toc}{section}{4.1 \hspace{0.15cm} The channel
model: discretization of errors}

Again we need to show that we can discretize the errors introduced
by $\cN$ in the entanglement-assisted communication scenario. This
can be done using steps described in Section 3.1. The continuous map
$\cN$ then can be replaced by the error set $S = {\rm supp}(\cN)$ by
Theorem \ref{QECC_discretization}.

\section*{4.2 \hspace{2pt} The entanglement-assisted canonical code}
\addcontentsline{toc}{section}{4.2 \hspace{0.15cm} The
entanglement-assisted canonical code}

The entanglement-assisted quantum error-correcting codes (EAQECCs)
come from a simple idea: replacing some portion of the ancillas of
the canonical codes (\ref{QECC_en}) by the maximally entangled
states shared between the sender and receiver. We can construct the
entanglement-assisted (EA) canonical code $\Cea_0$ with the
following trivial encoding operation $\cE_0=\hat{U}_0$ defined by
\begin{equation}%
\label{EAQECC_en}%
U_0:\ket{\varphi} \to \ket{\b0}\otimes\ket{\Phi}^{\otimes
c}\otimes\ket{\varphi} .
\end{equation}%
The operation simply appends $\ell$ ancilla qubits in the state
$\ket{\b0}$, and $c$ copies of $\ket{\Phi}$ (the maximally entangled
state shared between sender Alice and receiver Bob), to the initial
register containing the state $\ket{\varphi}$ of size $k$ qubits,
where $\ell+k+c=n$.

\begin{Prop}
\label{EAQECC_code1}%
The encoding given by $\cE_0$ and a suitably-defined decoding map
$\cD_0$ can correct the error set
\begin{equation}%
\bE_0 = \{X^\ba Z^\bb  \otimes Z^{\ba_1}X^{\ba_2} \otimes
X^{\alpha(\ba,\ba_1,\ba_2)} Z^{\beta(\ba,\ba_1,\ba_2)} :  \ba,\bb
\in (\bbZ_2)^\ell, \ba_1,\ba_2 \in (\bbZ_2)^c \} ,
\end{equation}%
 for any fixed functions $\alpha,\beta:(\bbZ_2)^\ell \times
(\bbZ_2)^c \times (\bbZ_2)^c \to (\bbZ_2)^k$.
\end{Prop}

\begin{figure}
\centering
\includegraphics[width=5.0in]{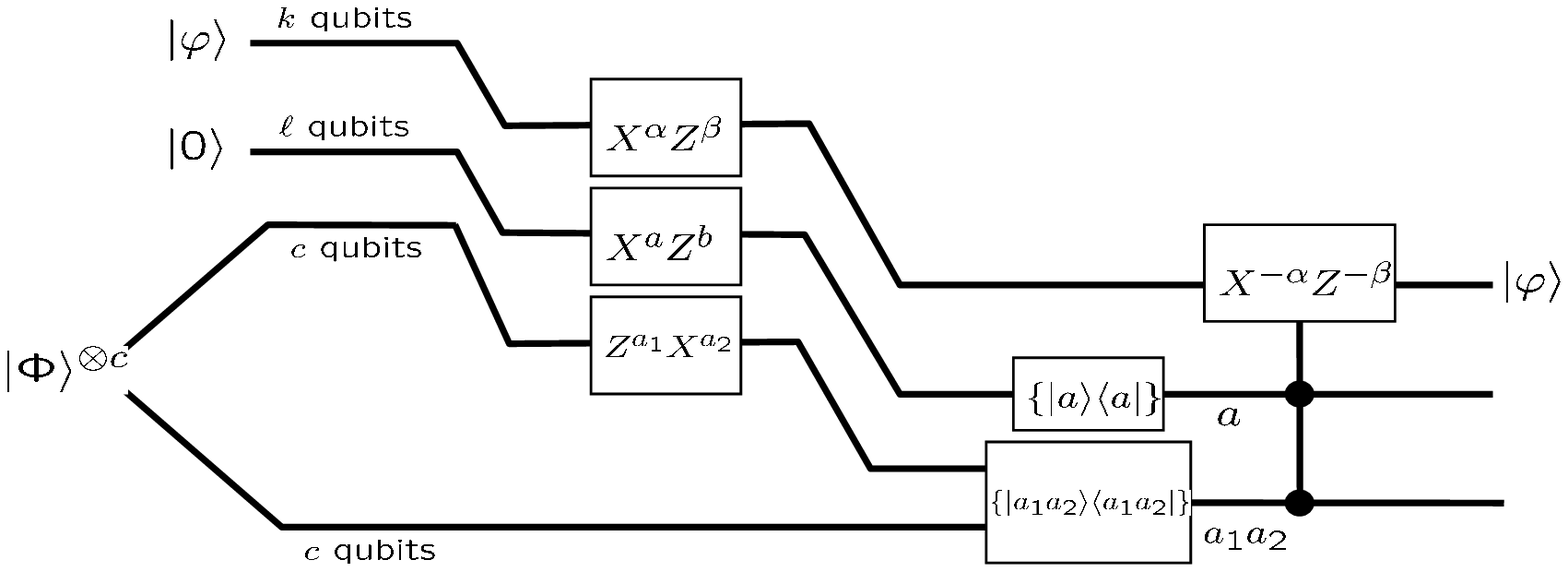}
\caption{The entanglement-assisted canonical code.}
\label{QECC_figdumb}
\end{figure}

\begin{proof}
The protocol is shown in Figure \ref{QECC_figdumb}. After applying
an error $E \in \bE_0$, the channel output becomes (up to a phase
factor):
\begin{equation}%
\label{EAQECC_out}%
(X^\ba Z^\bb)\ket{\b0} \otimes (Z^{\ba_1}X^{\ba_2}\otimes I^B)
\ket{\Phi}^{\otimes c} \otimes (X^{\alpha(\ba,\ba_1,\ba_2)}
Z^{\beta(\ba,\ba_1,\ba_2)})\ket{\varphi} = \ket{\ba} \otimes
\ket{\ba_1,\ba_2} \otimes \ket{\varphi'}
\end{equation}%
where
\begin{eqnarray}
\ket{\ba} &=& X^{\ba} Z^{\bb}\ket{\b0} =X^{\ba}\ket{\b0} \\
\ket{\ba_1, \ba_2} &=& (Z^{\ba_1} X^{\ba_2} \otimes I^B)
\ket{\Phi}^{\otimes c},\\
\ket{\varphi'}&=& X^{\alpha(\ba, \ba_1,\ba_2)} Z^{\beta(\ba, \ba_1,
\ba_2)} \ket{\varphi}. \\
\end{eqnarray}
As the vector $(\ba, \ba_1, \ba_2, \bb)$ completely specifies the
error $E$, it is called the \emph{error syndrome}. The state
(\ref{EAQECC_out}) only depends on the \emph{reduced syndrome} $\br
= (\ba, \ba_1, \ba_2)$. In effect, $\ba$ and $(\ba_1, \ba_2)$ have
been encoded using elementary and superdense coding, respectively.
Bob, who holds the entire state (\ref{EAQECC_out}), can identify the
reduced syndrome. Bob simultaneous measures the $Z^{\be_1}, \dots,
Z^{\be_\ell}$ observables to decode $\ba$, the $X^{\be_1} \otimes
X^{\be_1}, \dots, X^{\be_c} \otimes X^{\be_c}$ observables to decode
$\ba_1$, and the $Z^{\be_1} \otimes Z^{\be_1}, \dots, Z^{\be_c}
\otimes Z^{\be_c}$ observables to decode $\ba_2$. He then performs
$X^{\alpha(\ba, \ba_1, \ba_2)} Z^{\beta(\ba, \ba_1, \ba_2)}$ on the
remaining $k$ qubit system $\ket{\varphi'}$, restoring it to the
original state $\ket{\varphi}$.

Since the goal is the transmission of quantum information, no actual
measurement is necessary. Instead, Bob can perform the following
decoding $\cD_0$ consisting of the controlled unitary
\begin{equation}%
\label{EAQECC_decode}%
U_{0,\rm dec} = \sum_{\ba,\ba_1,\ba_2} \ket{\ba}\bra{\ba}\otimes
\ket{\ba_1,\ba_2}\bra{\ba_1,\ba_2} \otimes
X^{-\alpha(\ba,\ba_1,\ba_2)}Z^{-\beta(\ba,\ba_1,\ba_2)},
\end{equation}
followed by discarding the unwanted subsystems.

\end{proof}

We can rephrase the above error-correcting procedure in terms of the
stabilizer formalism. Let $\cS_0=\langle
\cS_{0,I},\cS_{0,E}\rangle$, where $\cS_{0,I}=\langle Z_1 , \cdots ,
Z_\ell \rangle$ is the isotropic subgroup of size $2^\ell$ and
$\cS_{0,E}=\langle
Z_{\ell+1},\cdots,Z_{\ell+c},X_{\ell+1},\cdots,X_{\ell+c}\rangle$ is
the \emph{symplectic} subgroup of size $2^{2c}$. We can easily
construct an Abelian extension of $\mathcal{S}_0$ that acts on $n+c$
qubits, by specifying the following generators:
\begin{equation}
\begin{split}
\label{EACC_gen}
Z_{1} &\otimes I, \\
&\vdots  \\
Z_{\ell} &\otimes I, \\
Z_{\ell+1} &\otimes Z_{1}, \\
X_{\ell+1} &\otimes X_{1}. \\
&\vdots  \\
Z_{\ell+c} &\otimes Z_{c}, \\
X_{\ell+c} &\otimes X_{c},
\end{split}
\end{equation}%
where the first $n$ qubits are on the side of the sender (Alice) and
the extra $c$ qubits are taken to be on the side of the receiver
(Bob). The operators $Z_{i}$ or $X_{i}$ to the right of the tensor
product symbol above is the Pauli operator $Z$ or $X$ acting on
Bob's $i$-th qubit. We denote such an Abelian extension of the group
$\mathcal{S}_0$ by $\widetilde{\cS}_0$. It is easy to see that the
group $\widetilde{\cS}_0$ fixes the code space $\cC^{\text{EA}}_0$
(therefore $\widetilde{\cS}_0$ is the stabilizer for $\Cea_0$), and
we will call the group $\cS_0$ the \emph{entanglement-assisted
stabilizer} for $\cC^{\text{EA}}_0$.

Consider the parameters of the EA canonical code. The number of
ancillas $\ell$ is equal to the number of generators for the
isotropic subgroup $\cS_{0,I}$. The number of ebits $c$ is equal to
the number of symplectic pairs that generate the entanglement
subgroup $\cS_{0,E}$. Finally, the number of logical qubits $k$ that
can be encoded in $C_0^{\text{EA}}$ is equal to $n-\ell-c$. To sum
up, $C_0^{\text{EA}}$ defined by $\cS_0$ is an $[[n,k;c]]$ EAQECC
that fixes a $2^k$-dimensional code space.

\begin{Prop}%
\label{EAQECC_dumb}%
The EAQECC $\cC_0^{\text{EA}}$ defined by $\cS_0=\langle
\cS_{0,I},\cS_{0,E}\rangle$ can correct an error set $\bE_0$ if for
all $E_1,E_2\in \bE_0$, $E_2^\dagger E_1 \in \cS_{0,I}\bigcup
(\cG_n-\cZ(\langle \cS_{0,I},\cS_{0,E} \rangle))$.
\end{Prop}%
\begin{proof}%
Since the vector $(\ba,\ba_1,\ba_2,\bb)$ completely specifies the
error operator $E$, we consider the following two different cases:
\begin{itemize}
\item If two error operators $E_1$ and $E_2$ have the same reduced
syndrome $(\ba,\ba_1,\ba_2)$, then the error operator $E_2^\dagger
E_1$ gives us the all-zero syndrome. Therefore, $E_2^\dagger E_1\in
\cS_{0,I}$. This error $E_2^\dagger E_1$ has no effect on the
codewords of $\cC_0^{\text{EA}}$.
\item If two error operators $E_1$ and $E_2$ have different reduced
syndromes, and let $(\ba,\ba_1,\ba_2)$ be the reduced syndrome of
$E_2^\dagger E_1$, then $E_2^\dagger E_1 \not\in Z(\langle
\cS_{0,I},\cS_{0,E}\rangle)$. This error $E_2^\dagger E_1$ can be
corrected by the decoding operation given in (\ref{EAQECC_decode}).
\end{itemize}
\end{proof}

\section*{4.3 \hspace{2pt} The general case}
\addcontentsline{toc}{section}{4.3 \hspace{0.15cm} The general case}

\begin{theorem}%
\label{EAQECC_general}%
Given a general group $\cS=\langle\cS_I,\cS_E\rangle$ with the sizes
of $\cS_I$ and $\cS_E$ being $2^{n-k-c}$ and $2^{2c}$, respectively,
there exists an $[[n,k;c]]$ EAQECC $\cC^{\text{EA}}$ defined by the
encoding and decoding pair $(\cE,\cD)$ with the following
properties:
\begin{enumerate}%
\item The code $\cC^{\text{EA}}$ can correct
the error set $\bE$ if for all $E_1, E_2 \in \bE$, $E_2^\dagger E_1
\in \cS_I\bigcup (\cG_n-\cZ(\langle\cS_I,\cS_E\rangle))$.
\item The codespace $\Cea$ is a simultaneous eigenspace of the
Abelian extension of $\cS$, $\tilde{\cS}$.
\item To decode, the reduced error syndrome is obtained by
simultaneously measuring the observables from $\widetilde{\cS}$.
\end{enumerate}%

\end{theorem}

\begin{figure}
\centering
\includegraphics[width=5.0in]{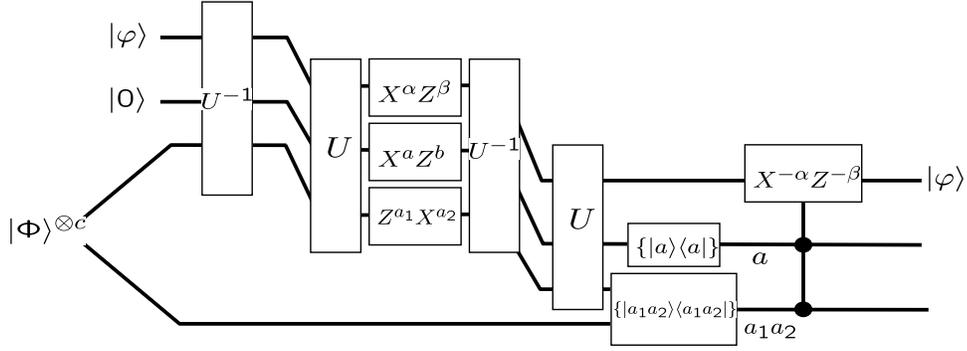}
\caption{Generalizing the entanglement-assisted canonical code
construction.} \label{fignodumb}
\end{figure}

\begin{proof}%
Since the commutation relations of $\cS$ are the same as the EA
stabilizer $\cS_0$ for the EA canonical code $\Cea_0$ in the
previous section, by Lemma \ref{sim}, there exists an unitary matrix
$U$ such that $\cS_0=U\cS U^{-1}$. The protocol is shown in Figure
\ref{fignodumb}. Define $\cE=U^{-1}\circ\cE_0$ and $\cD=\cD_0\circ
\hat{\bar{U}}$, where $\bar{U}$ is the trivial extension of $U$ are
Bob's Hilbert space, and $\cE_0$ and $\cD_0$ are given in
(\ref{EAQECC_en}) and (\ref{EAQECC_decode}), respectively.
\begin{enumerate}
\item Since
\[
\cD_0\circ E_0\circ \cE_{0}=\text{id}^{\otimes k}
\]
for any $E_0\in\bE_0$, then
\[
\cD\circ E\circ\cE=\text{id}^{\otimes k}
\]
follows for any $E\in\bE$. Thus, the encoding and decoding pair
$(\cE,\cD)$ corrects $\bE$. Following Proposition \ref{EAQECC_dumb},
the correctable error set $\bE$ contains all $E_1, E_2$ such that
$E_2^\dagger E_1 \in \cS_I \bigcup
(\cG_n-\cZ(\langle\cS_I,\cS_E\rangle))$.

\item Since $\Cea_0$ is the simultaneous $+1$ eigenspace of $\widetilde{\cS}_0$,
$\cS=U^{-1}\cS_0 U$, and by definition $\Cea=\bar{U}^{-1}(\Cea_0)$,
we conclude that $\Cea$ is a simultaneous eigenspace of
$\widetilde{\cS}$.

\item The decoding operation $\cD_0$
involves
\begin{enumerate}[i.]
\item measuring the set of generators of $\widetilde{\cS}_0$, yielding the error
syndrome according to the error $E_0$.
\item  performing a recovering operation $E_0$ again to undo the
error.
\end{enumerate}
By Lemma \ref{simi2}, performing $\cD=\cD_0\circ \hat{\bar{U}}$ is
equivalent to measuring $\widetilde{\cS}=U^{-1}\widetilde{\cS}_0 U$,
followed by performing the recovering operation $U^{-1}E_0 U$ based
on the measurement outcome, followed by $\hat{U}$ to undo the
encoding.
\end{enumerate}
\end{proof}%
%
%

\section*{4.4 \hspace{2pt} Generalized construction from quaternary codes}
\addcontentsline{toc}{section}{4.4 \hspace{0.15cm} Generalized
construction from quaternary codes}

\begin{Prop}
If a classical $[n,k,d]$ code $C_4$ exists then an $[[n, 2k-n+c, d;
c]]$ EAQECC exists for some non-negative integer $c$.
\end{Prop}

\begin{proof}
Let $H_4$ be the $(n - k) \times n$ quaternary parity check matrix
for $C_4$. By Proposition \ref{prop_QtoSP}, there exists an
$[n,2k-n,d]$ symplectic code $\Csp$ with parity check matrix $\Hsp =
\gamma(\tilde{H}_4)$, where
\begin{equation}%
\tilde{H}_4 = \left(\begin{array}{c}
 \omega H_4 \\
\bar{\omega} H_4
\end{array} \right).
\end{equation}%
Notice that even if $2k-n<0$, the following still holds
$$\Hsp \odot \bu^\T \neq \b0^\T,$$
for each nonzero $\bu \in (\bbZ_2)^{2n}$ with $\wt(\bu) < d$.

For simplicity, let $V = \text{rowspace}({\Hsp})$. Theorem
\ref{GSsymp} shows that there exists a symplectic basis consisting
of hyperbolic pairs $(\bu_i,\bv_i)$, $i=1,2,\cdots,n$, such that
$\{\bu_1,\cdots,\bu_{c+\ell},\bv_1,\cdots,\bv_c\}$ is a basis for
$V$. Then by the map $N:(\bbZ_2)^{2n}\to \cG_n$, the group
$\cS=\langle\cS_I,\cS_E\rangle$, defines an $[[n,2k-n+c,d;c]]$
EAQECC by Theorem \ref{EAQECC_general}, where
\begin{eqnarray*}
\cS_E&=&\langle N_{\bu_1},N_{\bv_1},\cdots,N_{\bu_c},N_{\bv_c}\rangle \\
\cS_I&=&\langle N_{\bu_{c+1}},\cdots,N_{\bu_{(c+\ell)}}\rangle
\end{eqnarray*}
and $$ c = \frac{1}{2} \dim \text{symp}{V}.$$

When $c=0$, $V$ is dual-containing. The above construction will give
us standard quantum error-correcting codes.
\end{proof}

Any classical binary $[n,k,d]$ code may be viewed as a quaternary
$[n,k,d]_4$ code. In this case, the above construction gives rise to
a CSS-type code.

\section*{4.5 \hspace{2pt} Bounds on performance}
\addcontentsline{toc}{section}{4.5 \hspace{0.15cm} Bounds on
performance}

\label{secVII}%
In this section we shall see that the performance of EAQECCs is
comparable to the performance of QECCs (which are a special case of
EAQECCs).

The two most important outer bounds for QECCs are the quantum
Singleton bound \cite{KL97, Pre98} and the quantum Hamming bound
\cite{Got96}. Given an $[[n,k,d]]$ QECC (which is an $[[n,k,d;0]]$
EAQECC), the quantum Singleton bound reads
$$
n - k \geq 2 (d-1).
$$
The quantum Hamming  bound holds only for non-degenerate codes and
reads
$$
\sum_{j =0}^{\lfloor \frac{d-1}{2} \rfloor} 3^j {n \choose j} \leq
2^{n - k}.
$$
The proofs of these bounds \cite{Got96,Pre98} are easily adapted to
EAQECCs. This was first noted by Bowen \cite{Bow02} in the case of
the quantum Hamming bound. Consequently, an $[[n,k,d;c]]$ EAQECC
satisfies both bounds for any value of $c$. Note that the $\bbF_4$
construction connects the quantum Singleton bound to the classical
Singleton bound $n-k \geq d-1$. An $[n,k,d]$ quaternary code
saturating the classical Singleton bound implies an
$[[n,2k-n+c,d;c]]$ EAQECC saturating the quantum Singleton bound,
that is $n-(k-c)\geq 2(d-1)$.

It is instructive to examine the asymptotic performance of quantum
codes on a particular channel. A popular choice is the tensor power
channel $\cN^{\otimes n}$, where $\cN$ is the depolarizing channel
with Kraus operators $\{ \sqrt{p_0} I, \sqrt{p_1} X, \sqrt{p_2} Y,
\sqrt{p_3} Z \}$, for some probability vector ${\bf p} = (p_0, p_1,
p_2, p_3 )$.

It is well known that the maximal transmission rate $R = k/n$
achievable by a non-degenerate QECC (in the sense of vanishing error
for large $n$ on the channel $\cN^{\otimes n}$) is equal to the
\emph{hashing bound} $R = 1 - H({\bf p})$. Here $H({\bf p})$ is the
Shannon entropy of the probability distribution ${\bf p}$. This
bound is attained by picking a random self-orthogonal code. However
no explicit constructions are known which achieve this bound.

Interestingly, the $\bbF_4$ construction also connects the hashing
bound to the Shannon bound for quaternary channels. Consider the
quaternary channel $a \mapsto a + c$, where $c$ takes on values $0,
\omega,  1, \bar{\omega}$, with respective probabilities $p_0, p_1,
p_2, p_3$. The maximal achievable rate $R = k/n$ for this channel
was proved by Shannon to equal $R = 2 - H({\bf p})$. An $[n,k]$
quaternary code saturating the Shannon bound implies an
$[[n,2k-n+c;c]]$ EAQECC, achieving the hashing bound!


\section*{4.6 \hspace{2pt} Table of codes}
\label{secIX}%
In \cite{CRSS98} a table of best known QECCs was given. Below we
show an updated table which includes EAQECCs.

\bigskip

\begin{table}[h]
\centering
\begin{tabular}{|c||c|c|c|c|c|c|c|c|c|c|c|}
  \hline
  $n \backslash k-c$ & 0 & 1 & 2 & 3 & 4 & 5 & 6 & 7 & 8 & 9 & 10  \\ \hline
  3              & 2 & $2^*$ & 1 & 1 &   &   &   &  & & &     \\ \hline
  4              & $3^*$ & 2 & 2& 1 & 1 &  &  &  & & &      \\ \hline
  5              & 3 & 3 & 2 & $2^*$ & 1 & 1 &  &  & & &      \\ \hline
  6              & 4 & 3 & 2 & 2 & 2 & 1 & 1 & & & &     \\ \hline
  7              & 3 & 3 & 2 & 2 & 2 & $2^*$ & 1 & 1 &  & &  \\ \hline
  8              & 4 & 3 & 3 & 3 & 2 & 2 & 2 & 1 & 1 &  &  \\ \hline
  9              & 4 & $4^*$ & 3 & 3 & 2 & 2 & 2 & $2^*$ & 1 & 1 &   \\ \hline
  10             & $5^*$ & 4 & 4 & 3 & 3 & 2 & 2 & 2 & 2 & 1 & 1 \\ \hline
  \hline
\end{tabular}
\caption{Highest achievable minimal distance $d$ in any
$[[n,k,d;c]]$ EAQECCs.}
\end{table}
\bigskip

The entries with an asterisk mark the improvements over the table
from \cite{CRSS98}. All these are obtained from  Proposition 3.1.
The corresponding classical quaternary code is available online at
{\tt http://www.win.tue.nl/$\sim$aeb/voorlincod.html}.

The general methods from \cite{CRSS98} for constructing new codes
from old also apply here. Moreover, new constructions are possible
since the self-orthogonality condition is removed. An example is
given by the following Theorem.

\begin{theorem}
\begin{enumerate}[(a)]
\item Suppose an $[[n,k,d;c]]$ code exists, then an
$[[n+1,k-1,d';c']]$ code exists for some $c'$ and $d' \geq d$;
\item Suppose a non-degenerate $[[n,k,d;c]]$ code exists, then an
$[[n-1,k+1,d-1;c']]$ code exists for some $c'$.
\end{enumerate}
\end{theorem}

\begin{proof}
(a) Recall that the net yield is $\hat{k}=k-c$. Let $H$ be the
$(n-\hat{k} \times 2n)$ parity check matrix of the $[[n,k,d;c]]$
code. The parity check matrix of the new $[[n+1,\hat{k}-1,d';c']]$
is then
\begin{equation}%
H'=\left(\begin{tabular}{cc|cc}
  0 $\cdots$ 0 & 0 & 1 $\cdots$ 1 & 1\\
  1 $\cdots$ 1 & 1 & 0 $\cdots$ 0 & 0\\
   & 0  &  &0\\
   \raisebox{0.5ex}[0cm][0cm]{$H_Z$} & \vdots & \raisebox{0.5ex}[0cm][0cm]{$H_X$} & \vdots \\
   & 0  &  &0
\end{tabular} \right).
\end{equation}%
This corresponds to the classical construction of adding a parity
check at the end of the codeword  \cite{FJM77}. The additional rows
ensure that errors involving the last qubit are detected. Sometimes
the distance actually increases: for instance, the $[[8,0,4]]$ is
obtained from the $[[7,1,3]]$ code in this way.

(b) We mimic the classical ``puncturing'' method \cite{FJM77}. Let
$C$ be the $(n+\hat{k})$-dimensional subspace of $(\bbZ_2)^{2n}$
corresponding to the  $[[n,k,d;c]]$ EAQEC code. Puncturing $C$ by
deleting the first $Z$ and $X$ coordinate, we obtain a new ``code''
$C'$ which is an $(n+\hat{k})$-dimensional subspace of
$(\bbZ_2)^{2(n-1)}$. This corresponds to an $[[n-1,k+1,d-1;c']]$
EAQEC code, as the minimum distance between the ``codewords'' of $C$
decreases by at most $1$.
\end{proof}

\section*{4.7 \hspace{2pt} Discussion}
\label{secX}%

Motivated by recent developments in quantum Shannon theory, we have
introduced a generalization of the stabilizer formalism to the
setting in which the encoder Alice and decoder Bob pre-share
entanglement (EAQECCs). The powerful canonical code technique again
provides us essential insight into the error-correcting property.
First of all, the entanglement-assisted canonical code is obtained
by replacing some ancillas of the standard canonical code with
maximally entangled states. The codewords of the
entanglement-assisted canonical code then can be described by a set
of commuting operators (see (\ref{EACC_gen})). The error syndrome of
each correctable error can be seen as classical information being
encoded in the entanglement-assisted canonical code by either
elementary coding or superdense coding. Therefore, reading out the
error syndrome is equivalent to recovering the classical message.
Then we can restore the codewords of the entanglement-assisted
canonical code by performing a correction operation based on the
measurement outcome since the outcome tells us which error happens.
These two steps, reading out the error syndrome and performing
correction operation, are called the decoding operation.

Up to this point, the entanglement-assist canonical code is nothing
but the stabilizer formalism. What makes the entanglement-assisted
canonical code different is when half of the maximally entangled
states are assumed to be originally possessed by the receiver Bob
(These half of ebits do not go through the noisy channel). Then the
operators on Alice's sie form a non-commuting set of generators,
allowing us to map arbitrary classical quaternary codes to EAQECCs.

There are two practical advantages of EAQECCs over standard QECCs:
\begin{enumerate}
\item They are much easier to construct from classical
codes because self-orthogonality is not required. This allows us to
import the classical theory of error correction wholesale, including
capacity-achieving modern codes. The attraction of these modern
codes comes from the existence of efficient decoding algorithms that
provide excellent trade-off between decoding complexity and decoding
performance. In fact, these decoding algorithms, such as sum-product
algorithm, can be modified to decode the error syndromes effectively
\cite{MMM04QLDPC}. The only problem of using these iterative
decoding algorithms on quantum LDPC actually comes from those
shortest 4-cycles that were introduced inevitably because of
self-orthogonality constrain. However, we have demonstrated recently
that by allowing assisted entanglement, those 4-cycles can be
eliminated completely, and the performance of the iterative decoding
improves dramatically by our numerically simulation results (see
Chapter 7). This finding further confirms the
contribution of our EA formalism. 

\item Comparing $[[n,k,d;c]]$ EAQECCs to
$[[n,k,d;0]]$ QECCs is not being entirely fair to former, since the
entanglement used in the protocol is a strictly weaker resource than
quantum communication. However, by using an EAQECC, we typically
achieve a \emph{higher rate} for the same distance, or a
\emph{higher distance} for the same rate, than a QECC; and because
entanglement is a ``cheaper'' resources, this is often a worthwhile
trade-off.  Or to think of it a different way, if we construct an
EAQECC and a QECC from two classical codes with the same parameters
$[n,k,d]$, the EAQECC will have a higher rate; or by using an EAQECC
derived from a classical code with higher distance and lower rate,
we can achieve the same rate and a higher distance than a QECC.
\end{enumerate}

If one is interested in applications to fault tolerant quantum
computation, where the resource of entanglement is meaningless, high
values of $c$ are unwelcome because they require a long seed QECCs.
We expect this obstacle to be overcome by bootstrapping.

Another fruitful line of investigation connects to quantum
cryptography. Quantum cryptographic protocols, such as BB84, are
intimately related to CSS QECCs. In \cite{LD07QKD} it is shown that
EAQECCs analogues of CSS codes give rise to key expansion protocols
which do not rely on the existence of long self-orthogonal codes.

\chapter*{Chapter 5: \hspace{1pt} Operator quantum error-correcting codes}
\addcontentsline{toc}{chapter}{Chapter 5:\hspace{0.15cm} Operator
quantum error-correcting codes}

\label{cp_V}%
In this chapter, we will briefly review the well-known operator
quantum error-correcting codes (OQECCs), using the canonical code
method and linking to the operator stabilizer formalism.

\section*{5.1 \hspace{2pt} The canonical code}
\addcontentsline{toc}{section}{5.1 \hspace{0.15cm} The canonical
code}

The idea of OQECCs also comes from a simple idea: replacing some
portion of the ancillas of the canonical code (\ref{QECC_en}) by
some garbage states. We can construct the operator canonical code
$\Cop_0$ with the following trivial encoding operation $\cE_0$
defined by
\begin{equation}
\label{OQECC_en}%
\cE_0:\proj{\varphi} \to \proj{\b0} \otimes \sigma \otimes
\proj{\varphi} .
\end{equation}
The operation simply appends $s$ ancilla qubits in the state
$\ket{\b0}$, and an arbitrary state $\sigma$ of size $r$ qubits, to
the initial register containing the state $\ket{\varphi}$ of size
$k$ qubits, where $s+k+r=n$. These $r$ extra garbage qubits are
called the gauge qubits. Two states of this form which differ only
in $\sigma$ are considered to encode the same quantum information.

\begin{Prop}
\label{prop_OQECC_cc}%
The encoding given by $\cE_0$ and a suitably-defined decoding map
$\cD_0$ can correct the error set
\begin{equation}%
\bE_0 = \{X^\ba Z^\bb \otimes X^\bc Z^\bd \otimes X^{\alpha(\ba)}
Z^{\beta(\ba)} : \ba,\bb \in (\bbZ_2)^s, \bc,\bd \in (\bbZ_2)^r \},
\end{equation}%
for any fixed functions $\alpha,\beta:(\bbZ_2)^s \to (\bbZ_2)^k$.
\end{Prop}
\begin{proof}
\begin{figure}
\centering
\includegraphics[width=5.8in]{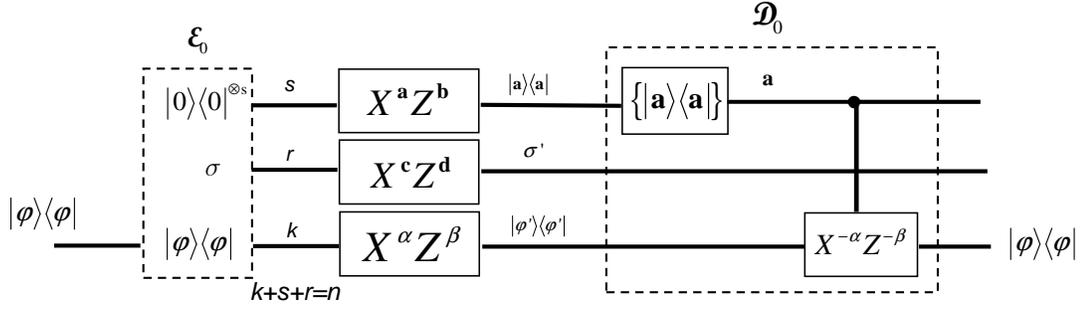}
\caption{The operator canonical code.} \label{OQECC_CC}
\end{figure}
The protocol is shown in Figure \ref{OQECC_CC}. After applying an
error $E \in \bE_0$, the channel output becomes (up to a phase
factor):
\begin{equation}%
\begin{split}
(X^\ba Z^\bb)\proj{\b0} & (X^\ba Z^\bb)^\dagger \otimes (X^\bc
Z^\bd) \sigma (X^\bc Z^\bd)^\dagger \otimes (X^{\alpha(\ba)}
Z^{\beta(\ba)})\proj{\varphi} (X^{\alpha(\ba)} Z^{\beta(\ba)})^\dagger \\
= &\proj{\ba} \otimes \sigma' \otimes \proj{\varphi'}
\end{split}
\end{equation}%
where
\begin{eqnarray}
\ket{\ba}&=& X^{\ba}\ket{\b0}, \\
\sigma' &=& (X^\bc Z^\bd) \sigma (X^\bc Z^\bd)^\dagger, \\
\ket{\varphi'}&=&(X^{\alpha(\ba)} Z^{\beta(\ba)})\ket{\varphi}.
\end{eqnarray}

As the vector $(\ba, \bb, \bc, \bd)$ completely specifies the error
operator $E$, it is called the \emph{error syndrome}. However, in
order to correct this error, only the \emph{reduced syndrome} $\ba$
matters. Here two kinds of \emph{passive} error correction are
involved. The errors that come from vector $\bb$ are passively
corrected because they do not affect the encoded state given in
(\ref{OQECC_en}). The errors that come from vector $(\bc,\bd)$ are
passively corrected because of the subsystem structure inside the
code space: $\rho\otimes\sigma$ and $\rho\otimes\sigma'$ represent
the same information, differing only by a gauge operation. Though
these errors change the encoded states, they do not damage the
information encoded in the states.

The decoding operation $\cD_0$ is constructed based on the reduced
syndrome, and is also known as \emph{collective measurement}. Bob
can recover the state $\ket{\varphi}$ by performing the decoding
$\cD_0$:
\begin{equation}%
\label{OQECC_decode} \cD_0 = \sum_{\ba} \ket{\ba}\bra{\ba}\otimes I
\otimes X^{-\alpha(\ba)}Z^{-\beta(\ba)},
\end{equation}%
followed by discarding the unwanted systems.
\end{proof}

We can rephrase the above error-correcting procedure in terms of the
stabilizer formalism. Let $\cS_0=( \cS_{0,I},\cS_{0,G})$, where
$\cS_{0,I}=\langle Z_1 , \cdots , Z_s \rangle$ is the isotropic
subgroup of size $2^s$ and $\cS_{0,G}=\langle
Z_{s+1},\cdots,Z_{s+r},X_{s+1},\cdots,X_{s+r}\rangle$ is the
\emph{symplectic} subgroup of size $2^{2r}$.

It follows that the two subgroups $(\cS_{0,I},\cS_{0,G})$ define the
canonical OQECC $\Cop_0$ given in (\ref{OQECC_en}). The subgroup
$\cS_{0,I}$ defines a $2^{k+r}$-dimensional code space $\Cop_0$, and
the gauge subgroup $\cS_{0,G}$ specifies all possible operations
that can happen on the gauge qubits. Thus we can use $\cS_{0,G}$ to
define an equivalence class between two states in the code space of
the form: $\rho\otimes \sigma$ and $\rho\otimes\sigma'$, where
$\rho$ is a state on $\cH_2^{\otimes k}$, and $\sigma,\sigma'$ are
states on $\cH_2^{\otimes r}$. Consider the parameters of the
canonical code. The number of ancillas $s$ is equal to the number of
generators for the isotropic subgroup $\cS_{0,I}$. The number of
gauge qubits $r$ is equal to the number of symplectic pairs for the
gauge subgroup $\cS_{0,G}$. Finally, the number of logical qubits
$k$ that can be encoded in $\Cop_0$ is equal to $n-s-r$. To sum up,
$\Cop_0$ defined by $(\cS_{0,I},\cS_{0,G})$ is an $[[n,k;r]]$ OQECC
that fixes a $2^{k+r}$-dimensional code space, within which
$\rho\otimes \sigma$ and $\rho\otimes \sigma'$ are considered to
carry the same information. Notice that there is a tradeoff between
the number of encoded bits and gauge bits, in that we can reduce the
rate by improving the error-avoiding ability or vice versa.

\begin{Prop}
\label{OQECC_dumb} The OQECC $\Cop_0$ defined by
$(\cS_{0,I},\cS_{0,G})$ can correct an error set $\bE_0$ if for all
$E_1,E_2\in \bE_0$, $E_2^\dagger E_1 \in
\langle\cS_{0,I},\cS_{0,G}\rangle\bigcup (\cG_n-\cZ(\cS_{0,I}))$.
\end{Prop}
\begin{proof}
Since the vector $(\ba,\bb,\bc,\bd)$ completely specifies the error
operator $E$, we consider the following two different cases:
\begin{itemize}
\item If two error operators $E_1$ and $E_2$ have the same reduced
syndrome $\ba$, then the error operator $E_2^\dagger E_1$ gives us
all-zero reduced syndrome with some vector $(\bb,\bc,\bd)$.
Therefore, $E_2^\dagger E_1\in\langle\cS_{0,I},\cS_{0,G}\rangle$.
This error $E_2^\dagger E_1$ has no effect on the logical state
$\proj{\varphi}$.
\item If two error operators $E_1$ and $E_2$ have different reduced
syndromes, and let $\ba$ be the reduced syndrome of $E_2^\dagger
E_1$, then $E_2^\dagger E_1 \not\in \cZ(\cS_{0,I})$. This error
$E_2^\dagger E_1$ can be corrected by the decoding operation given
in (\ref{OQECC_decode}).
\end{itemize}
\end{proof}

\section*{5.2 \hspace{2pt} The general case}
\addcontentsline{toc}{section}{5.2 \hspace{0.15cm} The general case}

\begin{theorem}%
\label{OQECC_general}%
Given a general group $\cS=\langle\cS_I,\cS_G\rangle$ with the sizes
of $\cS_I$ and $\cS_G$ being $2^{n-k-r}$ and $2^{2r}$, respectively,
there exists an $[[n,k;r]]$ OQECC $\Cop$ defined by the encoding and
decoding pair $(\cE,\cD)$ with the following properties:
\begin{enumerate}%
\item The code $\Cop$ can correct
the error set $\bE$ if for all $E_1, E_2 \in \bE$, $E_2^\dagger E_1
\in \langle \cS_I,\cS_G\rangle \bigcup (\cG_n-\cZ(\cS_I))$.
\item The codespace $\Cop$ is a simultaneous eigenspace of $\cS_I$.
\item To decode, the reduced error syndrome is obtained by
simultaneously measuring the observables from $\cS_I$.
\end{enumerate}%
\end{theorem}%

\begin{proof}%
\begin{figure}
\centering
\includegraphics[width=5.8in]{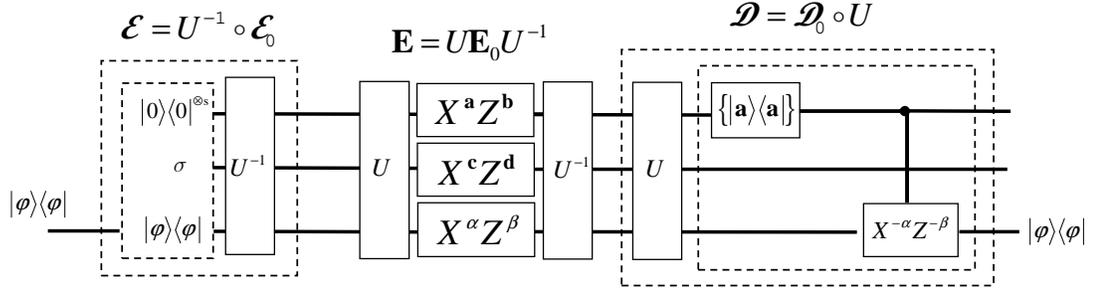}
\caption{The operator quantum error-correcting code.} \label{OQECC}
\end{figure}

Since the commutation relations of $\cS=(\cS_{I},\cS_{G})$ are the
same as the OP stabilizer $\cS_0=(\cS_{0,I},\cS_{0,G})$ for the OP
canonical code $\Cop_0$ in the previous section, by Lemma \ref{sim},
there exists an unitary matrix $U$ such that $\cS_0=U\cS U^{-1}$.
The protocol is shown in Figure \ref{OQECC}. Define
$\cE=\hat{U}^{-1}\circ\cE_0$ and $\cD=\cD_0\circ \hat{U}$, and
$\cE_0$ and $\cD_0$ are given in (\ref{OQECC_en}) and
(\ref{OQECC_decode}), respectively.
\begin{enumerate}
\item Since
\[
\cD_0\circ E_0\circ \cE_{0}=\text{id}^{\otimes k}
\]
for any $E_0\in\bE_0$, then
\[
\cD\circ E\circ\cE=\text{id}^{\otimes k}
\]
follows for any $E\in\bE$. Thus, the encoding and decoding pair
$(\cE,\cD)$ corrects $\bE$. Following Proposition \ref{OQECC_dumb},
the correctable error set $\bE$ contains all $E_1, E_2$ such that
$E_2^\dagger E_1 \in \langle \cS_I,\cS_G\rangle \bigcup
(\cG_n-\cZ(\cS_I))$.

\item Since $\Cop_0$ is the simultaneous $+1$ eigenspace of $\cS_{0,I}$,
$\cS=U^{-1}\cS_0 U$, and by definition $\Cop=U^{-1}(\Cop_0)$, we
conclude that $\Cop$ is a simultaneous eigenspace of $\cS_I$.

\item The decoding operation $\cD_0$
involves
\begin{enumerate}[i.]
\item measuring the set of generators of $\cS_0$, yielding the error
syndrome according to the error $E_0$.
\item  performing a recovering operation $E_0$ again to undo the
error.
\end{enumerate}
By Lemma \ref{simi2}, performing $\cD=\cD_0\circ \hat{U}$ is
equivalent to measuring $\cS=U^{-1}\cS_0 U$, followed by performing
the recovering operation $U^{-1}E_0 U$ based on the measurement
outcome, followed by $\hat{U}$ to undo the encoding.
\end{enumerate}
\end{proof}%

\section*{5.3 \hspace{2pt} Discussion}
\addcontentsline{toc}{section}{5.3 \hspace{0.15cm} Discussion}

The idea of the operator canonical code comes from replacing some
portion of ancillas of the standard canonical code with an arbitrary
garbage state that we do not care about. In terms of the operator
stabilizer formalism, the codespace of the operator canonical code
is described by a set of commuting Pauli $Z$ operators together with
a set of anti-commuting operators specifying all possible operations
that can occur on the garbage state. These operations on the garbage
state do not affect our quantum information, therefore no correction
is needed, and thus the passive error-correcting power is increased.
The error syndrome of each correctable error can be seen as
classical information being encoded in the operator canonical code
by elementary coding. Therefore, reading out the error syndrome is
equivalent to recovering the classical message. Then we can restore
the codewords of the operator canonical code by performing a
correction operation based on the measurement outcome since the
outcome tells us which error happens. These two steps, reading out
the error syndrome and performing correction operation, are called
the decoding operation.

The operator quantum error-correcting codes are a combination of
standard quantum error-correcting codes (active error correction)
and the passive passive error-avoiding schemes, such as
decoherence-free subspaces and noiseless subsystems. The operator
stabilizer is generated by a set of non-commuting generators.
Therefore, we can map arbitrary classical quaternary codes to
OQECCs, though the distance of the OQECCs is not always guaranteed.
There has been a couple of clever construction of OQECCs whose
distance is inherited from their classical counterpart
\cite{AKS06OQECC,KS06}.

The advantage of OQECCs comes from the fact that it is not necessary
to actively correct all errors, but rather only to perform
correction modulo the subsystem structure. One potential benefit of
the new decoding procedure is to improve the threshold of
fault-tolerant quantum computation. This research direction remains
a hot topic in quantum computation.

\chapter*{Chapter 6: \hspace{1pt} Entanglement-assisted operator quantum error-correcting codes}
\label{cp_VI}%
\addcontentsline{toc}{chapter}{Chapter 6:\hspace{0.15cm}
Entanglement-assisted operator quantum error-correcting codes}

Now it becomes clear how to combine the idea of
entanglement-assisted and operator formalism, to construct the
entanglement-assisted operator quantum error-correcting codes
(EAQECCs). We will begin with its canonical code.

\section*{6.1 \hspace{2pt} The canonical code}
\addcontentsline{toc}{section}{6.1 \hspace{0.15cm} The canonical
code}

We illustrate the idea of EAOQECCs by the following canonical code.
Consider the trivial encoding operation $\cE_0$ defined by
\begin{equation}
\label{EAOQECC_en} \cE_0:\proj{\varphi} \to \proj{\b0}^{\otimes s}
\otimes \proj{\Phi}^{\otimes c} \otimes \sigma \otimes
\proj{\varphi} .
\end{equation}
The operation simply appends $s$ ancilla qubits in the state
$\ket{\b0}$, $c$ copies of $\ket{\Phi}$ (a maximally entangled state
shared between sender Alice and receiver Bob), and an arbitrary
state $\sigma$ of size $r$ qubits, to the initial register
containing the state $\ket{\varphi}$ of size $k$ qubits, where
$s+k+r+c=n$. These $r$ extra qubits are the gauge qubits. Two states
of this form which differ only in $\sigma$ are considered to encode
the same quantum information.

\begin{Prop}
\label{EAOQECC_code} The encoding given by $\cE_0$ and a
suitably-defined decoding map $\cD_0$ can correct the error set
\begin{equation}%
\begin{split} \bE_0 =& \{X^\ba Z^\bb  \otimes
Z^{\ba_1}X^{\ba_2} \otimes X^\bc Z^\bd \otimes
X^{\alpha(\ba,\ba_1,\ba_2)} Z^{\beta(\ba,\ba_1,\ba_2)} :
\\ & \ba,\bb \in (\bbZ_2)^s, \ba_1,\ba_2 \in (\bbZ_2)^c, \bc,\bd \in (\bbZ_2)^r \} ,
\end{split}
\end{equation}%
for any fixed functions $\alpha,\beta:(\bbZ_2)^s \times (\bbZ_2)^c
\times (\bbZ_2)^c \to (\bbZ_2)^k$.
\end{Prop}

\begin{proof}
\begin{figure}
\centering
\includegraphics[width=5.8in]{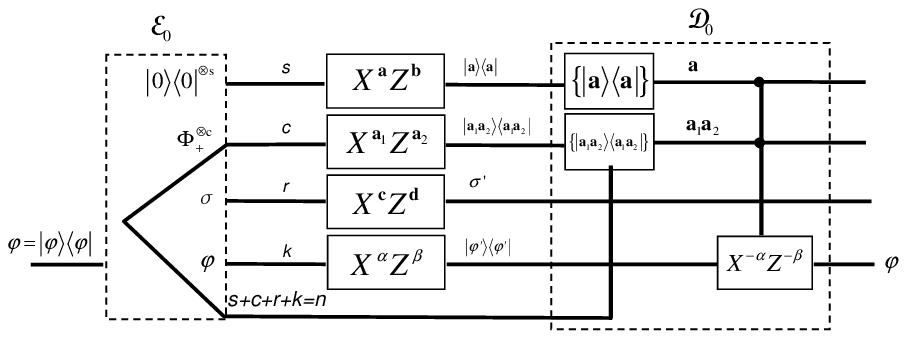}
\caption{The entanglement-assisted operator canonical code.}
\label{EAOQECC_CC}
\end{figure}
The protocol is shown in Figure \ref{EAOQECC_CC}. After applying an
error $E \in \bE_0$, the channel output becomes (up to a phase
factor):
\begin{equation}%
\proj{\ba} \otimes \proj{\ba_1,\ba_2}\otimes \sigma'
\otimes\proj{\varphi'},
\end{equation}%
where
\begin{eqnarray}
\ket{\ba}&=& X^{\ba}\ket{\b0}, \\
\ket{\ba_1,\ba_2} &=& (Z^{\ba_1} X^{\ba_2}\otimes I^B)
\ket{\Phi}^{\otimes c}, \\
\sigma' &=&(X^\bc Z^\bd) \sigma (X^\bc Z^\bd)^\dagger, \\
\ket{\varphi'}&=&(X^{\alpha(\ba,\ba_1,\ba_2)}
Z^{\beta(\ba,\ba_1,\ba_2)}) \ket{\varphi}.
\end{eqnarray}

As the vector $(\ba, \ba_1, \ba_2, \bb, \bc, \bd)$ completely
specifies the error operator $E$, it is called the \emph{error
syndrome}. However, in order to correct this error, only the
\emph{reduced syndrome} $(\ba,\ba_1,\ba_2)$ matters. The
entanglement-assisted operator canonical code $\Ceao_0$ keeps
advantages of both EAQECCs and OQECCs. On one hand, the two kinds of
passive error correction are preserved. On the other hand, the power
of active error correction is increased by the use of pure
entanglement.

The decoding operation $\cD_0$ is constructed based on the reduced
syndrome. Bob can recover the state $\ket{\varphi}$ by performing
the decoding $\cD_0$:
\begin{equation}%
\label{EAOQECC_decode}
\begin{split}
\cD_0 = &\sum_{\ba,\ba_1,\ba_2} \ket{\ba}\bra{\ba}\otimes
\ket{\ba_1,\ba_2}\bra{\ba_1,\ba_2} \otimes I \\ &\otimes
X^{-\alpha(\ba,\ba_1,\ba_2)}Z^{-\beta(\ba,\ba_1,\ba_2)},
\end{split}
\end{equation}%
followed by discarding the unwanted systems.
\end{proof}

We can rephrase the above error-correcting procedure in terms of the
stabilizer formalism. Let $\cS_0=\langle
\cS_{0,I},\cS_{0,S}\rangle$, where $\cS_{0,I}=\langle Z_1 , \cdots ,
Z_s \rangle$ is the isotropic subgroup of size $2^s$ and
$\cS_{0,S}=\langle
Z_{s+1},\cdots,Z_{s+c+r},X_{s+1},\cdots,X_{s+c+r}\rangle$ is the
\emph{symplectic} subgroup of size $2^{2(c+r)}$. We can further
divide the symplectic subgroup $\cS_{0,S}$ into an entanglement
subgroup
\[
\cS_{0,E}=\langle Z_{s+1},\cdots,Z_{s+c},X_{s+1},\cdots,X_{s+c}
\rangle
\]
of size $2^{2c}$ and a gauge subgroup
\[
\cS_{0,G}=\langle
Z_{s+c+1},\cdots,Z_{s+c+r},X_{s+c+1},\cdots,X_{s+c+r} \rangle
\]
of size $2^{2r}$, respectively. The generators of
$(\cS_{0,I},\cS_{0,E},\cS_{0,G})$ are arranged in the following
form:
\begin{equation}
\begin{array}{cccc}
Z^{\be_i} & I & I & I \\
I & Z^{\be_j} & I & I \\
I & X^{\be_j} & I & I \\
I & I & Z^{\be_l} & I \\
I & I & X^{\be_l} & I \\
\overleftrightarrow{s} & \overleftrightarrow{c} &
\overleftrightarrow{r} & \overleftrightarrow{k}
\end{array}
\end{equation}
where $\{\be_i\}_{i\in[s]}$, $\{\be_j\}_{j\in[c]}$, and
$\{\be_l\}_{l\in[r]}$ are the set of standard bases in $(\bbZ_2)^s$,
$(\bbZ_2)^c$, and $(\bbZ_2)^r$, respectively, and
$[k]\equiv\{1,\cdots,k\}$.

It follows that the three subgroups
$(\cS_{0,I},\cS_{0,E},\cS_{0,G})$ define the canonical code
$\Ceao_0$ given in (\ref{EAOQECC_en}). The subgroups $\cS_{0,I}$ and
$\cS_{0,E}$ define a $2^{k+r}$-dimensional code space
$\Ceao_0\subset \cH^{\otimes (n+c)}$, and the gauge subgroup
$\cS_{0,G}$ specifies all possible operations that can happen on the
gauge qubits. Thus we can use $\cS_{0,G}$ to define an equivalence
class between two states in the code space of the form: $\rho\otimes
\sigma$ and $\rho\otimes\sigma'$, where $\rho$ is a state on
$\cH^{\otimes k}$, and $\sigma,\sigma'$ are states on $\cH^{\otimes
r}$. Consider the parameters of the canonical code. The number of
ancillas $s$ is equal to the number of generators for the isotropic
subgroup $\cS_{0,I}$. The number of ebits $c$ is equal to the number
of symplectic pairs that generate the entanglement subgroup
$\cS_{0,E}$. The number of gauge qubits $r$ is equal to the number
of symplectic pairs for the gauge subgroup $\cS_{0,G}$. Finally, the
number of logical qubits $k$ that can be encoded in $\Ceao$ is equal
to $n-s-c-r$. To sum up, $\Ceao$ defined by
$(\cS_{0,I},\cS_{0,E},\cS_{0,G})$ is an $[[n,k;r,c]]$ EAOQECC that
fixes a $2^{k+r}$-dimensional code space, within which $\rho\otimes
\sigma$ and $\rho\otimes \sigma'$ are considered to carry the same
information.

\begin{Prop}
\label{EAOQECC_dumb} The EAOQECC $\Ceao$ defined by
$(\cS_{0,I},\cS_{0,E},\cS_{0,G})$ can correct an error set $\bE_0$
if for all $E_1,E_2\in \bE_0$, $E_2^\dagger E_1 \in
\langle\cS_{0,I},\cS_{0,G}\rangle\bigcup (\cG_n-\cZ(\langle
\cS_{0,I},\cS_{0,E} \rangle))$.
\end{Prop}
\begin{proof}
Since the vector $(\ba,\ba_1,\ba_2,\bb,\bc,\bd)$ completely
specifies the error operator $E$, we consider the following two
different cases:
\begin{itemize}
\item If two error operators $E_1$ and $E_2$ have the same reduced
syndrome $(\ba,\ba_1,\ba_2)$, then the error operator $E_2^\dagger
E_1$ gives us all-zero reduced syndrome with some vector
$(\bb,\bc,\bd)$. Therefore, $E_2^\dagger
E_1\in\langle\cS_{0,I},\cS_{0,G}\rangle$. This error $E_2^\dagger
E_1$ has no effect on the logical state $\proj{\varphi}$.
\item If two error operators $E_1$ and $E_2$ have different reduced
syndromes, and let $(\ba,\ba_1,\ba_2)$ be the reduced syndrome of
$E_2^\dagger E_1$, then $E_2^\dagger E_1 \not\in Z(\langle
\cS_{0,I},\cS_{0,E}\rangle)$. This error $E_2^\dagger E_1$ can be
corrected by the decoding operation given in (\ref{EAOQECC_decode}).
\end{itemize}
\end{proof}

\section*{6.2 \hspace{2pt} The general case}
\addcontentsline{toc}{section}{6.2 \hspace{0.15cm} The general case}
\begin{theorem}
\label{EAOQECC_general} Given the subgroups $(\cS_I,\cS_E,\cS_G)$,
there exists an $[[n,k;r,c]]$ entanglement-assisted operator quantum
error-correcting code $C^{\text{eao}}$ defined by the encoding and
decoding pair: $(\cE,\cD)$. The code $\Ceao$ can correct the error
set $\bE$ if for all $E_1, E_2 \in \bE$, $E_2^\dagger E_1 \in
\langle\cS_I,\cS_G\rangle \bigcup
(\cG_n-\cZ(\langle\cS_I,\cS_E\rangle))$.
\end{theorem}

\begin{proof}
\begin{figure}
\centering
\includegraphics[width=5.8in]{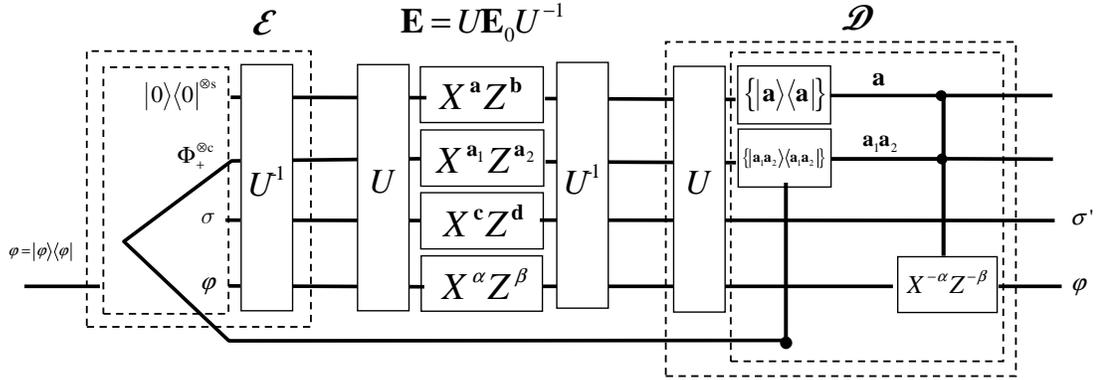}
\caption{The entanglement-assisted operator quantum error-correcting
code.} \label{EAOQECC}
\end{figure}
Since $\cS\sim\cS_0$, there exists an unitary matrix $U$ that
preserves the commutation relations. The protocol is shown in Figure
\ref{EAOQECC}. Define $\cE=U^{-1}\circ\cE_0$ and $\cD=\cD_0\circ U$,
where $\cE_0$ and $\cD_0$ are given in (\ref{EAOQECC_en}) and
(\ref{EAOQECC_decode}), respectivley. Since
\[
\cD_0\circ E_0\circ \cE_{0}=\text{id}^{\otimes k}
\]
for any $E_0\in\bE_0$, then
\[
\cD\circ E\circ\cE=\text{id}^{\otimes k}
\]
follows for any $E\in\bE$. Thus, the encoding and decoding pair
$(\cE,\cD)$ corrects $\bE$.
\end{proof}

We say that the $[[n,k,d;r,c]]$ EAOQECC $C^{\text{eao}}$ has
distance $d$ if it can correct any error set $\bE$ such that for
each operator $E\in \bE$, the weight $t$ of $E$ satisfies $2t+1\leq
d$.

\section*{6.3 \hspace{2pt} Properties of EAOQECCs}
\addcontentsline{toc}{section}{6.3 \hspace{0.15cm} Properties of
EAOQECCs}

In the description earlier in this chapter, we assumed that the
gauge subgroup was generated by a set of symplectic pairs of
generators. In some cases, it may make sense to start with a gauge
subgroup which itself has both an isotropic (i.e., commuting) and a
symplectic subgroup.  In this case, we can arbitrarily add a
symplectic partner for each generator in the isotropic subgroup of
the gauge group.  This can be useful in constructing EAOQECCs from
EAQECCs, in a way analogous to how OQECCs can be constructed by
starting from standard QECCs.  Poulin shows in \cite{DP05} that it
is possible to move generators from the stabilizer group into the
gauge subgroup, together with their symplectic partners, without
changing the essential features of the original code. We provide an
example of such a construction in section 6.4.2.

There is further flexibility in trading between active error
correction ability and passive noise avoiding ability
\cite{AKS06OQECC}. This is captured by the following theorem:

\begin{theorem}
\label{trans} We can transform any $[[n,k+r,d_1;0,c]]$ code $C_1$
into an $[[n,k,d_2;r,c]]$ code $C_2$, and transform the
$[[n,k,d_2;r,c]]$ code $C_2$ into an $[[n,k,d_3;0,c]]$ code $C_3$,
where $d_1\leq d_2 \leq d_3$.
\end{theorem}
\begin{proof}
There exists an isotropic subgroup $\cS_I$ and an entanglement
subgroup $\cS_E$ associated with $C_1$ of size $2^{s}$ and $2^{2c}$,
respectively. These parameters satisfy $s+c+k+r=n$. This code $C_1$
corresponds to an $[[n,k+r,d_1;0,c]]$ EAQECC for some $d_1$. If we
add the gauge subgroup $\cS_G$ of size $2^{2r}$, then
$(\cS_I,\cS_E,\cS_G)$ defines an $[[n,k,d_2;r,c]]$ EAOQECC $C_2$ for
some $d_2$, which follows from Theorem \ref{EAOQECC_general}. Let
$\bE_1$ be the error set that can be corrected by $\cC_1$, and
$\bE_2$ be the error set that can be corrected by $\cC_2$. Clearly,
$\bE_1\subset \bE_2$ (see the following table), so $\cC_2$ can
correct more errors than $\cC_1$. By sacrificing part of the
transmission rate, we have gained additional passive correction, and
$d_2\geq d_1$.

If we now throw away half of each symplectic pair in $\cS_G$ and
include the remaining generators in $\cS_I$, which becomes $\cS_I'$,
the size of the isotropic subgroup increases by a factor of $2^r$.
Then $(\cS_I',\cS_E)$ defines an $[[n,k,d_3;0,c]]$ EAQECC $C_3$. Let
$\bE_3$ be the error set that can be corrected by $C_3$. Let
$E\in\bE_2$, then either $E\in\langle \cS_I,\cS_G\rangle$ or
$E\not\in\cZ(\langle \cS_I,\cS_E\rangle)$.
\begin{itemize}
\item If $E\in\langle \cS_I,\cS_G\rangle$, then either $E\in \cS_I'$
or $E\in\langle \cS_I,\cS_G\rangle/\cS_I'$. If $E\in\langle
\cS_I,\cS_G\rangle/\cS_I'$, this implies $E\not\in\cZ(\cS_I')$.
Thus, $E\in\bE_3$.
\item Since $\langle \cS_I,\cS_E\rangle\subset \langle
\cS_I',\cS_E\rangle$, we have $\cZ(\langle
\cS_I',\cS_E\rangle)\subset\cZ(\langle \cS_I,\cS_E\rangle)$. If
$E\not\in\cZ(\langle \cS_I,\cS_E\rangle)$, then
$E\not\in\cZ(\langle\cS_I',\cS_E \rangle)$. Thus, $E\in\bE_3$.
\end{itemize}
Putting these together we get $ \bE_2\subset \bE_3$. Therefore
$d_3\geq d_2$.
\end{proof}

To conclude this section, we list the different error-correcting
criteria of a conventional stabilizer code (QECC), an EAQECC, an
OQECC, and an EAOQECC in Table \ref{sum_QEC}.
\begin{table}%
\label{sum_QEC}%
\centering
\begin{tabular}{|c|c|}
\hline QECC & EAQECC \\ \hline
$E_2^\dagger E_1\not\in \cZ(\cS_I)$ & $E_2^\dagger E_1\not\in \cZ(\langle\cS_I,\cS_E\rangle)$ \\
$E_2^\dagger E_1\in\cS_I$ & $E_2^\dagger E_1\in\cS_I$ \\ \hline
OQECC & EAOQECC \\ \hline
$E_2^\dagger E_1\not\in \cZ(\cS_I)$ & $E_2^\dagger E_1\not\in \cZ(\langle\cS_I,\cS_E\rangle)$ \\
$E_2^\dagger E_1\in \langle\cS_I,\cS_G\rangle$ & $E_2^\dagger E_1\in\langle\cS_I, \cS_G\rangle$ \\
\hline
\end{tabular}%
\caption{Summary of error-correcting criteria.}
\end{table}%

\section*{6.4 \hspace{2pt} Examples}
\addcontentsline{toc}{section}{6.4 \hspace{0.15cm} Examples}

\subsection*{6.4.1 \hspace{2pt} EAOQECC from EAQECC}%
\addcontentsline{toc}{subsection}{6.4.1 \hspace{0.15cm} EAOQECC from
EAQECC}

\label{DBexample}

Our first example constructs an $[[8,1,3;c=1,r=2]]$ EAOQECC from an
[[8,1,3;1]] EAQECC.  Consider the EAQECC code defined by the group
$\cS$ generated by the operators in Table~\ref{EAQECC1}. Here
$\bar{Z}$ and $\bar{X}$ refer to the logical $Z$ and $X$ operation
on the codeword, respectively. The isotropic subgroup is
$\cS_I=\langle S_1,S_2,S_3,S_4,S_5,S_8\rangle$, the entanglement
subgroup is $\cS_E=\langle S_6,S_7 \rangle$, and together they
generate the full group $\cS=\langle \cS_I,\cS_E\rangle$. This code
$C(\cS_I,\cS_E)$ encodes one qubit into eight physical qubits with
the help of one ebit, and therefore is an $[[8,1;1]]$ code. It can
be easily checked that this code can correct an arbitrary
single-qubit error, and it is degenerate.

\begin{table}[htdp]
\begin{center}
\begin{tabular}{c|cccccccc|c}
  & \multicolumn{8}{c} {Alice} & Bob        \\ \hline\hline
 $S_1$ & Z & Z & I & I & I & I & I & I & I\\
 $S_2$ & Z & I & Z & I & I & I & I & I & I\\
 $S_3$ & I & I & I & Z & Z & I & I & I & I\\
 $S_4$ & I & I & I & Z & I & Z & I & I & I\\
 $S_5$ & I & I & I & I & I & I & Z & Z & I\\
 $S_6$ & I & I & I & I & I & I & I & Z & Z\\
 $S_7$ & X & X & X & I & I & I & X & X & X\\
 $S_8$ & X & X & X & X & X & X & I & I & I\\ \hline
 $\bar{Z}$ & Z & I & I & Z & I & I & I & Z & I\\
 $\bar{X}$ & I & I & I & X & X & X & I & I & I\\
 \hline\hline
\end{tabular}
\end{center}
\caption{The original [[8,1,3;$c=1$]] EAQECC encodes one
qubit into eight physical qubits with the help of one ebit.}%
\label{EAQECC1}
\end{table}%

By inspecting the group structure of $\cS$, we can recombine the
first four stabilizers of the code to give two isotropic generators
(which we retain in $\cS_I$), and two generators which we include,
together with their symplectic partners, in the subgroup $\cS_G$,
for two qubits of gauge symmetry. This yields an $[[8,1,3;c=1,r=2]]$
EAOQECC whose generators are given in Table~\ref{EAOQECC1}. where
$\cS_I=\langle S_1',S_2',S_3',S_6'\rangle$, $\cS_E=\langle
S_4',S_5'\rangle$, and $\cS_G=\langle
g_1^z,g_1^x,g_2^z,g_2^x\rangle$.

\begin{table}[htdp]
\begin{center}
\begin{tabular}{c|cccccccc|c}
  & \multicolumn{8}{c} {Alice} & Bob        \\ \hline\hline
 $S_1'$ & Z & Z & I & Z & Z & I & I & I & I\\
 $S_2'$ & Z & I & Z & Z & I & Z & I & I & I\\
 $S_3'$ & I & I & I & I & I & I & Z & Z & I\\
 $S_4'$ & I & I & I & I & I & I & I & Z & Z\\
 $S_5'$ & X & X & X & I & I & I & X & X & X\\
 $S_6'$ & X & X & X & X & X & X & I & I & I\\ \hline
 $\bar{Z}$ & Z & I & I & Z & I & I & I & Z & I\\
 $\bar{X}$ & I & I & I & X & X & X & I & I & I\\ \hline
 $g_1^z$ & Z & Z & I & I & I & I & I & I & I\\
 $g_1^x$ & I & X & I & I & X & I & I & I & I\\
 $g_2^z$ & I & I & I & Z & I & Z & I & I & I\\
 $g_2^x$ & I & I & X & I & I & X & I & I & I\\
 \hline\hline
\end{tabular}
\end{center}
\caption{The resulting [[8,1,3;$c=1$,$r=2$]] EAOQECC encodes one
qubit into eight physical qubits with the help of one ebit, and
create two gauge qubits for passive error correction.}
\label{EAOQECC1}
\end{table}%

\subsection*{6.4.2 \hspace{2pt} EAOQECCs from classical BCH codes}%
\addcontentsline{toc}{subsection}{6.4.2 \hspace{0.15cm} EAOQECCs
from classical BCH codes}

EAOQECCs can also be constructed directly from classical binary
codes. Before we give examples, however, we need one more theorem:

\begin{theorem}
\label{ebit} Let $H$ be any binary parity check matrix with
dimension $(n-k)\times n$. We can obtain the corresponding
$[[n,2k-n+c;c]]$ EAQECC, where $c = {\rm rank}(H H^T)$ is the number
of ebits needed.
\end{theorem}
\begin{proof}By the CSS construction,
let $\tilde{H}$ be
\begin{equation}%
\label{h4} \tilde{H}=\left(\begin{array}{c|c}
 H & \mathbf{0} \\ \mathbf{0} & H
\end{array}\right).
\end{equation}%
Let $\cS$ be the group generated by $\tilde{H}$, then $\cS=\langle
Z^{\br_1},\cdots, Z^{\br_{n-k}},X^{\br_{1}},\cdots,X^{\br_{n-k}}
\rangle$, where $\br_i$ is the $i$-th row vector of $H$. Now we need
to determine how many symplectic pairs are in group $\cS$. Since
rank$(H H^T)=c$, there exists a matrix $P$ such that
\[
P H H^T P^T=\left(\begin{array}{cccc} I_{p\times p} & \b0  & \b0 &
\b0
\\ \b0  &
  \mathbf{0} & I_{q\times q} & \b0 \\ \b0 &    I_{q\times q} & \mathbf{0}
  & \b0
  \\ \b0 & \b0 & \b0 & \mathbf{0}
\end{array}\right)_{(n-k)\times(n-k)}
\]
where $p+2q=c$. Let $\br_i'$ be the $i$-th row vector of the new
matrix $P H$, then $\cS=\langle Z^{\br_1'},\cdots, Z^{\br_{n-k}'}
,X^{\br_{1}'},\cdots,X^{\br_{n-k}'}\rangle $.

Using the fact that $\{Z^{\ba},X^{\bb}\}=0$ if and only if $\ba\cdot
\bb=1$, we know that the operators $Z^{\br_i'},X^{\br_i'}$ for $1\le
i\le p$, and the operators $Z^{\br_{p+j}'},X^{\br_{p+q+j}'}$ for $1
\le j\le q$, generate a symplectic subgroup in $\cS$ of size
$2^{2c}$.
\end{proof}

\label{BCHexample}
\begin{Def}\cite{FJM77}
A cyclic code of length $n$ over \text{GF}($p^m$) is a \emph{BCH
code of designed distance $d$} if, for some number $b\geq 0,$ the
generator polynomial $g(x)$ is
\[
g(x)=\text{lcm}\{M^{b}(x),M^{b+1}(x),\cdots,M^{b+d-2}(x)\},
\]
where $M^{k}(x)$ is the minimal polynomial of $\alpha^k$ over
GF($p^m$). I.e. $g(x)$ is the lowest degree monic polynomial over
GF($p^m$) having $\alpha^b,\alpha^{b+1},\cdots,\alpha^{b+d-2}$ as
zeros. When $b=1$, we call such BCH codes \emph{narrow-sense} BCH
codes. When $n=p^m-1$, we call such BCH codes \emph{primitive}.
\end{Def}

Consider the primitive narrow-sense BCH code over GF($2^6$). This
code has the following parity check matrix
\begin{equation}%
\label{BCH} H_q=\left(\begin{array}{ccccc}
1 & \alpha & \alpha^2 & \cdots & \alpha^{n-1} \\
1 & \alpha^3 & \alpha^6 & \cdots & \alpha^{3(n-1)} \\
1 & \alpha^5 & \alpha^{10} & \cdots & \alpha^{5(n-1)} \\
1 & \alpha^7 & \alpha^{14} & \cdots & \alpha^{7(n-1)}
\end{array}\right),
\end{equation}%
where $\alpha\in \text{GF}(2^6)$ satisfies $\alpha^6+\alpha+1=0$ and
$n=63$. Since all finite fields of order $p^m$ are
\emph{isomorphic}, there exists a one-to-one correspondence between
elements in $\{\alpha^j:j=0,1,\cdots,p^m-2,\infty\}$ and elements in
$\{a_0,a_1,\cdots,a_m: a_i\in\text{GF}(p)\}$. If we replace
$\alpha^{j}\in\text{GF}(2^6)$ in (\ref{BCH}) with its binary
representation, this gives us a binary $[63,39,9]$ BCH code whose
parity check matrix $H_2$ is of size $24 \times 63$. If we carefully
inspect the binary parity check matrix $H_2$, we will find that the
first 18 rows of $H_2$ give a $[63,45,7]$ dual-containing BCH code.

From Theorem \ref{ebit}, it is easy to check that $c=\text{rank}(H_2
H_2^T)=6$. Thus by the CSS construction \cite{BDH06}, this binary
$[63,39,9]$ BCH code will give us a corresponding $[[63,21,9;6]]$
EAQECC.

If we further explore the group structure of this EAQECC, we will
find that the 6 symplectic pairs that generate the entanglement
subgroup $\cS_E$ come from the last 6 rows of $H_2$.  (Remember that
we are using the CSS construction.) If we remove one symplectic pair
at a time from $\cS_E$ and add it to the gauge subgroup $\cS_G$, we
get EAOQECCs with parameters given in Table~\ref{BCHtable}.

\begin{table}[h]
\begin{center}
\begin{tabular}{|c|c|c|c|c|} \hline
n & k & d & r & c \\ \hline
63 & 21 & 9 & 0 & 6 \\
63 & 21 & 7 & 1 & 5 \\
63 & 21 & 7 & 2 & 4 \\
63 & 21 & 7 & 3 & 3 \\
63 & 21 & 7 & 4 & 2 \\
63 & 21 & 7 & 5 & 1 \\
63 & 21 & 7 & 6 & 0 \\ \hline
\end{tabular}
\end{center}
\caption{Parameters of the EAOQECCs constructed from a classical
[63,39,9] BCH code, where $r$ represents the amount of gauge qubits
created and $c$ represents the amount of ebits needed.}
\label{BCHtable}
\end{table}%

In general, there could be considerable freedom in which of the
symplectic pairs is to be removed. There are plenty of choices in
the generators of $\cS_E$. In fact, it does not matter which
symplectic pair we remove first in this example, due to the
algebraic structure of this BCH code. The distance is always lower
bounded by $7$.

One final remark: this example gives EAOQECCs with positive net
rate, so they could be used as catalytic codes.

\subsection*{6.4.3 \hspace{2pt} EAOQECCs from classical quaternary codes}%
\addcontentsline{toc}{subsection}{6.4.3 \hspace{0.15cm} EAOQECCs
from classical quaternary codes}

\label{qexample}

In the following, we will show how to use MAGMA \cite{MAGMA} to
construct EAOQECCs from classical quaternary codes with positive net
yield and without too much distance degradation. Consider the
following parity check matrix $H_4$ of a $[15,10,4]$ quaternary
code:
\begin{equation} \label{QC}
H_4=\left(\begin{array}{ccccccccccccccc}
1&0&0&0&1&1&\omega^2&0&1&\omega^2&0&\omega&\omega^2&1&0 \\
 0&1&0&0&1&0&\omega&\omega^2&1&\omega&0&0&1&\omega&1 \\
 0&0&1&0&\omega&\omega^2&1&\omega&1&0&0&\omega&1&\omega^2&\omega \\
 0&0&0&1&1&\omega^2&0&1&\omega^2&\omega&0&\omega^2&1&0 &\omega^2 \\
 0&0&0&0&0&0&0&0&0&0&1&0&0&0&0 \\
\end{array}\right),
\end{equation}
where $\{0,1,\omega,\omega^2\}$ are elements of GF(4) that satisfy:
$1+\omega+\omega^2=0$ and $\omega^3=1$.  This quaternary code has
the largest minimum weight among all known $[n=15,k=10]$ linear
quaternary codes. By the construction given in \cite{BDH06}, this
code gives a corresponding $[[15,9,4;c=4]]$ EAQECC with the
stabilizers given in Table~\ref{EAQECC2}.

\begin{table}[htdp]
\begin{center}
\begin{tabular}{|c|ccccccccccccccc|}
\hline\hline \multirow{8}{*}{${\cal{S}}_E$} &I&I&Y&I&Z&X&Y&Z&Y&I&I&Z&Y&X&Z \\
&I&Y&I&I&Y&I&Z&X&Y&Z&I&I&Y&Z&Y \\
&I&Z&Y&I&I&X&Z&X&X&X&I&Z&X&I&I \\
&I&I&X&I&Y&Z&X&Y&X&I&I&Y&X&Z&Y \\
&I&I&I&I&I&I&I&I&I&I&Z&I&I&I&I \\
&I&I&I&I&I&I&I&I&I&I&Y&I&I&I&I \\
&I&Z&Z&Z&X&I&Y&I&Y&I&I&Z&Z&Z&I \\
&I&Y&Y&Y&Z&I&X&I&X&I&I&Y&Y&Y&I \\ \hline
\multirow{2}{*}{${\cal{S}}_I$}&Z&Z&Y&I&Z&Y&X&X&Y&Z&I&Y&Z&Z&I \\
&Y&Y&X&I&Y&X&Z&Z&X&Y&I&X&Y&Y&I\\
 \hline\hline
\end{tabular}
\end{center}
\caption{Stabilizer generators of the $[[15,9,4;c=4]]$ EAQECC
derived from the classical code given by Eq.~(\ref{QC}). The size of
$\cS_E$ is equal to $2^{2c}$.}%
\label{EAQECC2}
\end{table}

The entanglement subgroup $\cS_E$ of this EAQECC has $c=4$
symplectic pairs. Our goal is to construct an EAOQECC from this
EAQECC such that the power of error correction is largely retained,
but the amount of entanglement needed is reduced. In this example,
the choice of which symplectic pair is removed strongly affects the
distance $d$ of the resulting EAOQECC. By using MAGMA to perform a
random search of all the possible sympletic pairs in $\cS_E$, and
then putting them into the gauge subgroup $\cS_G$, we can obtain a
$[[15,9,3;c=3,r=1]]$ EAOQECC with stabilizers given in
Table~\ref{EAOQECC2}. The distance is reduced by one, which still
retains the ability to correct all one-qubit errors; the amount of
entanglement needed is reduced by one ebit; and we gain some extra
power of passive error correction, due to the subsystem structure
inside the code space, given by the gauge subgroup $\cS_G$.
\begin{table}[h]
\begin{center}
\begin{tabular}{|c|ccccccccccccccc|}
\hline\hline \multirow{6}{*}{${\cal{S}}_E$} &I&I&Y&I&Z&X&Y&Z&Y&I&I&Z&Y&X&Z \\
&I&Y&I&I&Y&I&Z&X&Y&Z&I&I&Y&Z&Y \\
&I&Z&Y&I&I&X&Z&X&X&X&I&Z&X&I&I \\
&I&I&X&I&Y&Z&X&Y&X&I&I&Y&X&Z&Y \\
&I&I&I&I&I&I&I&I&I&I&Z&I&I&I&I \\
&I&I&I&I&I&I&I&I&I&I&Y&I&I&I&I \\
\hline \multirow{2}{*}{${\cal{S}}_G$}&I&Z&Z&Z&X&I&Y&I&Y&I&I&Z&Z&Z&I \\
&I&Y&Y&Y&Z&I&X&I&X&I&I&Y&Y&Y&I \\ \hline
\multirow{2}{*}{${\cal{S}}_I$}&X&X&Z&I&X&Z&Y&Y&Z&X&I&Z&X&X&I \\
&Z&Z&Y&I&Z&Y&X&X&Y&Z&I&Y&Z&Z&I\\
\hline\hline
\end{tabular}
\end{center}
\caption{Stabilizer generators of the $[[15,9,3;c=3,r=1]]$ EAOQECC
derived from the EAQECC given by Table~\ref{EAQECC2}. The size of
$\cS_E$ and $\cS_G$ is equal to $2^{2c}$ and $2^{2r}$,
respectively.} \label{EAOQECC2}
\end{table}

\section*{6.5 \hspace{2pt} Discussion}
\addcontentsline{toc}{section}{6.5 \hspace{0.15cm} Discussion}

We have shown a very general quantum error correction scheme that
combines two extensions of standard stabilizer codes. This scheme
includes the advantages of both entanglement-assisted and operator
quantum error correction.

In addition to presenting the formal theory of EAOQECCs, we have
given several examples of code construction. The methods of
constructing OQECCs from standard QECCs can be applied directly to
the construction of EAOQECCs from EAQECCs.  We can also construct
EAOQECCs directly from classical linear codes.

We also show that, by exploring the structure of the symplectic
subgroup, we can construct versatile classes of EAOQECCs with
varying powers of passive versus active error correction.  Starting
with good classical codes, this entanglement-assisted operator
formalism can be used to construct quantum codes tailored to the
needs of
particular applications.  

\chapter*{Chapter 7: \hspace{1pt} Quantum quasi-cyclic low-density parity-check codes}
\label{cp_VII}%
\addcontentsline{toc}{chapter}{Chapter 7:\hspace{0.15cm} Quantum
quasi-cyclic low-density parity-check codes}

\section*{7.1 \hspace{2pt} Classical low-density parity-check codes}
\addcontentsline{toc}{section}{7.1 \hspace{0.15cm} Classical
low-density parity-check codes}

Given a binary parity check matrix $H$, its \emph{density} is
defined to be the ratio of the number of ``1'' entries to the total
number of entries in $H$. When the density is less than
$\frac{1}{2}$, we call such code ``low-density parity-check (LDPC)
code''. LDPC codes were first proposed by Gallager \cite{RG63thesis}
in the early 1960s, and were rediscovered
\cite{MN96,davey98low,mackay99good} in the 90s. It has been shown
that these codes can achieve a remarkable performance that is very
close to the Shannon limit. Sometimes, they perform even better
\cite{ Mackay98turbo} than their main competitors, the Turbo codes.
These two families of codes are called modern codes.

A LDPC code is \emph{regular}, if its parity check matrix $H$ has
fixed weight for columns and rows; otherwise, it is
\emph{irregular}. A $(J,L)$-regular LDPC code is defined to be the
null space of a Boolean parity check matrix $H$ with the following
properties: (1) each column consists of $J$ ``ones'' (each column
has weight $J$); (2) each row consists of $L$ ``ones'' (each row has
weight $L$); (3) both $J$ and $L$ are small compared to the length
of the code $n$ and the number of rows in $H$.

We define a \emph{cycle} in $H$ to be of length $2s$ if there is an
ordered list of $2s$ matrix elements such that: (1) all $2s$
elements of $H$ are equal to 1; (2) successive elements in the list
are obtained by alternately changing the row or column only (i.e.,
two consecutive elements will have either the same row and different
columns, or the same column and different rows); (3) the positions
of all the $2s$ matrix elements are distinct, except the first and
last ones. We call the cycle of the shortest length the \emph{girth}
of the code.

Several methods of constructing good families of regular LDPC codes
have been proposed \cite{mackay99good,KLF01,Fossorier04}. However,
probably the easiest method is based on circulant permutation
matrices \cite{Fossorier04}, which was inspired by Gallager's
original LDPC construction. In the following, we will first review
several relevant properties of binary circulant matrices, and then
show the construction of this type of classical LDPC codes using
circulant matrices.

\subsection*{7.1.1 \hspace{2pt} Properties of binary circulant matrices}
\addcontentsline{toc}{subsection}{7.1.1 \hspace{0.15cm} Properties
of binary circulant matrices}

Let $M$ be an $r \times r$ circulant matrix over $\bbF_2$. We can
uniquely associate with $M$ a polynomial $M(X)$ with coefficients
given by entries of the first row of $M$. If
$\bc=(c_0,c_1,\cdots,c_{r-1})$ is the first row of the circulant
matrix $M$, then
\begin{equation}
\label{polyM} M(X)=c_0+c_1 X+c_2 X^2+\cdots+c_{r-1}X^{r-1}.
\end{equation}
Adding or multiplying two circulant matrices is equivalent to adding
or multiplying their associated polynomials modulo $X^r-1$. We now
give some useful properties of these matrices and polynomials.

\begin{Prop}
\label{iso} The set of binary circulant matrices of size $r \times
r$ forms a ring isomorphic to the ring of polynomials of degree less
than $r$: $\bbF_2[X]/\langle X^r-1\rangle$.
\end{Prop}

\begin{lemma}
\label{rank} Let $M(X)$ be the polynomial associated with the
$r\times r$ binary circulant matrix $M$. If $\gcd(M(X),X^r-1)=K(X)$,
and the degree of K(X) is $k$, then the rank of $M$ is $r-k$.
\end{lemma}
\begin{proof}
Let $L(X)=(X^r-1)/K(X)$, and let $\bb\in(\bbZ_2)^r$ be the
coefficient vector associated with $L(X)$. Since the degree of
$L(X)$ is $r-k$, $b_i=0$ for $i>r-k$. It follows that
\begin{equation}
\label{mod}%
L(X)M(X)=0\ \text{mod}\ (X^r-1).
\end{equation}
If $\br_i$ is the $i$-th row of $M$, then (\ref{mod}) gives the
following $k$ linearly dependent equations:
\begin{equation}
\begin{split}
&b_0\br_0+b_1\br_1+\cdots+ b_{r-k} \br_{r-k} =0  \\
&b_0\br_1+b_1\br_2+\cdots+ b_{r-k} \br_{r-k+1} =0 \\
&\ \ \ \ \ \ \ \vdots \\
&b_0\br_{k-1}+b_1 \br_{k}+ \cdots +b_{r-k}\br_{r-1}=0.
\end{split}
\end{equation}
The set $\{\br_{r-k},\cdots,\br_{r-1}\}$ can therefore be expressed
as linear combinations of $\{\br_0,\cdots,\br_{r-k-1}\}$, and the
rank of $M$ is $r-k$.
\end{proof}

\begin{theorem}
\label{method1} Let $r=p\cdot q$, and let
$\bc=(c_0,c_1,\cdots,c_{r-1})$ be the first row of an $r\times r$
circulant matrix $M$. If $c_{i}$ is $1$ only when $i=0\ mod\ p$,
then rank$(M)=p$.
\end{theorem}
\begin{proof}
Let $M(X)=\sum_{i=0}^{q-1}X^{pi}$ be the polynomial associated with
$M$, with degree $r-p$. Since $M(X)|(X^r-1)$, the degree of
$K(X)=\gcd(M(X),X^r-1)=M(X)$ is also $r-p$. Therefore, by
lemma~\ref{rank}, the rank of $M$ is $p$.
\end{proof}

\begin{theorem}
\label{method2} Let $r=p\cdot q$, and let
$\bc=(c_0,c_1,\cdots,c_{r-1})$ be the first row of an $r\times r$
circulant matrix $M$. If $c_{i}$ is $1$ only when $i < p$, then
rank$(M)=r-p+1$.
\end{theorem}
\begin{proof}
In this case, $M(X)=1+X+\cdots X^{p-1}$ has degree $p-1$. Since
$M(X)|X^r-1$, again by lemma~\ref{rank} the rank of $M$ is $r-p+1$.
\end{proof}

\begin{Cor}
\label{kappa} Let $r=p\cdot q$, and let
$\bc=(c_0,c_1,\cdots,c_{r-1})$ be the first row of an $r\times r$
circulant matrix $M$ such that the weight of $\bc$ is $p$. If
$M(X)|(X^r-1)$, then the rank $\kappa$ of $M$ is lower-bounded by
$r-p+1$.
\end{Cor}
\begin{proof}
Since the weight of $\bc$ is $p$, the lowest possible degree of
$M(X)$ is $p-1$. Then by the method of Theorem~\ref{method2}, the
rank $\kappa$ is at least $r-p+1$.
\end{proof}

\subsection*{7.1.2 \hspace{2pt} Classical quasi-cyclic LDPC codes}
\addcontentsline{toc}{subsection}{7.1.2 \hspace{0.15cm} Classical
quasi-cyclic LDPC codes}

\begin{Def}
A binary linear code $C(H)$ of length $n=r\cdot L$ is called a
quasi-cyclic (QC) code with period $r$ if any codeword which is
cyclically right-shifted by $r$ positions is again a codeword. Such
a code can be represented by a parity-check matrix $H$ consisting of
$r\times r$ blocks, each of which is an (in general different)
$r\times r$ circulant matrix.
\end{Def}

By the isomorphism mentioned in Prop.~\ref{iso}, we can associate
with each quasi-cyclic parity-check matrix $H\in\bbF_2^{Jr\times
Lr}$ a $J\times L$ polynomial parity-check matrix
$\bH(X)=[h_{j,l}(X)]_{j\in[J],l\in[L]}$ where $h_{j,l}(X)$ is the
polynomial, as defined in Eq.~(\ref{polyM}), representing the
$r\times r$ circulant submatrix of $H$, and the notation
$[J]:=\{1,2,\cdots,J\}$.

Generally, there are two ways of constructing  $(J,L)$-regular
QC-LDPC by using circulant matrices \cite{SV04}:
\begin{Def}
We say that a QC-LDPC code is Type-I if it is given by a polynomial
parity-check matrix $\bH(X)$ with all monomials. We say that a
QC-LDPC code is Type-II if it is given by a polynomial parity-check
matrix $\bH(X)$ with either binomials, monomials, or zero.
\end{Def}

\subsubsection*{7.1.2.1 \hspace{2pt} Type-I QC-LDPC}
\addcontentsline{toc}{subsubsection}{8.4.1.2 \hspace{0.15cm} Type-I
QC-LDPC}

To give an example, let $r=16$, $J=3$, and $L=8$.  The following
polynomial parity check matrix
\begin{equation}
\label{typeI}
\bH(X)=\left[\begin{array}{cccccccc}%
X & X & X & X & X & X & X & X \\
X^2 & X^5 & X^{3} & X^{5} & X^2 & X^5 & X^3 & X^5 \\
X^2 & X^3 & X^4 & X^5 & X^6 & X^7 & X^8 & X^9 \end{array}\right]
\end{equation}
gives a Type-I $(3,8)$-regular QC-LDPC code of length
$n=16\cdot8=128$. Later on, we will also express $\bH(X)$ by its
{\it exponent matrix} $H_E$. For example, the exponent matrix of
(\ref{typeI}) is
\begin{equation}
H_E=\left[\begin{array}{cccccccc}%
1 & 1 & 1 & 1 & 1 & 1 & 1 & 1 \\
2 & 5 & 3 & 5 & 2 & 5 & 3 & 5 \\
2 & 3 & 4 & 5 & 6 & 7 & 8 & 9\end{array}\right].
\end{equation}
The difference of arbitrary two rows of the exponent matrix $H_E$ is
defined as
\begin{equation}
\label{diff}%
\bd_{ij}=\bc_i-\bc_j=\left((c_{i,k}-c_{j,k}) \text{mod}\
r\right)_{k\in[L]},
\end{equation}
where $\bc_i$ is the $i$-th row of $H_E$ and $r$ is the size of the
circulant matrix. We then have
\begin{eqnarray*}
\bd_{21} &=& (1,4,2,4,1,4,2,4)\\
\bd_{31} &=& (1,2,3,4,5,6,7,8) \\
\bd_{32} &=& (0,14,1,0,4,2,5,4).
\end{eqnarray*}
We call an integer sequence $\bd=(d_0,d_1,\cdots,d_{L-1})$ {\it
multiplicity even} if each entry appears an even number of times.
For example, $\bd_{21}$ is multiplicity even, but $\bd_{32}$ is not,
since only $0$ and $4$ appear an even number of times. We call $\bd$
{\it multiplicity free} if no entry is repeated; for example,
$\bd_{31}$.

A simple necessary condition for Type-I $(J,L)$-regular QC-LDPC
codes to give girth $g\geq 6$ is given in \cite{Fossorier04}.
However, a stronger result (both sufficient and necessary condition)
is shown in \cite{HI07QLDPC}. We state these theorems  from
\cite{HI07QLDPC} without proof.
\begin{theorem}
\label{type-I1} A Type-I QC-LDPC code $C(H_E)$ is dual-containing if
and only if $\bc_i-\bc_j$ is multiplicity even for all $i$ and $j$,
where $\bc_i$ is the $i$-th row of the exponent matrix $H_E$.
\end{theorem}
\begin{theorem}
\label{type-I3} A necessary and sufficient condition for a Type-I
QC-LDPC code $C(H_E)$ to have girth $g\geq 6$ is $\bc_{i}-\bc_j$ to
be multiplicity free for all $i$ and $j$.
\end{theorem}
\begin{theorem}
\label{type-I2} There is no dual-containing Type-I QC-LDPC having
girth $g\geq 6$.
\end{theorem}

\subsubsection*{7.1.2.2 \hspace{2pt} Type-II QC-LDPC}
\addcontentsline{toc}{subsubsection}{8.4.2.2 \hspace{0.15cm} Type-II
QC-LDPC}

Take $r=16$, $J=3$, and $L=4$. The following is an example of a
Type-II (3,4)-regular QC-LDPC code:
\begin{equation}
\label{typeII}
\bH(X) = \left[\begin{array}{cccc} X+X^4 & 0 & X^7+X^{10} & 0 \\ X^5 & X^6 & X^{11} & X^{12} \\
0 & X^2+X^{9} & 0 & X^7+X^{13} \end{array}\right].
\end{equation}
The exponent matrix of (\ref{typeII}) is
\begin{equation}
H_E = \left[\begin{array}{cccc} (1,4) & \infty & (7,10) & \infty \\
5 & 6 & 11 & 12 \\
\infty & (2,9) & \infty & (7,13) \end{array}\right].
\end{equation}
Here we denote $X^{\infty}=0$.

The difference of two arbitrary rows of $H_E$ is defined similarly
to (\ref{diff}) with the following additional rules: (1) if for some
entry $c_{i,k}$ is $\infty$, then the difference of $c_{i,k}$ and
other arbitrary term is again $\infty$; (2) if the entries $c_{i,k}$
and $c_{j,k}$ are both binomial, then the difference of $c_{i,k}$
and $c_{j,k}$ contains four terms. In this example, we have
\begin{eqnarray*}
\bd_{21} &=& \left((4,1),\infty,(4,1),\infty \right) \\
\bd_{31} &=& \left(\infty,\infty,\infty,\infty \right) \\
\bd_{32} &=& \left(\infty,(12,3),\infty,(11,1)\right) \\
\bd_{11} &=& \left((0,3,13,0),\infty,(0,3,13,0),\infty\right) \\
\bd_{22} &=& \left(0,0,0,0 \right) \\
\bd_{33} &=& \left(\infty,(0,9,7,0),\infty,(0,10,6,0)\right).
\end{eqnarray*}
The definition of {\bf multiplicity even} and {\bf multiplicity free
} is the same except that we do not take $\infty$ into account. For
example, $\bd_{32}$ is multiplicity free, since there is no pair
with the same entry except $\infty$. Unlike Type-I QC-LDPC codes
whose $\bd_{ii}$ is always the zero vector, $\bd_{ii}$ of Type-II
QC-LDPC codes can have non-zero entries. Therefore it is possible to
have cycles of length 4 in a single layer if $\bd_{ii}$ is not
multiplicity free. Each layer is said to be a set of rows of size
$r$ in the original parity check matrix $H$ that corresponds to the
row of $H_E$. For example, $\bd_{11}$ is multiplicity even,
therefore the first layer of this Type-II regular QC-LDPC parity
check matrix contains 4-cycles.

In the following, we will generalize theorems
\ref{type-I1}-\ref{type-I3} given in the previous section to include
the Type-II QC-LDPC case.
\begin{theorem}
\label{type-II-dual} $C(H_E)$ is a dual-containing Type-II regular
QC-LDPC code if and only if $\bc_i-\bc_j$ is multiplicity even for
all $i$ and $j$.
\end{theorem}
\begin{proof}
Let $\bH(X)=[h_{j,l}(X)]_{j\in[J],l\in[L]}$ be the polynomial parity
check matrix associated with a Type-II $(J,L)$-regular QC-LDPC
parity check matrix $H$. Denote the transpose of $\bH(X)$ by
$\bH(X)^T=[h^t_{l,j}(X)]_{l\in[L],j\in[J]}$, and we have
\begin{equation}
\begin{split}
h^{t}_{l,j}(X)=\begin{cases} 0  & \text{if} \ h_{j,l}(X)=0 \\
X^{r-k} & \text{if} \ h_{j,l}(X)=X^k  \\
X^{r-k_1}+X^{r-k_2} & \text{if} \ h_{j,l}(X)=X^{k_1}+X^{k_2}
\end{cases}.
\end{split}
\end{equation}
Let $\hat{\bH}(X)=\bH(X)\bH(X)^T$, and let the $(i,j)$-th component
of $\hat{\bH}(X)$ be $\hat{h}_{i,j}(X)$. Then
\begin{equation}
\label{Hhat}%
\hat{h}_{i,j}(X)=\sum_{l\in[L]} h_{i,l}(X)h^t_{l,j}(X).
\end{equation}
The condition that $\bc_i-\bc_j$ is multiplicity even implies that
$\hat{h}_{i,j}(X)=0$ modulo $X^r-1$, and vice versa.
\end{proof}

\begin{theorem}
\label{type-II-free} A necessary and sufficient condition for a
Type-II regular QC-LDPC code $C(H_E)$ to have girth $g\geq 6$ is
that $\bc_{i}-\bc_j$ be multiplicity free for all $i$ and $j$.
\end{theorem}
\begin{proof}
The condition that $\bc_{i}-\bc_j$ is multiplicity free for all $i$
and $j$ guarantees that there is no 4-cycle between layer $i$ and
layer $j$, and vice versa.
\end{proof}

\begin{theorem}
There is no dual-containing QC-LDPC having girth $g\geq 6$.
\end{theorem}
\begin{proof}
This proof follows directly from theorem \ref{type-II-dual} and
theorem \ref{type-II-free}. If the Type-II regular QC-LDPC code is
dual-containing, then by theorem \ref{type-II-dual}, $\bc_i-\bc_j$
must be multiplicity even for all $i$ and $j$. However, theorem
\ref{type-II-free} says that this QC-LDPC must contain cycles of
length 4.
\end{proof}

\subsection*{7.1.3 \hspace{2pt} Iterative decoding algorithm}
\addcontentsline{toc}{subsection}{7.1.3 \hspace{0.15cm} Iterative
decoding algorithm}

There are various methods for decoding classical LDPC codes
\cite{KLF01}. Among them, \emph{sum-product algorithm} (SPA)
decoding \cite{mackay99good} provides the best trade-off between
error-correction performance and decoding complexity. Before leaving
this section, we will review this SPA decoding procedure for
classical LDPC codes. It turns out that the same SPA decoding
algorithm can be used in the quantum case to decode the error
syndromes effectively.

Let $\bs,\br\in(\bbZ_2)^n$ be the encoded signal and the received
signal, respectively, such that
\begin{eqnarray}%
\inner{H}{\bs}&=&\b0^T, \\
\br&=&\bs+\bn,
\end{eqnarray}
where $\bn\in(\bbZ_2)^n$ is the noise vector introduced by the
binary symmetric channel, and $H$ is the parity check matrix. The
decoder's task is to infer $\bs$ based on the received signal $\br$
and the knowledge of the noise $\bn$. The \emph{optimal decoder},
also known as the maximally likelihood decoder, returns the encoded
signal $\bs$ that maximizes the \emph{posterior} probability
\begin{equation}%
\label{eq_ml}%
P(\bs|\br)=\frac{P(\br|\bs)P(\bs)}{P(\br)}.
\end{equation}
It is known that this optimal decoding is an NP-complete problem
\cite{BMT78}.

If we assume that the \emph{prior} probability of $\bs$ is uniform,
and the noise $\bn$ is independent of $\bs$, then it follows that
estimating the encoded signal $\bs$ is the same as estimating the
noise $\bn$. This is because once $\bn$ is known, then the encoded
signal is
$$\bs=\br+\bn.$$
We can further reduce the decoding problem to the task of finding
the most probable noise vector $\bn$ based on the error syndrome
vector $\bz$ since
\begin{equation}%
\label{eq_er}%
\bz^T=\inner{H}{\bn}=\inner{H}{\br}.
\end{equation}

Next, we will formally introduce the sum-product algorithm, also
known as a ``belief propagation algorithm'' \cite{Pearl88}. Assume
the parity check matrix $H$ is of size $m \times n$. The decoding
problem is to find a noise vector $\bn$ (given that $\bn$ is
independent of $\bs$) satisfying
$$\inner{H}{\bn}=\bz^\T.$$
The elements $\{n_i\}$, $i=1,2,\cdots,n$, are referred as
\emph{bits}, while the elements $\{z_j\}$, $j=1,2,\cdots,m$, are
referred as \emph{checks}. Together $\{n_i\}$ and $\{z_j\}$ form a
\emph{belief network}, and the network of checks and bits are a
\emph{bipartite graph}: bits only connect to checks and vice versa.

The algorithm presented below follows closely from
\cite{mackay99good}. The goal is to compute the marginal posterior
probability $P(n_i|\bz,H)$ for each $i$. Denote the set of bits that
participate in check $j$ by $\fN(j)=\{i: H_{ji}=1\}$. Denote the set
of checks in which bit $i$ participates by $\fM(i)=\{j: H_{ji}=1\}$.
Denote a set $\fN(j)$ with bit $i$ excluded by $\fN(j)\backslash i$.
Define the quantity $q_{ji}^x$ to be the probability that bit $i$ of
$\bn$ has the value $x\in\{0,1\}$, given the probability obtained
via checks other than check $j$, $\{r_{j'i}^x:j'\in\fM(i)\backslash
j\}$. Define the quantity $r_{ji}^x$ to be the probability of check
$j$ being satisfied if bit $i$ of $\bn$ is considered fixed at the
value $x$ and the other bits have a separable distribution given by
the probabilities $\{q_{ji'}:i'\in\fN(j)\backslash i\}$. These two
quantities $q_{ij}$ and $r_{ij}$ associated with each nonzero
element of $H$ are iteratively updated, and would produce the exact
marginal posterior probabilities of all the bits after a fixed
number of iterations if the bipartite graph defined by the matrix
$H$ contained no cycle \cite{Pearl88}. When cycles exist, the
algorithm produces inaccurate probabilities. However, the correct
marginal probabilities are not necessary as long as the decoding is
correct.

\textbf{Initialization.} Denote the prior probability that bit
$n_i=0$ by $p_i^0$, and $p_i^1=1-p_i^0$. Set $p_i^1=f$, where $f$ is
the crossover probability of binary symmetric channel. The variables
$q_{ji}^0$ and $q_{ji}^1$ are initialized to the value $p^0_i$ and
$p^1_i$ when $H_{ji}=1$.

\textbf{Horizontal step.}%
The procedure in the horizontal step of the algorithm is to run
through the checks $j$ and compute for each $i\in\yN(j)$ two
probabilities $r_{ji}^0$ and $r_{ji}^1$, where
\begin{eqnarray}
r_{ji}^0 &=& \sum_{n_{i'}:i'\in\yN(j)\backslash i} \left[
P\left(z_j|n_i=0,\{n_{i'}:i'\in\yN(j)\backslash i\}\right)
\prod_{i'\in\yN(j)\backslash i} q_{ji'}^{n_{i'}}\right], \\
r_{ji}^1 &=& \sum_{n_{i'}:i'\in\yN(j)\backslash i} \left[
P\left(z_j|n_i=1,\{n_{i'}:i'\in\yN(j)\backslash i\}\right)
\prod_{i'\in\yN(j)\backslash i} q_{ji'}^{n_{i'}}\right].
\end{eqnarray}
The quantity $r_{ji}^0$ is the probability of the observed value of
$z_j$ when $n_i$ is assumed to be $0$, given that the other bits
$\{n_{i'}: i' \in \yN(j)\backslash i\}$ have a separable
distribution given by the probabilities $\{q_{ji'}^0,q_{ji'}^1\}$.
The quantity $r^1_{ji}$ is defined similarly except $n_i$ is assumed
to be $1$.

\textbf{Vertical step.}%
The procedure in the vertical step of the algorithm is to take the
computed values of $r_{ji}^0$ and $r_{ji}^1$ and update the values
of the probabilities $q_{ji}^0$ and $q_{ji}^1$ for each j.
\begin{eqnarray}
q_{ji}^0 &=& \alpha_{ji} p^0_i \prod_{j'\in\yM(i)\backslash j}
r_{j'i}^0, \\
q_{ji}^1 &=& \alpha_{ji} p^1_i \prod_{j'\in\yM(i)\backslash j}
r_{j'i}^1,
\end{eqnarray}
where $\alpha_{ji}$ is chosen such that $q_{ji}^0+q_{ji}^1=1$.

\textbf{Decoding}%
The \emph{pseudoposterior} probabilities $q^0_i$ and $q^1_i$ are
calculated after each iteration of the horizontal and vertical
steps, where
\begin{eqnarray}%
q_{i}^0 &=& \alpha_{i} p^0_i \prod_{j\in\yM(i)}
r_{ji}^0, \\
q_{i}^1 &=& \alpha_{i} p^1_i \prod_{j\in\yM(i)} r_{ji}^1.
\end{eqnarray}%
These quantities are used to create a tentative decoding
$\hat{\bn}$. If $q_i^1>0.5$, $\hat{n}_i$ is set to 1. If $\hat{\bn}$
satisfies $\inner{H}{\hat{\bn}}=\bz^\T$, the decoding algorithm
stops. Otherwise, the algorithm repeats from the horizontal step. If
the number of iterations reaches some preset maximum number without
successful decoding, we declare a failure.

It has been shown that the performance of iterative decoding very
much depends on the cycles of shortest length \cite{Tanner81}---in
particular, cycles of length 4. These shortest cycles make
successive decoding iterations highly correlated, and severely limit
the decoding performance. Therefore, to use SPA decoding, it is
important to design codes without short cycles, especially cycles of
length 4.

The sum-product decoding algorithm can be directly applied to the
quantum codes constructed using the (generalized) CSS construction.
This is because the $Z$ errors and $X$ errors of a CSS-type quantum
code can be decoded separately. Therefore, decoding the quantum
errors is equivalent to using the SPA separately for each classical
code in the CSS construction (though this would throw away some
information about the correlations between $X$ errors and $Y$
errors).

\section*{7.2 \hspace{2pt} Quantum low-density parity-check codes}
\addcontentsline{toc}{section}{7.2 \hspace{0.15cm} Quantum
low-density parity-check codes}

The quantum versions of low-density parity-check codes
\cite{HI07QLDPC,MMM04QLDPC,COT05QLDPC,PC08QLDPC} are far less
studied than their classical counterparts. The main obstacle comes
from the dual-containing constraint of the classical codes that are
used to construct the corresponding quantum codes. While this
constraint was not too difficult to satisfy for relatively small
codes, it is a substantial barrier to the use of highly efficient
LDPC codes. However, with the entanglement-assisted formalism, such
constrains can be removed, and constructing quantum LDPC codes from
classical LDPC codes becomes transparent.

The second obstacle to constructing quantum LDPC codes comes from
the bad performance of the efficient decoding algorithm. Though the
SPA can be directly used to decode the quantum errors, the
performance of SPA decoding was severely limited by the many
4-cycles in the standard quantum LDPC codes. We show in this section
that using the entanglement-assisted formalism, we can completely
eliminate all the 4-cycles in the quantum LDPC codes. We will focus
on the quantum LDPC codes constructed from classical quasi-cyclic
LDPC codes, and demonstrate their performance using numerical
methods.

\subsection*{7.2.1 \hspace{2pt} Quantum quasi-cyclic LDPC codes}
\addcontentsline{toc}{subsection}{7.2.1 \hspace{0.15cm} Quantum
quasi-cyclic LDPC codes}

\label{III} It has been shown that any classical linear code can be
used to construct a corresponding entanglement-assisted quantum
error-correcting code.

In the following, we will consider conditions that will give us
$(J,L)$-regular QC-LDPC codes $C(H)$ with girth $g\geq 6$ and with
the rank of $H H^T$ as small as possible. In general, $\hat{\bH}(X)$
represents a square Hermitian matrix $\hat{H}$ with size $Jr\times
Jr$ that contains $J^2$ circulant $r\times r$ matrices represented
by $\hat{h}_{i,j}(X)$ as defined in (\ref{Hhat}). Next, we provide
two examples to illustrate two different ways of minimizing the rank
of the square Hermitian matrix represented by $\hat{\bH}(X)$.

The first method is to make the matrix $\hat{H}=HH^T$ become a
circulant matrix with a small rank. This can be achieved by choosing
$\bH(X)$ such that
\[
\hat{h}_{i,j}(X)=\hat{h}_{i+1,j+1}(X),
\]
for $i,j=0,1,\cdots,J-2.$ The rank $\kappa$ of $\hat{H}$ can then be
read off by lemma~\ref{rank}. If $\gcd(\hat{\bH}(X),X^{Jr}-1)=K(X)$,
and the degree of $K(X)=k$, then $\kappa=Jr-k$. Let's look at an
example of this type using a classical Type-I QC-LDPC code. Take
$r=16$, $J=3$, and $L=8$. The following polynomial parity check
matrix $\bH(X)$ gives the corresponding quantum QC-LDPC code with
length 128:
\begin{equation}
\label{ex1}
\bH(X)=\left[\begin{array}{cccccccc} X & X & X & X & X & X & X & X \\
X & X^2 & X^3 & X^4 & X^5 & X^6 & X^7 & X^8 \\
X & X^3 & X^5 & X^7 & X^9 & X^{11} & X^{13} & X^{15}
\end{array}\right].
\end{equation}
Then
\begin{equation}
\hat{h}_{i,j}(X)=\begin{cases}0, &\text{$i=j$},  \\
\sum_{k=0}^{7}X^k,  & i=j+1 \\ \sum_{k=0}^{7}X^{2k}, & i=j+2
\end{cases}
\end{equation}
It can be easily verified that $\hat{\bH}(X)$ represents a circulant
matrix, and the polynomial associated with $\hat{H}$ is
\[
\hat{\bH}(X) = X^{16}\left(\sum_{k=0}^{7}X^k\right) +
X^{32}\left(\sum_{k=0}^{7}X^{2k}\right).
\]
The degree of $\gcd(\hat{\bH}(X),X^{48}-1)=30$, therefore by
lemma~\ref{rank}, the number of ebits that were needed to construct
the corresponding quantum code is only 18. Actually, (\ref{ex1})
gives us a $[[128,48,6;18]]$ EAQECC, and we will refer to this
example as ``ex1'' later in section 7.3.

The second method is to minimize the rank of each circulant matrix
inside $\hat{H}$. Let the rank of the circulant matrix represented
by $\hat{h}_{i,j}(X)$ be $\kappa_{i,j}$. Let the rank of $\hat{H}$
be $\kappa$. Then
\begin{equation}
\label{u_rank}%
\kappa \leq \sum_{i=1}^{J} \max_{j\in[J]}{\kappa_{i,j}}.
\end{equation}
This upper bound is not tight for Type-I $(J,L)$-regular QC-LDPC
codes when $L$ is odd. This is because $\kappa_{i,i}=r$ for every
$i$. When $L$ is even, we have $\kappa_{i,i}=0$ for every $i$. We
can obtain a tighter upper bound for $\kappa$ by carefully choosing
the exponents of $\bH(X)$ such that the degree of
$\gcd(\hat{h}_{i,j}(X),X^r-1)$ is as large as possible for every $i$
and $j$.
\begin{theorem}
\label{bound} Given a Type-I $(J,L)$-regular QC-LDPC code with
$\bH(X)$, if $L$ is even and $\gcd(\hat{h}_{i,j}(X),X^r-1)>1$ for
$i\neq j$, then the rank $\kappa$ is upper bounded by $J(r-L+1)$.
\end{theorem}
\begin{proof}
Let $\hat{h}_{i,j}$ be the circulant matrix associated with the
polynomial $\hat{h}_{i,j}(X)$, then the weight of the coefficient
vector of $\hat{h}_{i,j}$ is $L$. By Corollary~\ref{kappa},
$\kappa_{i,j} \leq r-L+1$. Therefore
\[
\kappa\leq \sum_{i=1}^{J} \max_{j \in[J]} \kappa_{i,j} \leq
J(r-L+1).
\]
\end{proof}

Our second example comes from a classical Type-II QC-LDPC code.
Again take $r=16$, $J=3$, and $L=8$. The following polynomial parity
check matrix $\bH(X)$ gives the corresponding quantum QC-LDPC code
with length 128:
\begin{equation}
\label{ex2}%
\bH(X) = \left[\begin{array}{cccccccc} X+X^2 & 0 & X+X^4 & 0 & X+X^6 & 0 & X+X^8 & 0 \\
X^5 & X^5 & X^6 & X^6 & X^7 & X^7 & X^8 & X^8 \\
0 & X+X^2 & 0 & X+X^4 & 0 & X+X^6 & 0 & X+X^8 \end{array}\right].
\end{equation}
Then
\begin{equation}
\hat{h}_{i,j}(X) = \begin{cases}0, &(i,j)=(2,2),(1,3),\text{or} (3,1) \\
\sum_{k=0}^{7}X^{1+2k}, &(i,j)=(1,1),(3,3) \\
\sum_{k=0}^{7}X^k, &(i,j)=(2,1),(2,3)
\end{cases}
\end{equation}
In this example, each layer of the matrix $\hat{\bH}(X)$ has rank
less than $9$. Actually, (\ref{ex2}) gives a $[[128,48,6;18]]$
quantum QC-LDPC code, and we will refer to this example as ``ex2''
in section 7.3.

\section*{7.3 \hspace{2pt} Performance}%
\label{IV}%
\addcontentsline{toc}{section}{7.3 \hspace{0.15cm} Performance}

In this section, we compare the performance of the QLDPC codes given
in Sec.~7.2 to conventional (dual-containing) QLDPC codes that have
been derived in the existing literature.  The easiest way of
constructing a QLDPC is the following technique, proposed by MacKay
et al. in \cite{MMM04QLDPC}.  Take an $n/2 \times n/2$ cyclic matrix
$C$ with row weight $L/2$, and define
\[
H_0=[C,C^T].
\]
Then we delete some rows from $H_0$ to obtain a matrix $H$ with $m$
rows. It is easy to verify that $H$ is dual-containing. Therefore by
the CSS construction, we can obtain conventional QLDPC codes of
length $n$. The advantage of this construction is that the choice of
$n,m$, and $L$ is completely flexible; however, the column weight
$J$ is not fixed. We picked $n=128$, $m=48$, and $L=8$, and called
this quantum LDPC code ``ex-MacKay.''

The second example of constructing a conventional QLDPC is described
in the following theorem \cite{HI07QLDPC}:
\begin{theorem}
\label{HIcont}%
Let $P$ be an integer which is greater than 2 and $\sigma$ an
element of $\bbZ_P^*:=\{z:z^{-1} \text{exists}\}$ with
$ord(\sigma)\neq |\bbZ_P^*|$, where
$ord(\sigma):=\min\{m>0|\sigma^m=1\}$ and $|X|$ means the
cardinality of a set $X$. If we pick any $\tau\in\bbZ_P^*=
\{1,\sigma,\sigma^2,\cdots\}$, define
\begin{eqnarray*}
c_{j,l}&:=& \begin{cases} \sigma^{-j+l} & 0\leq l < L/2
\\-\tau\sigma^{j-1+l} & L/2\leq l < L \end{cases} \\
d_{k,l}&:=& \begin{cases} \tau\sigma^{-k-1+l} & 0\leq l<L/2 \\
-\sigma^{k+l} &L/2\leq l < L \end{cases},
\end{eqnarray*}
and define the exponent matrix $H_C$ and $H_D$ as
\[
H_C=[c_{j,l}]_{j\in[J],l\in[L]},\ \ H_D=[d_{k,l}]_{k\in[K],l\in[L]},
\]
where $L/2=ord(\sigma)$ and $1\leq J,K \leq L/2$, then $H_C$ and
$H_D$ can be used to construct quantum QC-LDPC codes with girth at
least 6.
\end{theorem}
Here, we pick the set of parameters $(J,L,P,\sigma,\tau)$ to be
$(3,8,15,2,3)$. The exponent matrices $H_C$ and $H_D$ described in
theorem \ref{HIcont} are
\begin{eqnarray}
\label{HcHd}
H_C &=& \left[\begin{array}{cccccccc}%
1 & 2 & 4 & 8 & 6 & 12 & 9 & 3 \\
8 & 1 & 2 & 4 & 12 & 9 & 3 & 6 \\
4 & 8 & 1 & 2 & 9 & 3& 6 & 12
\end{array}\right] \\
H_D &=& \left[\begin{array}{cccccccc}%
9 & 3 & 6 & 12 & 14 & 13 & 11 & 7 \\
12& 9 & 3 & 6 & 13 & 11 & 7 & 14 \\
6 & 12 & 9 & 3 & 11 & 7 & 14 & 13
\end{array}\right],
\end{eqnarray}
and by the CSS construction, it will give a $[[120,38,4]]$ quantum
QC-LDPC code. We will call this code ``ex-HI''.

\begin{figure}[htbp]
    \includegraphics[width=0.9\textwidth]{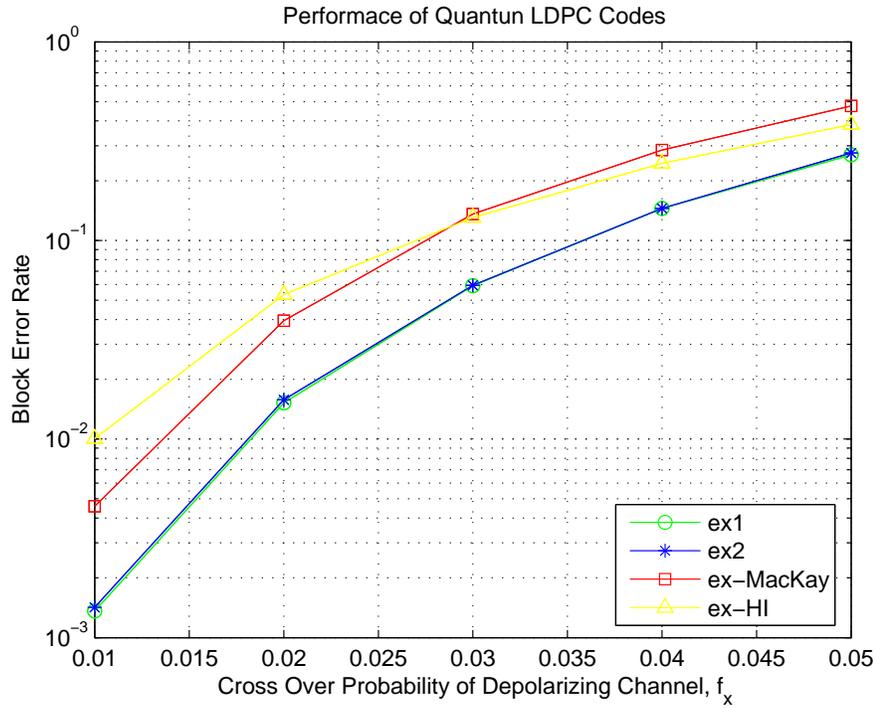}
  \caption{Performance of QLDPC with SPA decoding, and 100-iteration}
  \label{fig}
\end{figure}
%

We compare the performance of our examples in section 7.2.1 with
these two dual-containing quantum LDPC codes in figure \ref{fig}. In
the simulation, we assume the depolarizing channel and use of
sum-product decoding algorithm. The performances of ex1 and ex2 do
not differ much. This is not surprising, since these two codes have
similar parameters. The reason that the performance of ex-MacKay is
worse than our two examples is because there are so many 4-cycles in
ex-MacKay. These cycles impair the decoding performance of
sum-product algorithm. Our entanglement-assisted quantum QC-LDPC
codes also outperform the quantum QC-LDPC code of ex-HI, since the
classical QC-LDPC codes used to construct our examples have better
distance properties than the classical QC-LDPC of ex-HI. This
simulation result is also consistent with our result in
\cite{BDH06}: better classical codes give better quantum codes. Even
though the parameters are not exactly the same, our codes have
higher rate than the code rate of ex-HI.

It is not difficult to verify that the girth of ex1 is 6, and the
girth of ex2 is 8. We numerically investigated the performance of
these two examples with various numbers of iterations. According to
our simulation results, the performance of ex1 and ex2 is almost the
same. The result agrees with the classical result in
\cite{Fossorier04} showing that the increase of girth from 6 to 8 is
not of great help. The result is quite interesting since it implies
that we do not need to worry about constructing QLDPC with higher
girth.

\section*{7.4 \hspace{2pt} Conclusions}\label{V}%
\addcontentsline{toc}{section}{7.4 \hspace{0.15cm} Conclusions}

There are two advantages of Type-II QC-LDPCs over Type-I QC-LDPCs.
First, according to \cite{SV04} certain configurations of Type-II
QC-LDPC codes have larger minimum distance than Type-I QC-LDPC.
Therefore, we can construct better quantum QC-LDPCs from classical
Type-II QC-LDPC codes. Second, it seems likely that Type-II QC-LDPCs
will have more flexibility in constructing quantum QC-LDPC codes
with small amount of pre-shared entanglement, because of the ability
to insert zero submatrices. However, further investigation of this
issue is required.

By using the entanglement-assisted error correction formalism, it is
possible to construct EAQECCs from any classical linear code.  We
have shown how to do this for two classes of quasi-cyclic LDPC codes
(Type-I and Type-II), and proven a number of theorems that make it
possible to bound how much entanglement is required to send a code
block for codes of these types. Using these results, we have been
able to easily construct examples of quantum QC-LDPC codes that
require only a relatively small amount of initial shared
entanglement, and that perform better than previously constructed
dual-containing QLDPCs. Since in general the performance of quantum
codes follows directly from the performance of the classical codes
used to construct them, and the evidence of our examples suggests
that the iterative decoders can also be made to work effectively on
the quantum versions of these codes, this should make possible the
construction of large-scale efficient quantum codes.


\bibliographystyle{plain}
\bibliography{../../Ref}

\end{document}